\documentclass[12pt,onecolumn,draftclsnofoot,journal,twoside]{IEEEtran}

\usepackage{amsmath,dsfont,bbm,epsfig,amssymb,amsfonts,amstext,verbatim,amsopn,cite,subfigure}
\usepackage{multirow,multicol,lipsum,xfrac}
\usepackage{amsthm}
\usepackage{mathtools,amsthm}
\usepackage{url}
\usepackage{amsfonts}
\usepackage{epsfig}
\usepackage[font=small]{caption}
\usepackage{psfrag}	
\usepackage[Algorithm,ruled]{algorithm}
\usepackage{algorithmicx}
\usepackage{algpseudocode}
\usepackage{pifont}
\usepackage[utf8]{inputenc}
\usepackage[T1]{fontenc}  
\usepackage[nolist]{acronym}
\usepackage{paralist}
\usepackage{enumitem}
\usepackage{bbm}
\usepackage[process=auto]{pstool}
\usepackage{float}
\usepackage[colorlinks=false,linkcolor=blue,urlcolor=black,bookmarksopen=true]{hyperref}
\usepackage{bookmark}

\usepackage{tikz,pgfplots}
\usetikzlibrary{shapes.misc,arrows,positioning}

\usepackage[utf8]{inputenc} 
\usepackage{amsmath,dsfont,bbm,epsfig,amssymb,amsfonts,amstext,verbatim,amsopn,cite,subfigure}
\usepackage{multirow,multicol,lipsum,xfrac}
\usepackage{amsthm}
\usepackage{mathtools,amsthm}
\usepackage{url}
\usepackage{amsfonts}
\usepackage{tikz,pgfplots}
\usepackage[nolist]{acronym}

\newcommand{\setR}{\mathbbmss{R}}

\newcommand{\setS}{\mathbbmss{S}}

\newcommand{\setJ}{\mathbbmss{J}}

\newcommand{\setC}{\mathbbmss{C}}

\newcommand{\setI}{\mathbbmss{I}}
\newcommand{\setF}{\mathbbmss{F}}

\newcommand{\Ex}[2]{ \mathbbm{E}_{#2} \left\lbrace #1 \right\rbrace }

\newcommand{\rmj}{\mathrm{j}}

\newcommand{\rmp}{\mathrm{p}}

\newcommand{\argmin}{\mathop{\mathrm{argmin}}}

\newcommand{\mah}{\mathcal{H}}
\newcommand{\maf}{\mathcal{F}}

\newcommand{\mae}{\mathcal{E}}
\newcommand{\maz}{\mathcal{Z}}
\newcommand{\mas}{\mathcal{S}}
\newcommand{\mav}{\mathcal{V}}

\newcommand{\man}{\mathcal{N}}
\newcommand{\mam}{\mathcal{M}}

\newcommand{\mao}{\mathcal{O}}
\newcommand{\maq}{\mathcal{Q}}

\newcommand{\mai}{\mathcal{I}}

\newcommand{\mac}{\mathcal{C}}

\newcommand{\mal}{\mathcal{L}}

\newcommand{\bxx}{\mathbf{x}}

\newcommand{\bww}{\mathbf{w}}
\newcommand{\buu}{\mathbf{u}}
\newcommand{\bvv}{\mathbf{v}}
\newcommand{\byy}{\mathbf{y}}

\newcommand{\mPi}{\boldsymbol{\Pi}}

\newcommand{\rmw}{\mathrm{w}}

\newcommand{\bx}{{\boldsymbol{x}}}

\newcommand{\vv}{\mathrm{v}}

\newcommand{\xx}{\mathrm{x}}
\newcommand{\yy}{\mathrm{y}}
\newcommand{\zz}{\mathrm{z}}

\newcommand{\set}[1]{\left\lbrace#1\right\rbrace}

\newcommand{\brc}[1]{\left( #1 \right) }
\newcommand{\inner}[1]{\left\langle #1 \right\rangle }

\newcommand{\dbc}[1]{\left[ #1 \right] }

\newcommand{\bs}{{\boldsymbol{s}}}
\newcommand{\br}{{\mathbf{r}}}
\newcommand{\bu}{{\mathbf{u}}}
\newcommand{\bv}{{\mathbf{v}}}

\newcommand{\bzz}{{\mathbf{z}}}
\newcommand{\baa}{{\mathbf{a}}}

\newcommand{\dif}{\mathrm{d}}

\newcommand{\by}{{\boldsymbol{y}}}

\newcommand{\trp}{\mathsf{T}}

\newcommand{\mA}{\mathbf{A}}

\newcommand{\mI}{\mathbf{I}}
\newcommand{\mone}{\mathbf{1}}
\newcommand{\mJ}{\mathbf{J}}
\newcommand{\mG}{\mathbf{G}}
\newcommand{\mQ}{\mathbf{Q}}

\newcommand{\mU}{\mathbf{U}}

\newcommand{\mX}{\mathbf{X}}
\newcommand{\mV}{\mathbf{V}}

\newcommand{\mL}{\mathbf{L}}
\newcommand{\mC}{\mathbf{C}}
\newcommand{\mY}{\mathbf{Y}}

\newcommand{\md}{\mathrm{D}}

\newcommand{\norm}[1]{\lVert #1 \rVert}

\newcommand{\abs}[1]{\lvert #1 \rvert}
\newcommand{\tr}[1]{\mathrm{tr} \{ #1 \}}

\newtheoremstyle{mystyle}
{}
{}
{\it}
{}
{\bfseries}
{:}
{ }
{}

\theoremstyle{mystyle}

\newtheorem{definition}{Definition}
\newtheorem{proposition}{Proposition}
\newtheorem{remark}{Remark}

\newtheorem{conjecture}{Conjecture}

\newtheorem{example}{Example}[section]
\newtheorem{corollary}{Corollary}
\newtheorem{result}{Result}
\newtheorem{lemma}{Lemma}
\newtheorem{conclusion}{Conclusion}
\newcounter{bar}

\begin{document}

\begin{acronym}
	\acro{mimo}[MIMO]{multiple-input multiple-output}
	\acro{csi}[CSI]{channel state information}
	\acro{awgn}[AWGN]{additive white Gaussian noise}
	\acro{iid}[i.i.d.]{independent and identically distributed}
	\acro{uts}[UTs]{user terminals}
	\acro{bs}[BS]{base station}
	\acro{tas}[TAS]{transmit antenna selection}
	\acro{glse}[GLSE]{generalized least square error}
	\acro{rhs}[r.h.s.]{right hand side}
	\acro{lhs}[l.h.s.]{left hand side}
	\acro{wrt}[w.r.t.]{with respect to}
	\acro{rs}[RS]{replica symmetry}
	\acro{rsb}[RSB]{replica symmetry breaking}
	\acro{np}[NP]{non-deterministic polynomial-time}
	\acro{papr}[PAPR]{peak-to-average power ratio}
	\acro{rzf}[RZF]{regularized zero forcing}
	\acro{snr}[SNR]{signal-to-noise ratio}
	\acro{rem}[REM]{random energy model}
	\acro{mf}[MF]{matched filtering}
	\acro{gamp}[GAMP]{generalized AMP}
	\acro{amp}[AMP]{approximate message passing}
	\acro{vamp}[VAMP]{vector AMP}
	\acro{map}[MAP]{maximum-a-posterior}
	\acro{ml}[ML]{maximum likelihood}
	\acro{mmse}[MMSE]{minimum mean squared error}
	\acro{ap}[AP]{average power}
	\acro{ldgm}[LDGM]{low density generator matrix}
	\acro{tdd}[TDD]{time division duplexing}
	\acro{rss}[RSS]{residual sum of squares}
	\acro{rls}[RLS]{regularized least-squares}
	\acro{ls}[LS]{least-squares}
	\acro{erp}[ERP]{encryption redundancy parameter}
	\acro{ra}[RA]{reflect-array}
	\acro{ta}[TA]{transmit-array}
	\acro{had}[HAD]{hybrid analog-digital}
\end{acronym}

\title{Bayesian Inference with Nonlinear Generative Models: Comments on Secure Learning}

\author{
\IEEEauthorblockN{
Ali Bereyhi,
Bruno Loureiro,
Florent Krzakala,
Ralf R. M\"uller,\\ and
Hermann Schulz-Baldes
\thanks{This work has been presented in part in the 2022 IEEE International Symposium on Information Theory (ISIT) in Espoo, Finland \cite{bereyhi2022isit}.}
\thanks{Ali Bereyhi and Ralf R. M\"uller are with the Institute for Digital Communications (IDC) at Friedrich-Alexander Universit\"at (FAU) Erlangen-N\"urnberg; \textit{Emails:} \texttt{\{ali.bereyhi,ralf.r.mueller\}@fau.de}. Bruno Loureiro and Florent Krzakala are with the Information, Learning and Physics Lab (IdePHICS) at \'Ecole Polytechnique F\'ed\'erale de Lausanne (EPFL); \textit{Emails}: \texttt{\{bruno.loureiro,florent.krzakala\}@epfl.ch}. Hemann Schulz-Baldes is with the Department of Mathematics at FAU Erlangen-N\"urnberg; \textit{Email:}: \texttt{schuba@mi.uni-erlangen.de}.}
}
\thanks{This work was supported by the Emerging Talents Initiative (ETI) at  FAU Erlangen-N\"urnberg.}
}


\IEEEoverridecommandlockouts

\maketitle

\begin{abstract}
Unlike the classical linear model, nonlinear generative models have been addressed sparsely in the literature of statistical learning. This work aims to bringing attention to these models and their secrecy potential. To this end, we invoke the replica method to derive the asymptotic normalized cross entropy in an inverse probability problem whose generative model is described by a Gaussian random field with a generic covariance function. Our derivations further demonstrate the asymptotic statistical decoupling of the Bayesian estimator and specify the decoupled setting for a given nonlinear model.

The replica solution depicts that \textit{strictly nonlinear} models establish an \textit{all-or-nothing} phase transition: There exists a critical load at which the optimal Bayesian inference changes from \textit{perfect} to an \textit{uncorrelated} learning. Based on this finding, we design a new secure coding scheme which achieves the secrecy capacity of the wiretap channel. This interesting result implies that \textit{strictly nonlinear} generative models are perfectly secured without any secure coding. We justify this latter statement through the analysis of an illustrative model for \textit{perfectly secure and reliable} inference.
\end{abstract}

\begin{IEEEkeywords}
Nonlinear generative models, Bayesian inference, random Gaussian fields, information-theoretically secure learning, replica method, decoupling principle.
\end{IEEEkeywords}

\IEEEpeerreviewmaketitle

\section{Introduction}
Inference over linear models is a widely-investigated problem in the literature of statistical learning \cite{bereyhi2020thesis,bereyhi4,bereyhi1Extension,barbier2016mutual,barbier2018mutual,barbier2020mutual,barbier2021performance,tramel2014statistical}. The popularity of linear models comes from a simple fact: They have a large scope of applications from communications and information theory \cite{tanaka2002statistical,guo2005randomly,bereyhi3,bereyhi10,tulino2013support,barbier2014replica,bereyhi2021detection} to signal processing and machine learning \cite{bereyhi7,barbier2021statistical,maillard2021perturbative,yang2014compressive,lelarge2017fundamental,lelarge2019asymptotic,mignacco2020role,thrampoulidis2020theoretical}. Though linear models have been extensively studied, the problem of inference from observations given by a nonlinear model has been touched very sparsely \cite{sourlas1989spin,kabashima1999statistical,fyodorov2019spin}. This work bridges this gap by demonstrating the potential gains which can be achieved by nonlinear models.

Unlike in the statistical learning literature, nonlinear models are well known in physics. Indeed, various spin glass models from statistical mechanics are described by nonlinear fields \cite{edwards1975theory,mezard1987spin,panchenko2013sherrington,talagrand2006free}. Following the close bonds between the applied concepts in information theory, such as channel and source coding, and theoretical concepts in statistical mechanics, such as Gibbs' inequality, \cite{jaynes1957information,jaynes1957information2,nishimori2001statistical,merhav2010statistical}, it is natural to ask how nonlinear models of spin glasses are connected to information-theoretic concepts. This connection has been studied only sparsely in the literature of statistical learning. The most seminal line of work in this direction was considered by Sourlas in \cite{sourlas1989spin}, where it is shown that by designing an error-correcting code based on a nonlinear spin glass model, one can achieve an \textit{excellent} performance which, under some circumstances, approaches Shannon's limit\footnote{Sourlas' seminal result is discussed in more detail later in this manuscript.}. This seminal work opened the door to a key application of nonlinear models in statistical learning and led to a fruitful synergy between concepts in information theory and statistical physics; see for instance \cite{sourlas1994spin,kabashima1999statistical}. Nevertheless, the development of these study lines into a general framework for learning from a nonlinear generative model has been left unaddressed in the literature of statistical learning.


In this work, we study the problem of inferring model parameters from observations given by a Gaussian nonlinear generative model. Our derivations reveal an intuitive result: while Bayesian inference from linear observations always contain information about the ground truth, the output of a nonlinear model can become completely uncorrelated to the true model parameters. This behavior is analytically described through asymptotic characterization of the joint statistics of the optimal Bayesian estimate and the ground truth via the replica method. Our derivations show that the overlap between the optimal Bayesian inference and the true model parameters exhibits a first-order phase transition from \textit{perfect learning} to \textit{uncorrelated learning}, when the generative model is \textit{strictly} nonlinear. In the light of this fundamental finding, we illustrate the application of nonlinear models to the problem of secure coding. It is shown that using a strictly nonlinear generative model as an encoder, the secrecy capacity of the wiretap channel \cite{wyner1975wire} is asymptotically achieved.

Our findings on the bonds between nonlinear models and the problem of secure transmission further answers the question concluded by Fyodorov in his recent study \cite{fyodorov2019spin}. In this work, a nonlinear Gaussian field is used to encrypt data for transmission over an \ac{awgn} channel. The results of this work lead to an interesting conclusion: Using strictly nonlinear fields, the average error of recovery shows a second-order phase transition with respect to the noise variance. Following this finding, Fyodorov states this question: can a nonlinear encryption model lead to secure transmission with a publicly known codebook\footnote{This questions is stated by Fyodorov through an intuitive short discussion on the main results in \cite[Secion 1.2.1]{fyodorov2019spin}.}? We answer to this question by linking Fyodorov's encryption with the seminal work of Sourlas and establishing a coding scheme for secure transmission over a wiretap channel\footnote{We show this connection in Section~\ref{sec:Secrecy2}.}. Our derivations demonstrate that although non-zero secrecy rate is achievable by combining Fyodorov's framework and that of Sourlas, it does not achieve the secrecy capacity. This is in contrast to our proposed secure coding scheme which achieves the secrecy capacity of the wiretap channel. 


\subsection{Contributions and Tools}
We study Bayesian inference over a nonlinear generative model which is described by a Gaussian random field. 
We investigate the inverse probability problem for a generally mismatched Bayesian inference algorithm. For the considered generic model, we derive the normalized cross entropy as a generic performance metric, from which various information metrics can be derived as special cases, e.g., information rate. We further show the decoupling principle for this model and derive explicitly the decoupled scalar setting. This latter result enables us to evaluate a larger class of asymptotic metrics, such as a generic average input-output distortion and the average cross correlation, known as the \textit{overlap}. Our derivations lead to the following key findings:
\begin{itemize}
	\item With linear models, inference via the optimal Bayesian algorithm \textit{always} contains \textit{non-zero} correlation with the ground truth; however, the achievable information rate never touches Shannon's limit.
	\item Strictly nonlinear models achieve the maximal information rate given by Shannon's limit; however, they show a first-order \textit{phase transition} with respect to the system load at which the learned parameters become \textit{uncorrelated} with the ground truth.
\end{itemize}

The above findings lead us to an interesting result: To achieve the secure capacity of a wiretap channel, one can directly encode the information bits by passing them through a strictly nonlinear Gaussian random field. This finding not only re-discovers Wyner's result on the wiretap channel, but also suggests a new secure coding scheme which achieves the secure capacity of the wiretap channel. An intuitive conclusion from this result is that strictly nonlinear generative models are secure by nature. We give a heuristic proof to this fact by considering the problem of \textit{perfectly secure and reliable} inference in a wiretap setting. For strictly nonlinear models, the maximum number of securely inferred labels per dimension meets the information-theoretic limit given by Wyner in \cite{wyner1975wire}. 
We conclude this work by discussing some new directions for further investigations.

The key tool for large-system analysis in this work is the replica method from the theory of spin glasses and commonly utilized as an analytical tool in the context of information theory and signal processing; see for example \cite{tanaka2002statistical,guo2005randomly,mezard2009information,rangan2012asymptotic,bereyhi3,zaidel2012vector}. A comprehensive introduction to the replica method and its application to the analysis of linear models can be found in \cite{bereyhi2020thesis}. The derivations in this work consider the \textit{replica symmetric} ansatz. This follows from the fact that for Bayes-optimal inference the corresponding spin glass exhibits replica symmetry\footnote{This property is discussed in detail later on.} \cite{zdeborova2016statistical}. However, deviation from the indicated conditions can lead to the so-called \textit{replica symmetry breaking}. More detailed discussions are given in Section~\ref{Sec:Spin} and Appendix~\ref{app:2}.

\subsection{Notation}
Scalars, vectors and matrices are represented with non-bold, bold lower-case and bold upper-case letters, respectively. The transposed of $\mA$ is indicated by $\mA^{\trp}$, $\mI_N$ is an $N\times N$ identity matrix, and $\mone_N$ is an $N\times N$ all-one matrix. The Euclidean norm of $\bx$ is denoted by $\norm{\bx}$. For two $K$-dimensional vectors $\bxx$ and $\byy$, we define the \textit{normalized inner-product} as
\begin{align}
	\inner{\bxx;\byy} = \frac{\bxx^\trp \byy}{K}. \label{eq:inner}
\end{align}
We extend this notation to matrix arguments with valid dimensions as
\begin{align}
	\inner{\mX;\mY} = \frac{\mX^\trp \mY}{K}, \label{eq:innerM}
\end{align}
where $K$ denotes the number of rows  in $\mX$ and $\mY$. The notation $\log\brc{\cdot}$ indicates the \textit{natural} logarithm. We use $\Ex{}{X}$ to refer to expectation with respect to $X$ using the true distribution of $X$. The notation $\Ex{}{q_X}$ indicates expectation with respect to $X$ when the distribution of $X$ is replaced by $q_X$.  By using $\Ex{}{}$, we denote expectation with respect to all random variables involved in the argument. The entropy of a discrete random variable $X$ and the differential entropy of a continuous random variable $Y$ are shown by $H\brc{X}$ and $h\brc{Y}$, respectively. The notation $\mathcal{N}\brc{\boldsymbol{\eta},\mC}$ denotes the multivariate Gaussian distribution with mean vector $\boldsymbol{\eta}$ and covariance matrix $\mC$. To refer to the real axis, we use the notation $\setR$. We further denote the unit $\brc{D-1}$-sphere by $\setS^{D-1}$ , i.e., 
\begin{align}
\setS^{D-1} = \set{	\bxx\in\setR^D: \norm{\bxx} = 1}.
\end{align}
For brevity, $\set{1,\ldots,N}$ and $\set{N_0,\ldots,N_1}$ are abbreviated as  $\dbc{N}$ and $\dbc{N_0:N_1}$, respectively.

\section{Statement of the Problem}
Consider the classical regression problem: An \textit{unknown} sequence of $K$ symbols
\begin{align}
	\bs = \dbc{s_1,\ldots,s_K}
\end{align}
is projected onto $\setR^N$ via a randomly generated but \textit{known} function $\mav\brc{\cdot}: \setR^K \mapsto \setR^N$. The output of this transform, i.e.,  
\begin{align}
	\bxx = \mav\brc{\bs} \in \setR^N \label{Mod:1}
\end{align}
is observed through an \ac{awgn} channel. This means that we are provided with
\begin{align}
	\byy = \bxx + \bww, \label{Mod:2}
\end{align}
where $\bww\in \setR^N$ is an \ac{iid} zero-mean Gaussian vector with variance $\sigma^2$, i.e., $\rmw_n \sim \man\brc{0,\sigma^2}$. Our objective is to learn $\bs$ from $\byy$ considering the fact that $\mav\brc{\cdot}$ is known as a side-information. The most well-known case of this setting is the linear regression model in which $\mav\brc{\cdot}$ is considered to be a linear function, i.e., $\mav\brc{\bs} = \mA \bs$ for some $\mA\in \setR^{N\times K}$.

We intend to deviate from this classical assumption in this work and investigate a larger class of generative models; namely, the higher-order Gaussian random fields defined as follows:
\begin{definition}[Gaussian Random Field]
	Let $\mav\brc{\cdot}: \setR^K \mapsto \setR^N$ be a random mapping whose output entries for an input $\bs \in \setR^K$ are denoted as
	\begin{align}
		\mav\brc{\bs} = \dbc{ \mav_1\brc{\bs} , \ldots , \mav_N\brc{\bs} }^\trp.
	\end{align}
	The mapping $\mav\brc{\cdot}$ is a Gaussian random field with covariance function $\Phi\brc\cdot$, if for any $\bs_1$ and $\bs_2 \in \setR^K$, the entries of $\mav\brc{\bs_i}$ for $i\in \set{1,2}$ are distributed Gaussian and satisfy\footnote{Note that $\inner{\bs_1;\bs_2}$ is the normalized inner-product defined in \eqref{eq:inner}.}
	\begin{align}
		\Ex{\mav_m\brc{\bs_1} \mav_n\brc{\bs_2} }{ } = \mone\set{ m = n} \Phi\brc{ \inner{\bs_1;\bs_2} } 
	\end{align}
	for any $n,m\in\dbc{N}$.
\end{definition} 

The above definition becomes clear if we consider some particular examples: 
\begin{example}[Linear model]
	The most basic Gaussian field is the linear field, for which $\Phi\brc{x} = x$. In this case,  $\mav\brc{\bs} = \mA\bs$ for an $N\times K$ \ac{iid} Gaussian matrix $\mA$ whose elements are zero-mean with variance $1/K$.
\end{example}

\begin{example}[Pure quadratic model]
	As a simple extension of the linear Gaussian field, we consider the quadratic Gaussian random field with $\Phi\brc{x} = x^2$. In this case, a particular output entry is determined as  
	\begin{align}
		\mav_n\brc{\bs} =  \bs^\trp \mJ_n \bs
	\end{align}
	where $\mJ_n$ for $n\in\dbc{N}$ are independent \ac{iid} $K\times K$ random matrices whose entries are zero-mean Gaussian random variables with variance $1/K^2$.
\end{example}

\begin{example}[General quadratic model]
	A more general form of a quadratic field contains also a linear term. This means that a particular output entry is determined as
	\begin{align}
		\mav_n\brc{\bs} =  \bs^\trp \mJ_n \bs + \baa_n^\trp \bs
	\end{align}
	where $\mJ_n\in \setR^{K\times K}$ for $n\in\dbc{N}$ are mutually independent \ac{iid} Gaussian random matrices whose entries are zero-mean with variance $1/K^2$, and $\baa_n\in \setR^K$ are independent \ac{iid} Gaussian vectors with mean zero and covariance $ \mI_K/K$.
\end{example}

In the sequel, we consider a generic covariance function which can be represented as
\begin{align}
	\Phi\brc{ u } = \sum_{\ell=1}^\infty c_\ell u^\ell
\end{align}
for coefficients $c_\ell \geq 0$.

\subsection{Origin of the Higher-Order Model}
The interest on higher-order Gaussian random fields comes from various sources. More relevant to the scope of this work is the so-called \textit{nonlinear Gaussian encryption} method, discussed in \cite{fyodorov2019spin}.  In this method, a nonlinear Gaussian random field is used to encrypt a data sequence $\bs$ into an $\bxx$ of larger length. As shown in \cite{fyodorov2019spin}, the non-linearity adds confusion to the system which prevents a degraded eavesdropper from overhearing the data sequence. As we show in the coming sections, this result can be observed as a non-trivial alternative derivation of Wyner's wiretap result \cite{wyner1975wire}. We discuss this aspect with more details later on in this manuscript.

Besides nonlinear Gaussian encryption, there are various other models which bring interest on higher-order Gaussian fields. An example is the \textit{multi-edge-type parity check code ensembles}, in particular the \ac{ldgm} codes, whose generator can be seen as a field with a polynomial covariance function $\Phi\brc{x}$ \cite{richardson2008modern,kudekar2008proof,kumar2014threshold}. This is the most relevant aspect connected to Sourlas' seminal work in \cite{sourlas1989spin}.

An alternative source of interest comes from a mathematical viewpoint. In fact, a polynomial order Gaussian field is the superposition of multiple multi-dimensional tensors with random entries. This can be seen as a natural extension of a linear field\footnote{A linear field can be seen as a two-dimensional tensor.} into higher dimensions. This viewpoint in particular is of interest in several machine learning applications \cite{chen2021statistical,luneau2019mutual,barbier2019optimal,loureiro2021learning,sellke2021optimizing,ElAlaoui2021spin}.

As mentioned in the introduction, Gaussian random fields are well-known in statistical mechanics; namely, in the theory of spin glasses, in which random fields are used to model the energy level of a spin glass. In the particular case of a linear field, the problem of analyzing the energy model in the thermodynamic limit is homomorphic to the asymptotic analysis of linear regression model\footnote{See for example \cite[Chapter 3]{bereyhi2020thesis}.}. Considering this homomorphism, one can interpret a higher-order Gaussian field as a natural extension of linear regression to the problem of regression with a nonlinear Gaussian generative model. More precisely, the inference from the output of a linear model is mathematically equivalent to the free-energy calculation for a quadratic Hamiltonian being described via the Sherrington–Kirkpatrick model \cite{panchenko2013sherrington,mezard1987spin}
\begin{align}
	\mae \brc{\bs} = -\sum_{k,\ell} J_{k,\ell} s_k s_\ell
\end{align}
Here, $\mae \brc{\bs}$ denotes the Hamiltonian of the spin glass which corresponds to the linear regression problem; see \cite{panchenko2013sherrington} for the formal definition. The entries $J_{k,\ell}$ are further random coefficients. This model determines the energy of a spin glass in terms of pair-wise interactions. A natural extension of this model is the so-called $P$-spin model in which the Hamiltonian is considered to be 
\begin{align}
	\mae \brc{\bs} = -\sum_{\ell_1,\ldots\ell_P} J_{\ell_1,\ldots,\ell_P} \prod_{p=1}^P s_{\ell_p}.
\end{align}
Here, the energy is determined as the sum of $P$-tuple interactions among the particles. This model in the asymptotic case of $P\uparrow\infty$ describes the \ac{rem} in which the energy levels are assigned \ac{iid} at random  and do not depend on the configuration of the system particles anymore \cite{kirkpatrick1987dynamics,gardner1985spin}. The particular case of \ac{rem} has been investigated extensively in the literature; see for example \cite{derrida1980random}.

\subsection{Bayesian Inference}
Let us now get back to our main objective, i.e., reconstructing $\bs$ from $\byy$ given $\mav\brc{\cdot}$. To this end, we estimate $\hat{\bs}$ via a Bayesian estimator as
\begin{align}
	\hat{\bs} = \Ex{\bs \vert \byy , \mav}{} = \int  \dif Q \brc{ \bs \vert \byy , \mav} \bs,
\end{align}
where $Q \brc{ \bs \vert \byy , \mav}$ denotes the posterior cumulative distribution of $\bs$ determined for  some postulated \ac{iid} prior belief
\begin{align}
	\bs \sim q_S\brc{\bs} = \prod_{k=1}^K q_S\brc{s_k},
\end{align}
and some postulated noise variance $\hat{\sigma}^2$. This estimator minimizes the expected error when it is calculated via the Euclidean distance for the given prior belief. 

To evaluate the performance of the Bayesian estimator, we determine the normalized \textit{cross entropy} between the following two distributions:
\begin{enumerate}
	\item The true output distribution, i.e., 
\begin{align}
	p_Y\brc{\byy \vert \mav} = \brc{2\pi \sigma^2}^{-N/2} \Ex{\exp\set{ - \frac{\norm{\byy-\mav\brc{\bs}}^2}{2\sigma^2}}}{p_S},
\end{align}
with the expectation being taken with respect to the true prior distribution $p_S\brc{\bs}$.
\item The output distribution \textit{induced} by the postulated prior and noise variance, i.e., 
\begin{align}
q_Y\brc{\byy \vert \mav} = \brc{2\pi \hat{\sigma}^2}^{-N/2} \Ex{\exp\set{ - \frac{\norm{\byy-\mav\brc{\bs} }^2 }{2\hat{\sigma}^2}}}{q_S}. \label{eq:q_Y}
\end{align}
We refer to this distribution in the sequel as the \textit{induced distribution}.
\end{enumerate}
Hence, the performance metric for fixed $\sigma$, $\hat{\sigma}$, $p_S$ and $q_S$ is given by
\begin{align}
	\Delta_N \brc{ \sigma, \hat{\sigma}, p_S, q_S } = - \frac{1}{N} \Ex{ \log q_Y\brc{\byy \vert \mav}  }{ \byy, \mav }.
\end{align}
Here, the expectation is taken with respect to the true output distribution. It is worth noting that in the definition of $\Delta_N$ ,we further determine the expectation with respect to the field. This means that the determined metric represents the \textit{expected} cross entropy.  The \textit{self-averaging} property indicates that a particular \textit{realization} of the cross entropy converges to this expectation in the large-system limit \cite{de2009self,krzakala2021statistical} .

In general, various metrics can be directly derived from the cross entropy. Some examples are as follows:
\begin{itemize}
	\item The output \textit{differential entropy} is directly determined by setting $\sigma = \hat{\sigma}$ and $p_S = q_S$, i.e.,
	\begin{align}
		\frac{1}{N}	h\brc{\byy \vert \mav} =  \Delta_N \brc{ \sigma, \sigma, p_S, p_S }. \label{eq:H_Delta}
	\end{align}
\item Given the fact that the conditional differential entropy $h\brc{\byy\vert \bs, \mav}$ in the end-to-end setting is determined via the differential entropy of the Gaussian noise process, the input-output information rate is readily written in terms of the cross entropy as
\begin{subequations}
		\begin{align}
	\frac{1}{N}I\brc{\bs;\byy\vert \mav } &= \frac{1}{N} h\brc{\byy \vert \mav} - \frac{1}{N} h\brc{\byy\vert \bs, \mav},\\
	&= \Delta_N \brc{ \sigma, \sigma, p_S, p_S } - \frac{1}{2} \log 2\pi e \sigma^2. \label{eq:MI}
\end{align}
\end{subequations}
\item The average \textit{Kullback–Leibler divergence} between the true distribution and the induced distribution is given by
\begin{subequations}
		\begin{align}
	\Ex{D_{\rm KL} \brc{p_Y \Vert q_Y}}{} &=h\brc{\byy \vert \mav} -  N\Delta_N \brc{ \sigma, \hat{\sigma}, p_S, q_S },\\
	&=N\brc{\Delta_N \brc{ \sigma, \sigma, p_S, p_S }  -  \Delta_N \brc{ \sigma, \hat{\sigma}, p_S, q_S }},
\end{align}
where the expectation is taken with respect to the random field.
\end{subequations}
\end{itemize}

Given the performance metric, our ultimate goal is to evaluate the normalized cross entropy function in the large-system limit, i.e., $N\uparrow \infty$. In particular, we are interested in the regime in which the data-set size, i.e., $K$, and the number of observations, i.e., $N$, grow significantly large with a bounded ratio. Such a regime can be studied by considering a sequence of problems indexed by $N$, in which the data-set size is determined as function of index $N$, i.e., $K\brc{N}$. We are then interested in the limit 
\begin{align}
	\Delta \brc{ \sigma, \hat{\sigma}, p_S, q_S }  = \lim_{N\uparrow \infty} \Delta_N \brc{ \sigma, \hat{\sigma}, p_S, q_S },
\end{align}
assuming that
\begin{align}
	\lim_{N\uparrow\infty} = \frac{K\brc{N}}{N} = R,
\end{align}
for some bounded $R$. In the remaining of this manuscript, we refer to $R$ as the \textit{load} factor of the generative model.

As shown in Section~\ref{Sec:Spin}, the asymptotic limit $\Delta \brc{ \sigma, \hat{\sigma}, p_S, q_S } $ is derived from the free energy of a corresponding spin-glass whose Hamiltonian is described by the $P$-spin model.

\subsection{An Initial Setup}
\label{sec:initial_Setup}
To start with the derivations, we focus on a special case in which the data symbols are simply labels, i.e., $s_k\in \set{\pm 1}$. We further assume that the labels are uniform \ac{iid} labels, and that the prior belief is perfectly matched with the true prior, i.e., 
\begin{align}
	p_S\brc{\bs} = q_S\brc{\bs} = \prod_{k=1}^K \mathrm{B} \brc{s_k; 0.5}
\end{align}
with $\mathrm{B} \brc{s; f}$ describing the Bernoulli distribution at $S=s$ for some $S$ with 
\begin{align}
	\Pr\set{S=+1} = 1-\Pr\set{S=-1} = f.
\end{align}

The interest on this particular prior follows from the connection of this generative model with the nonlinear Gaussian encryption being studied in a related but different setting in \cite{fyodorov2019spin}. We talk about this connection with details in the forthcoming section. It is however worth mentioning that the extension of the current derivations to more general priors is straightforward and is considered for future work. 

\subsection{Some Earlier Studies}
\label{sec:Fyod}

As mentioned, the most related earlier studies to this work are Sourlas' coding scheme explained in \cite{sourlas1989spin,sourlas1994spin} and the nonlinear encryption technique proposed by Fyodorov in \cite{fyodorov2019spin}. The former study by Sourlas utilizes a particular class of Gaussian random fields for channel coding and shows that this scheme sketches a cost-efficiency trade-off which is close to the information-theoretic limit given by Shannon in \cite{shannon1948mathematical}. This seminal idea will be further illustrated throughout our investigations in Section~\ref{Sec:LinvsNonLin}.

The recent related study by Fyodorov is also related to this work in its core idea; however, the setting in this work is different in various respects. For instance, in \cite{fyodorov2019spin}, Fyodorov considers the following setting: 
\begin{enumerate}
	\item The data symbols are supposed to be drawn from a \textit{uniform} distribution on a hypersphere, i.e., $\bs = \sqrt{DP_S} \bs_0$ where $\bs_0 \in \setS^{D-1}$ with $\setS^{D-1}$ denotes the $\brc{D-1}$-sphere.
	\item The learning task is performed via the method of least-squares, i.e., 
	\begin{align}
		\hat{\bs} = \argmin_{ \frac{\bu}{\sqrt{{DP_S}}} \; \in \; \setS^{D-1} } \norm{\byy  - \mav\brc{\bu}}^2. 
	\end{align}
In the Bayesian framework, this describes the \textit{maximum likelihood} estimation which due to the uniform prior is of the similar form as the \textit{\ac{map}} estimation.
\end{enumerate}

Fyodorov determines the asymptotic limit of the average normalized inner product between the ground truth $\bs$ and the estimated parameter $\hat{\bs}$, i.e., he finds the so-called \textit{overlap} $m^\star$ which is defined as\footnote{We later give more discussions on the meaning of $m^\star$.}
\begin{align}
	m^\star = \lim_{N\uparrow \infty} \frac{\inner{\bs;\hat{\bs}}}{P_S}, \label{eq:m_Fyo}
\end{align}
assuming constant ratio $\mu = D/N$. The result of Fyodorov contains both \ac{rs} and \ac{rsb} solutions. Despite its complicated look, the key conclusions are straightforwardly implied from the main result: For \textit{strictly} nonlinear fields, $m^\star$ observes a second-order phase transition at a threshold \ac{snr}, i.e., $P_S/\sigma^2$. Although this second-order phase transition makes a connection to the secure channel coding of Wyner \cite{wyner1975wire}, this connection was left unaddressed in \cite{fyodorov2019spin}, due to the fundamental differences in the inverse problem and the system model considered by Fyodorov. More details on Fyodorov's work can be found in Appendix~\ref{app:Fyodorov}.

\section{Corresponding Spin Glass with $P$-Spin Model}
\label{Sec:Spin}
Our key objective is to determine the limiting normalized cross entropy. We address this goal by using statistical mechanics: We define a corresponding spin glass corresponding to the generative model whose particle interactions are described via a $P$-spin model. We then show that the desired metrics in the inference problem are derived from the  free energy of the corresponding spin glass and utilize the replica method to determine the free energy. 

Let us now start with defining the corresponding spin glass:

\begin{definition}[Corresponding spin glass]
	For a given vector of observations $\byy$ and random field $\mav\brc{\cdot}$, the spin glass corresponding to the nonlinear Gaussian generative model is a thermodynamic system with $K$ particles. The microstate\footnote{Remember that in Section~\ref{sec:initial_Setup} we restricted the labels to be uniformly generated from $\set{\pm 1}$.} of this spin glass is described by $\bu\in \set{\pm 1}^K$ and its energy for a given realization of the microstate $\bu$ is determined via the following Hamiltonian:
	\begin{align}
	\mae \brc{\bu\vert \byy , \mav} = \frac{ \norm{\byy - \mav\brc{\bu}}^2}{2}.
\end{align}
\end{definition}

Starting from this definition, we now follow standard derivations to describe the \textit{macroscopic features} of the corresponding spin glass: The conditional Boltzmann distribution of the microstate  given the \textit{quenched random variables}\footnote{By \textit{quenched} random variables, we refer to $\byy$ and $\mav\brc{\cdot}$ which have different orders of randomness compared to the microstate. More details can be followed in \cite{bereyhi2020thesis}.} at the inverse temperature $\beta$ is given by
\begin{align}
	\rmp_{\beta} \brc{\bu \vert \byy , \mav} = \dfrac{ \exp\set{- \beta \mae\brc{\bu \vert \byy , \mav}} }{ \maz_{\beta} \brc{ \byy , \mav} },
\end{align}
where $\maz_{\beta} \brc{ \byy , \mav}$ is the \textit{partition function}, i.e.,
\begin{align}
	\maz_{\beta} \brc{ \byy , \mav} = \sum_{\bu\in\set{\pm 1}^K} \exp\set{- \beta \mae\brc{\bu\vert \byy , \mav} },
\end{align}

The macroscopic parameters of this spin glass at inverse temperature $\beta$ are fully described via the \textit{normalized average free energy} which is given by\footnote{Note the subscript $N$ refers to the fact that the free energy expression depends on the model dimensions. We however drop $K$, as we assume $K$ and $N$ are related by $K/N=R$.}
\begin{align}
	\maf_N \brc{\beta} = - \frac{1}{N\beta} \Ex{ \log \maz_{\beta} \brc{ \byy , \mav} }{}.
\end{align}
where the expectation is taken over the quenched random variables, i.e., $\byy$ and $\mav\brc{\cdot}$. We now show that the free energy determines our target metric, i.e., the cross entropy. 

\subsection{The Variational Problem}
Let us get back to our main objective, i.e., determining the cross entropy. Since in the current setting $p_S=q_S$, we further drop the arguments $p_S$ and $q_S$ in the cross entropy expression and write
\begin{subequations}
	\begin{align}
	\Delta_N \brc{ \sigma, \hat{\sigma} } &= - \frac{1}{N} \Ex{ \log q_Y\brc{\byy \vert \mav}  }{ },\\
	&= - \frac{1}{N} \Ex{ \log \brc{2\pi \hat{\sigma}^2}^{-N/2} \Ex{\exp\set{ - \frac{\norm{\byy-\mav\brc{\bs} }^2 }{2\hat{\sigma}^2}}}{\bs} }{ },\\
	&=  \frac{1}{2} \log 2\pi \hat{\sigma}^2- \frac{1}{N} \Ex{ \log \Ex{\exp\set{ - \frac{\norm{\byy-\mav\brc{\bs}}^2 }{2\hat{\sigma}^2} } }{\bs}   }{}. \label{Delta_1}
\end{align}
\end{subequations}
Here, we drop the subscript $\byy ,\mav$ for the outer expectation, as it is taken over the \textit{true} distribution of the quenched random variables. Keeping in mind the definition of the partition function and the assumption on the prior distribution, i.e.,
	\begin{align}
	p_{S}\brc{\bs} = 2^{-K}= 2^{-NR},
\end{align}
we can further write
\begin{subequations}
	\begin{align}
 \Ex{\exp\set{ - \frac{\norm{\byy-\mav\brc{\bs}}^2 }{2\hat{\sigma}^2}}}{\bs} &= 2^{-NR} \sum_{\bu\in\set{\pm 1}^K} \exp\set{ - \frac{\norm{\byy-\mav\brc{\bu} }^2 }{2\hat{\sigma}^2}},\\
 &= 2^{-NR} \maz_{1/\hat{\sigma}^{2}} \brc{ \byy , \mav}. \label{eq:Exp}
\end{align}
\end{subequations}
Finally, by replacing \eqref{eq:Exp} into \eqref{Delta_1}, we conclude that
\begin{subequations}
	\begin{align}
	\Delta_N \brc{ \sigma, \hat{\sigma} } &= \frac{1}{2} \log 2\pi \hat{\sigma}^2 + R \log 2 - \frac{1}{N} \Ex{ \log \maz_{1/\hat{\sigma}^{2}} \brc{ \byy , \mav}  }{  },\\
	&=  \frac{1}{2} \log 2\pi \hat{\sigma}^2 + R \log 2 + \frac{1}{\hat{\sigma}^2} 	\maf_N \brc{\frac{1}{\hat{\sigma}^2}} 
\end{align}
\label{eq:Delta_F}
\end{subequations}

The above identity clarifies the connection between the nonlinear generative model and the  corresponding spin glass. As a direct result of this identity, we consider the variational problem of calculating the free energy of the corresponding spin glass. The variational problem is of the same level of hardness as the original problem. We hence follow the standard approach based on the \textit{replica method}; see \cite{bereyhi2020thesis} and the references therein for details.

\subsection{Main Results}
We now present the final solution for the \textit{asymptotic} normalized free energy, i.e.,
\begin{align}
	\maf \brc{\beta} = \lim_{N\uparrow \infty} \maf_N \brc{\beta},
\end{align}
given by the replica method under the \ac{rs} assumption. The derivations are skipped at this point and presented in Appendix~\ref{app:2}. The solution is referred to as the \ac{rs} solution. This is known that for a generic setting, the \ac{rs} assumption is not necessarily valid and hence the corresponding solution is not always reliable. However, for some particular cases the validity of \ac{rs} is guaranteed. We discuss one of these cases, i.e., optimal Bayesian scenario, later on.
\begin{proposition}[RS solution]
	\label{proposition:RS}
	Consider the generative model described by \eqref{Mod:1} and \eqref{Mod:2}. Let $\Phi\brc\cdot$ denote the covariance function of $\mav \brc\cdot$, and $Q$ and $m $ be two scalars in $\dbc{0,1}$. Define $E\brc{Q,m}$ and $L\brc{Q,m}$ in terms of $Q$ and $m$ as
	\begin{subequations}
		\begin{align}
		L_\beta \brc{Q,m} &= \frac{\beta}{R} \Phi'\brc{Q} \frac{f_\beta\brc{Q,m}}{g_\beta^2\brc{Q}}, \\
		E_\beta\brc{Q,m} &=  \frac{\beta}{R} \Phi'\brc{m} \frac{1}{g_\beta\brc{Q}},
	\end{align}
	\end{subequations}
	for functions 
\begin{subequations}
	\begin{align}
	g_\beta\brc{Q} &= 1+\beta \dbc{ \Phi \brc{1} - \Phi \brc{Q}},\\
	f_\beta\brc{Q,m} &= \beta \dbc{ \sigma^2 + \Phi \brc{1} + \Phi \brc{Q} - 2 \Phi\brc{m} }.
\end{align}
\end{subequations}
Let the function $\Pi\brc{Q,m}$ for a given $Q$ and $m$ be
	\begin{align}
	\Pi_\beta\brc{Q,m}= &R \dbc{m E_\beta\brc{Q,m} + \frac{\brc{1-Q}L_\beta\brc{Q,m}}{2}} +  \frac{1}{2} \dbc{ \frac{f_\beta\brc{Q,m}}{ g_\beta\brc{Q}} + \log g_\beta\brc{Q}}.
\end{align}
Assuming the \ac{rs} ansatz, the normalized average free energy of the corresponding spin glass at inverse temperature $\beta$ is given by
	\begin{align}
	\maf \brc{\beta}  = \frac{F_\beta \brc{Q^\star,m^\star}}{\beta },
\end{align}
with $F_\beta \brc{Q,m}$ being defined on $ \dbc{0,1}^2$ as
	\begin{align}
	 F_\beta\brc{Q,m}  \coloneqq -R  \Ex{ \log  \cosh \brc{ \sqrt{L_\beta\brc{Q,m}} Z +  E_\beta\brc{Q,m} }  }{Z}  + \Pi_\beta\brc{Q,m} - R \log 2 .
	\end{align}
Here, $Z$ is a zero-mean Gaussian random variable with unit variance, i.e., $Z\sim \man\brc{0,1}$, and 
\begin{align}
\brc{Q^\star , m^\star} = \argmin_{\brc{Q,m}\in\setF} F_\beta \brc{Q,m}
\end{align}
with $\setF$ being the union of the boundary points of the feasible set $ \dbc{0,1}^2$ and the set of all solutions to the following fixed-point equations 
\begin{subequations}
		\begin{align}
		Q&= \Ex{ \tanh^2 \brc{ \sqrt{L_\beta\brc{Q,m}} Z +  E_\beta\brc{Q,m} }  }{Z},\\
		m &=  \Ex{ \tanh \brc{ \sqrt{L_\beta\brc{Q,m}} Z +  E_\beta\brc{Q,m} }  }{Z}.
	\end{align}
\end{subequations}
\end{proposition}
\begin{proof}
	See Appendix~\ref{app:2}.
\end{proof}

As mentioned, the above asymptotic result is derived under the \ac{rs} ansatz, assuming that the corresponding spin glass shows the so-called \textit{replica symmetry} at the thermal equilibrium. This is however not necessarily the case. At this point, we assume that for the settings of interest, the \ac{rs} ansatz is valid. Further discussions can be followed in Appendix~\ref{app:2}. More detailed discussions in this respect can be found in \cite{bereyhi2020thesis,bereyhi1,bereyhi4}.

Considering the \ac{rs} solution, the asymptotic normalized cross entropy term is readily derived from \eqref{eq:Delta_F} as follows
\begin{subequations}
		\begin{align}
	\Delta \brc{ \sigma, \hat{\sigma} } = & \phantom{+} \frac{1}{2} \log 2\pi \hat{\sigma}^2 + R \log 2 + \frac{1}{\hat{\sigma}^2} 	\maf \brc{\frac{1}{\hat{\sigma}^2}}, \\
	= &  \phantom{+} \frac{1}{2} \log 2\pi \hat{\sigma}^2 -R  \Ex{ \log  \cosh \brc{ \sqrt{L_{1/\hat{\sigma}^2 } \brc{Q^\star,m^\star}} Z +  E_{1/\hat{\sigma}^2 }\brc{Q^\star,m^\star} }  }{Z} \nonumber \\
	&+ \Pi_{1/\hat{\sigma}^2 } \brc{Q^\star,m^\star}.
\end{align}
\end{subequations}
In next sections, we mainly focus on the particular case of \textit{matched} Bayesian estimation, for which it is known that the \ac{rs} ansatz is valid; see for example \cite{iba1999nishimori,antenucci2019approximate,krzakala2021statistical}. Before starting with the optimal Bayesian case, let us discuss the asymptotic decoupling property of the nonlinear model.

\subsection{Decoupling Principle}
The asymptotic result derived by the replica method leads to the \textit{decoupled principle} for the original high-dimensional setting. This is a well-known principle developed for various inference problems with linear model \cite{muller2004capacity,bereyhi1,rangan2012asymptotic,guo2005randomly,vehkapera2010analysis}. The \ac{rs} solution extends this principle to higher-order fields. To state the decoupling principle, consider the following scalar inference problem: The label $S\in \set{\pm 1}$ distributed as $S\sim \mathrm{B} \brc{s; 0.5}$ is observed through a Gaussian channel whose additive noise term is $\rho Z$ with $Z\sim\man\brc{0,1}$. The received symbol is hence represented by
\begin{align}
	Y = S + \rho Z.
\end{align}
Given a particular observation $Y=y$, the scalar Bayesian estimator determines the expectation of $S$ via its postulated posterior distribution $p_{\hat{\rho}}\brc{s\vert y}$ which assumes the noise variance to be $\hat{\rho}$. The estimated symbol is hence given by
\begin{align}
	\hat{S} = \Ex{S \vert Y }{p_{\hat{\rho}} } = \tanh\brc{\frac{Y}{\hat{\rho}^2 }},
\end{align}
where the latter identity is found after few standard lines of derivations.

Considering the above scalar problem, an initial observation is that by setting
\begin{subequations}
	\begin{align}
		\rho &= \frac{\sqrt{ L_{1/\hat{\sigma}^2} \brc{Q^\star,m^\star}  }}{E_{1/\hat{\sigma}^2} \brc{Q^\star,m^\star} },\\
		\hat{\rho} &= \frac{1}{ \sqrt{ E_{1/\hat{\sigma}^2} \brc{Q^\star,m^\star}} },
	\end{align}
\label{eq:dec_eq}
\end{subequations}
the average pair-wise distortion determined by the \ac{rs} solution is given by the expected distortion between $S$ and $\hat{S}$ in the decoupled setting. More precisely, for a given distortion function $D\brc{\cdot; \cdot}: \set{\pm 1} \times \setR \mapsto \setR$, the following identity under the \ac{rs} ansatz exists:
\begin{align}
	\lim_{K\uparrow \infty} \frac{1}{K} \Ex{\sum_{ k=1}^K  D\brc{s_k;\hat{s}_k}}{} = \Ex{D\brc{S;\hat{S}}}{ }.
\end{align}
This finding is however the result of a more general principle: 

\begin{proposition}[Decoupling principle]
	\label{Prop:2}
	Assume \ac{rs} is a valid ansatz. For any $k\in \dbc{K}$, the pair $\brc{s_k,\hat{s}_k}$ converges in distribution to $\brc{S,\hat{S}}$ as $K\uparrow \infty$, when we set $\rho$ and $\hat{\rho}$ according to \eqref{eq:dec_eq} in the decoupled setting.
\end{proposition}
\begin{proof}
	The proof is given via the method of moments in Appendix~\ref{app:dec}.
\end{proof}

A direct result of the decoupling principle is that the fixed-points $Q^\star$ and $m^\star$ in the \ac{rs} solution are given in terms of the decoupled setting as follows:
\begin{subequations}
	\begin{align}
		Q^\star&= \Ex{ \hat{S}^2  }{},\\
		m^\star &=  \Ex{ S \hat{S}}{}.
	\end{align}
\end{subequations}
The implicit use of these identities is to simplify the \ac{rs} solution for optimal Bayesian inference. Nevertheless, they further give operational meaning to $m^\star$ and $Q^\star$. In this respect, $m^\star$ is of particular interest: it determines the expected normalized inner product between the true symbols and the learned ones. It is hence often called the  \textit{overlap}\footnote{Remember the performance metric defined by Fyodorov in \cite{fyodorov2019spin}, i.e., \eqref{eq:m_Fyo}.}. There are in principle two extreme cases for $m^\star$: 
\begin{inparaenum}
	\item Perfect recovery in which $\hat{S} = S$, and hence $m^\star = 1$. 
	\item Independent inference, in which $\hat{S}$ and $S$ are statistically independent resulting in $m^\star = 0$. 
\end{inparaenum}
These extreme cases limit $m^\star \in \dbc{0,1}$ which is consistent with the constraint stated in the \ac{rs} solution. They moreover provide a comprehensive interpretation for $m^\star$: \textit{Larger $m^\star$ means more information}.

\begin{remark}
	One should note that the decoupling principle only describes the \textit{marginal} distribution of each pair of true label and its estimation. Hence, it does not determine the metrics which take into account correlation among different pairs, e.g., the cross entropy.
\end{remark}

\subsection{Optimal Bayesian Inference}
The mean squared error between the true labels and their recovery is minimized with respect to $\hat{\sigma}$, when we set the postulated noise variance equal to the true noise variance, i.e., when $\hat{\sigma} = \sigma$. Due to this property, this particular form of the Bayesian inference is referred to as the \textit{optimal} Bayesian inference.  This case describes the Nishimori line which lies in the \ac{rs} stability region \cite{antenucci2019approximate}. This means that the \ac{rs} ansatz, for this particular case, is valid and no further derivations under \ac{rsb} ans\"atze are required. In addition to \ac{rs} stability, the symmetry of matched estimator leads to this basic conclusion that $m^\star=Q^\star$: Following the decoupling principle, $m^\star$ determines the cross correlation between the learned symbol and the true one, and $Q^\star$ calculates the second moment of the learned symbol. As the estimator is optimal Bayesian, the orthogonality principle is satisfied. This means that the estimation error is orthogonal to the estimation, i.e., 
	\begin{align}
  \Ex{ \brc{S - \hat{S}} \hat{S} }{} = 0.
\end{align}
This identity concludes that
	\begin{align}
 m^\star=  \Ex{ S \hat{S}}{} = \Ex{ \hat{S}^2  }{} = Q^\star.
\end{align}
Following this straightforward justification, one can see that $m^\star = Q^\star$ is a fundamental property in the replica analysis of optimal Bayesian inference. This fundamental property is in fact the key observation used in \cite[Appendix II]{muller2004capacity} to show the initial form of the decoupling principle.

From these properties of the optimal Bayesian estimator, we can write
\begin{subequations}
	\begin{align}
	g_{1/\sigma^2} \brc{m} &= 1 + \frac{ \Phi \brc{1} - \Phi \brc{m}}{\sigma^2} \\
	f_{1/\sigma^2} \brc{m,m} &=  1 + \frac{\Phi \brc{1} - \Phi \brc{m} }{\sigma^2} = g_{1/\sigma^2}\brc{m}
\end{align}
\end{subequations}
This leads to this conclusion that 
\begin{align}
	L_{1/\sigma^2} \brc{m,m}= E_{1/\sigma^2} \brc{m,m} = \frac{1}{R\sigma^2} \frac{\Phi'\brc{m}}{g\brc{m}}.
\end{align}
Based on these observations, we can simplify the notation by dropping the inverse temperature index\footnote{Since we are only operating at $\beta = 1/\sigma^2$.}, and defining $g\brc{m} \coloneqq g_{1/\sigma^2} \brc{m}$ and
\begin{subequations}
	\begin{align}
	E\brc{m} \coloneqq E_{1/\sigma^2} \brc{m,m} &= \frac{1}{R\sigma^2}  \frac{\Phi'\brc{m}}{g\brc{m}} \\
	&= -\frac{1}{R}  \frac{g'\brc{m}}{g\brc{m}} = -\frac{1}{R} \frac{\dif}{\dif m} \log g\brc{m}. \label{E_m}
\end{align}
\end{subequations}

We now recall that the normalized cross entropy in the matched case determines the normalized differential entropy of the output; see \eqref{eq:H_Delta}. We hence replace the simplified definitions into the \ac{rs} solution and conclude the following result:
\begin{corollary}[Optimal Bayesian Inference]
	\label{cor:1}
	Let $\mah\brc{\sigma}$ denote the normalized differential entropy of the observations, i.e., 
		\begin{align}
		\mah\brc{\sigma} = \Delta\brc{\sigma,\sigma} = \lim_{N\uparrow \infty} \frac{1}{N} h\brc{\byy \vert \mav}.
	\end{align}
Define function $g\brc{\cdot}$ as
	\begin{align}
	g \brc{m} &= 1 + \frac{ \Phi \brc{1} - \Phi \brc{m}}{\sigma^2} ,
\end{align}
and let $E\brc{m}$ be defined in terms of $g \brc{m} $ as in \eqref{E_m}. Then, $\mah\brc{\sigma}$ is determined as
			\begin{align}
	\mah\brc{\sigma}  = H_\sigma\brc{m^\star},
	\end{align}
for $H_\sigma\brc{m}:\dbc{0,1}\mapsto \setR$ defined as
\begin{align}
 H_\sigma\brc{m} \coloneqq -R  \Ex{ \log  \cosh \brc{ \sqrt{E \brc{m}} Z +  E\brc{m} }  }{Z} +  \Pi \brc{m},
\end{align}
where $Z\sim \man\brc{0,1}$, $\Pi \brc{\cdot }$ is defined as
\begin{align}
	\Pi \brc{m} \coloneqq \frac{1}{2} \dbc{ R\brc{1+m} E \brc{m} + \log {2\pi e {\sigma}^2 }  + \log {g\brc{m}} },
\end{align}
and $m^\star \in \dbc{0,1}$ is the minimizer of $H_\sigma\brc{m}$ over the union of boundary points $\set{0,1}$ and the set of solutions to the following fixed point equation
	\begin{align}
		m&=  \Ex{ \tanh \brc{ \sqrt{E\brc{m}} Z + E\brc{m} }  }{Z}. 
	\end{align}
\end{corollary}
\begin{proof}
	The proof follows simply from arguments on the stability of \ac{rs} ansatz in the matched Bayesian inference\footnote{For instance what reported in \cite{antenucci2019approximate}.} after setting $\hat{\sigma} = \sigma$ and $m=Q$ in the \ac{rs} solution. 
\end{proof}

With optimal Bayesian estimation, we are often more interested on the \textit{information rate} of the model, i.e., the mutual information per dimension. This follows the fact that by matched Bayesian inference, there is no information loss at the learning stage, and thus $I\brc{\bs;\hat{\bs}\vert \mav } = I\brc{\bs;\byy\vert \mav }$. Using the identity in \eqref{eq:MI}, we have
	\begin{align}
		\mai\brc{\sigma} = \lim_{N\uparrow \infty}  \frac{1}{N}I\brc{\bs;\byy\vert \mav } = \mah \brc{\sigma} - \frac{1}{2} \log 2\pi e \sigma^2.\label{eq:MI_F}
	\end{align}
We use this metric through experimental investigations.

\section{Linear Models versus Nonlinear Models}
\label{Sec:LinvsNonLin}
Given the main results, we now consider some illustrative special cases to investigate the impacts of having a non-linear Gaussian generative model. To this end, we focus on the optimal Bayesian case and consider a \textit{pure order $\lambda$} polynomial covariance function, i.e., $\Phi \brc{x} = c \left. x^\lambda \right.$ for some constant $c\in \setR^+$. With this covariance function, the random field describes a $\brc{\lambda+1}$-dimensional tensor\footnote{Remember that a matrix is a \textit{two-dimensional} tensor.} whose components are generated at random with some Gaussian prior. This means that the $n$-th component of the random transform is given in terms of $\bs$ as
\begin{align}
	\mav_n \brc{\bs} = \sum_{\ell_1,\ldots,\ell_\lambda = 1}^{K} J^n_{\ell_1,\ldots,\ell_\lambda} \prod_{t=1}^\lambda s_{\ell_t},
\end{align}
for some $J^n_{\ell_1,\ldots,\ell_\lambda}$ which are \ac{iid} Gaussian random variables in $n,\ell_1,\ldots,\ell_\lambda$. 

To keep the comparison among random fields of different orders fair, we assume that all transformed symbol are unit variance, i.e., 
\begin{align}
	\Ex{\mav_n \brc{\bs}^2 }{}= 1.
\end{align}
This concludes that $J^n_{\ell_1,\ldots,\ell_\lambda} \sim \man\brc{0,1/\sqrt{K^\lambda}}$ for $n \in \dbc{N}$ and $\ell_1,\ldots,\ell_\lambda\in \dbc{K}$ and 
\begin{align}
\Phi \brc{x} = x^\lambda.
\end{align}

Noting that $\Phi\brc{1} = 1$, we have $g\brc{m} = 1+ \xi_m\brc{\sigma}$, where
\begin{align}
	\xi_m\brc{\sigma} = \frac{1-m^\lambda}{\sigma^2}.
\end{align}
From $g\brc{m}$, $E\brc{m}$ is determined via \eqref{E_m}. By substituting in Corollary~\ref{cor:1}, we conclude that the asymptotic information rate reads $\mai\brc{\sigma} = \mal_{m^\star}\brc{\sigma}$ with $\mal_m\brc{\sigma}$ given as follows:
\begin{align}
	\mal_m\brc{\sigma}  =  \frac{1}{2} \log \brc{1+\xi_m\brc{\sigma}} + \maq_m \brc{\sigma}, \label{eq:L_m}
\end{align}
where the latter term is determined for a given $m$ and noise variance $\sigma$ as
\begin{align}
	\maq_m\brc{\sigma}  =  \frac{ \brc{1  +  m} R }{ 2  \rho_m^2\brc{\sigma} }  -R  \Ex{ \log  \cosh \brc{  \frac{1  + \rho_m\brc{\sigma}Z }{   \rho_m^2\brc{\sigma} } } }{Z} ,
\end{align}
with $ \rho_m\brc{\sigma} $ being defined as
\begin{align}
 \rho_m\brc{\sigma} = 	 \sqrt{\frac{ R \sigma^2  }{\lambda m^{\lambda-1} } \brc{1 +\xi_m \brc{\sigma} } } .
\end{align}
We refer to $\mal_m\brc{\sigma}$ as the \textit{modified} free energy. This appellation becomes clear shortly. 

The fixed-point $m^\star$ is further determined by minimizing $\mal_m\brc{\sigma}$ over $m\in\dbc{0,1}$. The minimizer is readily found by first solving 
\begin{align}
	m  =    \Ex{ \tanh \brc{  \frac{1  + \rho_m\brc{\sigma}Z }{   \rho_m^2\brc{\sigma} } } }{Z} \label{eq:fix_m}
\end{align}
for $m$ and then minimizing $\mal_m\brc{\sigma}$ over the union of the solution set and the set of boundary points $\set{0,1}$.

The fixed-point equation describes the extreme points of the free energy; see Appendix~\ref{app:2}. Following the fact that the information rate and the free energy behave the same\footnote{Remember that the information rate is simply a shifted version of the normalized free energy; see \eqref{eq:MI_F}.} in $m$, it is concluded that $m^\star$ is the global minimizer\footnote{This further clarifies the appellation.} of $\mal_m\brc{\sigma}$ within interval $m\in \dbc{0,1}$. The information rate is plotted in terms of $m$ for $\lambda=1$ and $\lambda=2$ in Fig.~\ref{fig:1}, where we set $\sigma^2 = 0.1$ and $R=1.76$. The solution $m^\star$ is further shown for both curves in the figure.
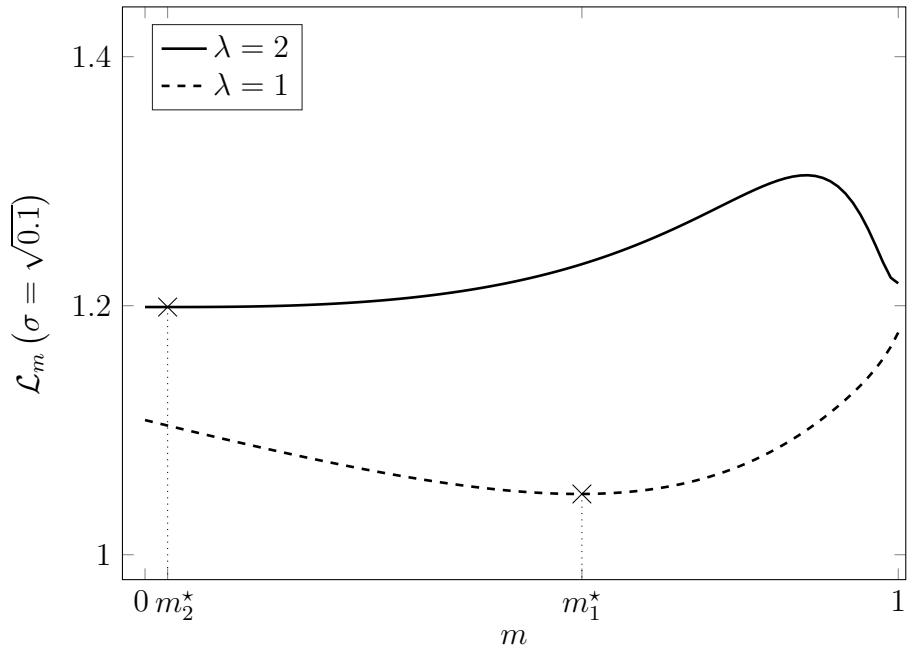
\begin{figure}
	\begin{center}
%
%
\definecolor{mycolor1}{rgb}{0.00000,0.44700,0.74100}%
\definecolor{mycolor2}{rgb}{0.85000,0.32500,0.09800}%
\definecolor{mycolor3}{rgb}{0.92900,0.69400,0.12500}%
\definecolor{mycolor4}{rgb}{0.49400,0.18400,0.55600}%
\begin{tikzpicture}

\begin{axis}[%
width=4.1in,
height=3in,
at={(1.011in,0.642in)},
scale only axis,
xmin=-0.03,
xmax=1.01,
xlabel={$m$},
xtick={0,0.03,0.58,1},
xticklabels={{$0\hspace*{1mm}$},{$\hspace{2mm}m_2^\star$},{$m_1^\star$},{$1$}},
ymin=0.98,
ymax=1.44,
ytick={1,1.2,1.4},
yticklabels={{$1$},{$1.2$},{$1.4$}},
ylabel={$\mal_m\brc{\sigma = \sqrt{0.1}}$},
axis background/.style={fill=white},
legend style={at={(0.23,.97)},legend cell align=left, align=left, draw=white!15!black}
]
\addplot [color=black,line width=1.0pt]
  table[row sep=crcr]{%
0	1.19894763639919\\
0.01	1.19894645082403\\
0.02	1.19894411421491\\
0.03	1.19894235677482\\
0.04	1.19894278874813\\
0.05	1.19894691527363\\
0.06	1.19895614938491\\
0.07	1.19897182346773\\
0.08	1.1989951994189\\
0.09	1.1990274777036\\
0.1	1.1990698054715\\
0.11	1.19912328386384\\
0.12	1.19918897462156\\
0.13	1.19926790608689\\
0.14	1.19936107867643\\
0.15	1.19946946989232\\
0.16	1.19959403892839\\
0.17	1.19973573092022\\
0.18	1.19989548088125\\
0.19	1.20007421736163\\
0.2	1.20027286586142\\
0.21	1.20049235202556\\
0.22	1.20073360464466\\
0.23	1.20099755848206\\
0.24	1.20128515694505\\
0.25	1.20159735461541\\
0.26	1.20193511965209\\
0.27	1.20229943607678\\
0.28	1.20269130595077\\
0.29	1.20311175144979\\
0.3	1.20356181684141\\
0.31	1.20404257036753\\
0.32	1.20455510603265\\
0.33	1.2051005452963\\
0.34	1.20568003866588\\
0.35	1.20629476718377\\
0.36	1.20694594379974\\
0.37	1.20763481461691\\
0.38	1.2083626599962\\
0.39	1.20913079550042\\
0.4	1.20994057265503\\
0.41	1.21079337949756\\
0.42	1.21169064088263\\
0.43	1.21263381850237\\
0.44	1.21362441057546\\
0.45	1.21466395114881\\
0.46	1.21575400894618\\
0.47	1.21689618568652\\
0.48	1.21809211378088\\
0.49	1.21934345330125\\
0.5	1.22065188809564\\
0.51	1.22201912090213\\
0.52	1.22344686728869\\
0.53	1.22493684821531\\
0.54	1.2264907809791\\
0.55	1.22811036826086\\
0.56	1.22979728494148\\
0.57	1.23155316229717\\
0.58	1.23337956911226\\
0.59	1.23527798916437\\
0.6	1.23724979443735\\
0.61	1.23929621329828\\
0.62	1.24141829273313\\
0.63	1.24361685356566\\
0.64	1.24589243738118\\
0.65	1.24824524363301\\
0.66	1.25067505511663\\
0.67	1.25318114964528\\
0.68	1.25576219533785\\
0.69	1.25841612642108\\
0.7	1.26113999583618\\
0.71	1.26392980020335\\
0.72	1.26678027181321\\
0.73	1.26968463125314\\
0.74	1.27263429300851\\
0.75	1.27561851487057\\
0.76	1.27862398020337\\
0.77	1.28163430004476\\
0.78	1.28462941963422\\
0.79	1.28758491130008\\
0.8	1.29047113279599\\
0.81	1.29325222737448\\
0.82	1.29588493957175\\
0.83	1.29831721969778\\
0.84	1.30048659192703\\
0.85	1.30231826841949\\
0.86	1.30372300999759\\
0.87	1.30459477119203\\
0.88	1.30480823934976\\
0.89	1.30421651132369\\
0.9	1.30264939630643\\
0.91	1.29991327938631\\
0.92	1.29579429369624\\
0.93	1.29006804429376\\
0.94	1.28252191621922\\
0.95	1.27300131682815\\
0.96	1.26150166384055\\
0.97	1.24834947421594\\
0.98	1.23456338908919\\
0.99	1.22260174951572\\
0.995	1.220215835065590\\
1	1.21802992061546\\
};
\addlegendentry{$\lambda=2$}

\addplot [color=black, dotted, mark=x, mark options={ mark size = 5.0pt,solid, black},forget plot]
  table[row sep=crcr]{%
0.03	1.19894235677482\\
0.03	-5\\
};

\addplot [color=black, dashed, line width = 1.0pt]
  table[row sep=crcr]{%
0	1.10816262622021\\
0.01	1.10668403217302\\
0.02	1.10521047826088\\
0.03	1.10374227056557\\
0.04	1.10227972693735\\
0.05	1.10082317748077\\
0.06	1.09937296506289\\
0.07	1.09792944584497\\
0.08	1.0964929898388\\
0.09	1.09506398148907\\
0.1	1.09364282028293\\
0.11	1.09222992138826\\
0.12	1.09082571632203\\
0.13	1.08943065365041\\
0.14	1.08804519972214\\
0.15	1.08666983943694\\
0.16	1.0853050770508\\
0.17	1.08395143702\\
0.18	1.08260946488584\\
0.19	1.08127972820234\\
0.2	1.07996281750887\\
0.21	1.07865934735036\\
0.22	1.07736995734723\\
0.23	1.07609531331782\\
0.24	1.07483610845591\\
0.25	1.07359306456623\\
0.26	1.07236693336084\\
0.27	1.07115849781953\\
0.28	1.06996857361742\\
0.29	1.06879801062322\\
0.3	1.0676476944715\\
0.31	1.06651854821281\\
0.32	1.06541153404521\\
0.33	1.06432765513135\\
0.34	1.06326795750487\\
0.35	1.06223353207046\\
0.36	1.06122551670165\\
0.37	1.06024509844071\\
0.38	1.05929351580494\\
0.39	1.05837206120373\\
0.4	1.05748208347069\\
0.41	1.05662499051501\\
0.42	1.05580225209631\\
0.43	1.05501540272659\\
0.44	1.05426604470315\\
0.45	1.05355585127521\\
0.46	1.05288656994714\\
0.47	1.0522600259199\\
0.48	1.05167812567167\\
0.49	1.05114286067736\\
0.5	1.0506563112653\\
0.51	1.05022065060765\\
0.52	1.04983814883881\\
0.53	1.04951117729355\\
0.54	1.04924221285327\\
0.55	1.04903384238506\\
0.56	1.04888876725361\\
0.57	1.04880980788008\\
0.58	1.04879990831604\\
0.59	1.04886214079215\\
0.6	1.0489997101922\\
0.61	1.04921595839194\\
0.62	1.04951436838932\\
0.63	1.04989856813724\\
0.64	1.05037233397252\\
0.65	1.0509395935138\\
0.66	1.05160442787778\\
0.67	1.05237107303567\\
0.68	1.05324392010171\\
0.69	1.054227514311\\
0.7	1.05532655240749\\
0.71	1.05654587812363\\
0.72	1.0578904753937\\
0.73	1.05936545890694\\
0.74	1.06097606157768\\
0.75	1.06272761849818\\
0.76	1.0646255469575\\
0.77	1.06667532217763\\
0.78	1.06888244856808\\
0.79	1.07125242657998\\
0.8	1.07379071572684\\
0.81	1.07650269514347\\
0.82	1.07939362435552\\
0.83	1.08246860899851\\
0.84	1.08573257947725\\
0.85	1.08919029564513\\
0.86	1.09284639853271\\
0.87	1.09670554256696\\
0.88	1.10077266118971\\
0.89	1.10505344940091\\
0.9	1.10955519527995\\
0.91	1.11428817003231\\
0.92	1.119267911102\\
0.93	1.12451893682835\\
0.94	1.13008076836473\\
0.95	1.1360177009505\\
0.96	1.1424347349719\\
0.97	1.14950376749417\\
0.98	1.15750716756249\\
0.99	1.16691142104047\\
1	1.17849412753497\\
};
\addlegendentry{$\lambda=1$}

\addplot [color=black, dotted, mark=x, mark options={mark size = 5.0pt,solid, black},forget plot]
  table[row sep=crcr]{%
0.58	1.04879990831604\\
0.58	-5\\
};

\end{axis}
\end{tikzpicture}%
	\end{center}
	\caption{Information rate against $m$ for linear and quadratic random Gaussian fields. The overlap $m^\star$ is given by the minima of the information rate.}
	\label{fig:1}
\end{figure}

Given the above derivations, we now study the asymptotic characteristics of the optimal Bayesian learning algorithm for different choices of $\lambda$. Before, we start with the investigations, let us derive a simple lower- and an upper-bound on the information rate which is useful throughout the investigations.

\begin{corollary}
	Let $m^\star$ be the overlap. Then, there exists $0 \leq C_{m^\star} \leq 1$, such that
	\begin{align}
		\mai\brc{\sigma}  = \frac{1}{2} \dbc{ \log \brc{1+ \xi_{m^\star}\brc{\sigma} } + \frac{\brc{C_{m^\star}  + m^\star} R}{ \rho_{m^\star}^2 \brc{\sigma} }} - R \log \cosh \brc{  \frac{1}{  \rho_{m^\star}^2 \brc{\sigma} } }.
	\end{align}
\label{cor:2}
\end{corollary}
\begin{proof}
The proof follows Jensen's inequality. The functions $\log \brc\cdot : \left(0, \infty\right] \mapsto \setR$ and $\log \cosh \brc\cdot : \setR  \mapsto \left[0, \infty\right]$ are concave and convex, respectively. We hence can write
\begin{subequations}
	\begin{align}
		\Ex{\log  \cosh \brc{  \frac{1  + \rho_m\brc{\sigma}Z }{   \rho_m^2\brc{\sigma} } } }{} &\leq \log  \Ex{\cosh \brc{  \frac{1  + \rho_m\brc{\sigma}Z }{   \rho_m^2\brc{\sigma} } } }{},\\
		&= \log \exp\set{\frac{1}{  2 \rho_m^2\brc{\sigma} }} \cosh \brc{  \frac{1}{   \rho_m^2\brc{\sigma} } }, \\
		&=  \frac{1}{  2 \rho_m^2\brc{\sigma} } + \log \cosh \brc{  \frac{1}{   \rho_m^2\brc{\sigma} } },
	\end{align}
\end{subequations}
using the concavity of $\log\brc{\cdot}$, and 
\begin{subequations}
	\begin{align}
		\Ex{\log  \cosh \brc{  \frac{1  + \rho_m\brc{\sigma}Z }{   \rho_m^2\brc{\sigma} } } }{} &\geq \log  \cosh \brc{  \Ex{\frac{1  + \rho_m\brc{\sigma}Z }{   \rho_m^2\brc{\sigma} } }{}},\\
		&= \log \cosh \brc{  \frac{1}{   \rho_m^2\brc{\sigma} } },
	\end{align}
\end{subequations}
from the convexity of $\log \cosh \brc\cdot$. By substituting into \eqref{eq:L_m}, we can show that $\mal_m\brc{\sigma}$ is bounded from below as
\begin{align}
	\mal_m\brc{\sigma}  \geq \frac{1}{2} \dbc{\log \brc{1+ \xi_{m}\brc{\sigma} } + \frac{mR}{ \rho_m^2\brc{\sigma} }} - R \log \cosh \brc{  \frac{1}{   \rho_m^2\brc{\sigma} } },
\end{align}
and from above as
\begin{align}
	\mal_m\brc{\sigma}  \leq \frac{1}{2} \dbc{\log \brc{1+ \xi_{m}\brc{\sigma} } + \frac{\brc{1+m}R}{ \rho_m^2\brc{\sigma} }} - R \log \cosh \brc{  \frac{1}{   \rho_m^2\brc{\sigma} } },
\end{align}
We therefore conclude that there exists $0 \leq C_m \leq 1$, such that 
\begin{align}
	\mal_m\brc{\sigma}  = \frac{1}{2} \dbc{\log \brc{1+ \xi_{m}\brc{\sigma} } + \frac{\brc{C_m+m}R}{ \rho_m^2\brc{\sigma} }} - R \log \cosh \brc{  \frac{1}{   \rho_m^2\brc{\sigma} } }.
\end{align}
By setting $\mai\brc{\sigma} = \mal_{m^\star}\brc{\sigma}$, the proof is concluded.
\end{proof}

\subsection{Classical Case: Linear Generative Model}
The classical linear model is studied by setting $\lambda=1$ in the given results. This is a widely-studied setting whose replica analysis is given in various lines of work in the literature, e.g., \cite{tanaka2002statistical,guo2005randomly,tulino2013support,bereyhi1Extension,bereyhi4}. We start our investigation by considering the overlap for this field. An initial finding in this respect is given in the following lemma:

\begin{lemma}[Inference guarantee]
	The overlap of linear model is always non-zero.
	\label{lem:1}
\end{lemma}
\begin{proof}
For $\lambda=1$, the term $\rho_m\brc{\sigma}$ reduces to
\begin{align}
	\rho_m\brc{\sigma} = 	 \sqrt{ R  \brc{\sigma^2 + 1 - m } }. \label{eq:rho_T_1}
\end{align}
By substituting in the fixed-point equation, it is observed that $m=0$ is not a solution for any $R$ and $\sigma$. This concludes that $m=0$ is not an saddle point of $\mal_m\brc{\sigma}$, or equivalently
\begin{align}
	\frac{\partial}{\partial m} \mal_m \brc{\sigma} \vert_{m=0} \neq 0, \label{der:1}
\end{align}
for any choice of $R$ and $\sigma$. We now show that there exists always a point $m\in \dbc{0,1}$ for which $\mal_m\brc{\sigma} $ is smaller than $\mal_0\brc{\sigma} $: Let us define the Gaussian random variable $\tilde{Z}$ as
\begin{align}
	\tilde{Z} \coloneqq \frac{1  + \rho_0\brc{\sigma}Z }{   \rho_0^2\brc{\sigma} } ,
\end{align}
for $Z\sim\man\brc{0,1}$. The derivative of $\mal_m\brc{\sigma}$ at $m=0$ is then determined in terms of $\tilde{Z}$ as
\begin{align}
	\frac{\partial}{\partial m} \mal_m \brc{\sigma} \vert_{m=0} = \frac{R^2}{2 \rho_0^2\brc{\sigma} } \brc{1 - \Ex{\tanh \tilde{Z} }{} - \rho_0^2\brc{\sigma} \Ex{\tilde{Z}\tanh \tilde{Z} }{} }.
\end{align}
Noting that $\tilde{Z}\sim\man\brc{1/\rho_0^2\brc{\sigma},1/\rho_0^2\brc{\sigma}}$, one can use integration by parts to show that
\begin{align}
 \Ex{\tilde{Z}\tanh \tilde{Z} }{} = \frac{1}{\rho_0^2\brc{\sigma}}  \Ex{1+ \tanh \tilde{Z} - \tanh^2 \tilde{Z} }{}.
\end{align}
As the result, we have 
\begin{align}
	\frac{\partial}{\partial m} \mal_m \brc{\sigma} \vert_{m=0} = \frac{R^2}{2 \rho_0^2\brc{\sigma} } \brc{\Ex{\tanh^2 \tilde{Z} }{} -2 \Ex{\tanh \tilde{Z} }{} } \leq 0. \label{der:2}
\end{align}
Considering \eqref{der:1} and \eqref{der:2}, we conclude that for any $R>0$ and $\sigma$,  $\mal_m \brc{\sigma} $ decreases in $m$ at $m=0$. This leads to this conclusion that there exists a point $\tilde{m}$ in the right neighborhood of $m=0$ for which we have $\mal_{\tilde{m}}\brc{\sigma} \leq \mal_0\brc{\sigma}$; therefore, $m=0$ cannot be the global minimizer of $\mal_m \brc{\sigma}$ on $\dbc{0,1}$.
\end{proof}

Non-zero overlap indicates that in the linear model the learned symbol are always correlated to the true labels. This is in fact an intuitive observation which confirms a rather known property: Linear models always carry information about the model parameters.

In general, the overlap is a decreasing function of the load $R$: For smaller $R$ the problem is better determined, and hence the learned labels are more correlated to the true symbols\footnote{This is in fact true for noisy observations. For the case of $\sigma^2 = 0$, however, the phase transition occurs with different scaling. Namely, to have an overlap smaller than one, one needs to let the load to be scaled logarithmically (or faster) with the dimension; see \cite{sedaghat2015optimum} for more details.}. This is shown in Fig.~\ref{fig:2}, where we plot the overlap $m^\star$ against $R$ for $\sigma^2 = 0.1$. 
\begin{figure}
	\begin{center}
%
%
\begin{tikzpicture}
	
	\begin{axis}[%
		width=4.1in,
		height=3in,
		at={(1.262in,0.703in)},
		scale only axis,
		xmin=0,
		xmax=6,
		xlabel style={font=\color{white!15!black}},
		xlabel={$R $},
		ymin=0,
		ymax=1.02,
		ylabel style={font=\color{white!15!black}},
		ylabel={$m^\star$},
		axis background/.style={fill=white},
		legend style={legend cell align=left, align=left, draw=white!15!black}
		]
		
		\addplot [color=black, dashed, line width = 1.0pt, forget plot]
		table[row sep=crcr]{%
			0.3	0.999999987881823\\
			0.4	0.999999106893226\\
			0.5	0.999987948275008\\
			0.6	0.999930455224047\\
			0.7	0.999752212823603\\
			0.8	0.999342952388489\\
			0.9	0.998558529839791\\
			1	0.997206103221256\\
			1.1	0.994995966237075\\
			1.2	0.991414847344588\\
			1.3	0.985332548993985\\
			1.4	0.973119690313951\\
			1.5	0.790221176140756\\
			1.6	0.671110883155409\\
			1.7	0.606605265394437\\
			1.8	0.558565891496958\\
			1.9	0.519788701734557\\
			2	0.487220071002547\\
			2.1	0.459192127901564\\
			2.2	0.434663497995862\\
			2.3	0.412927769807898\\
			2.4	0.393478006492662\\
			2.5	0.375935360878735\\
			2.6	0.360007903581219\\
			2.7	0.345465235010545\\
			2.8	0.33212197140472\\
			2.9	0.319826579773078\\
			3	0.308453535198812\\
			3.1	0.2978976598236\\
			3.2	0.288069964043189\\
			3.3	0.278894508697322\\
			3.4	0.27030598465554\\
			3.5	0.262247848961189\\
			3.6	0.25467085342708\\
			3.7	0.247531854410122\\
			3.8	0.240792873222419\\
			3.9	0.234420319423557\\
			4	0.228384364187946\\
			4.1	0.222658404479089\\
			4.2	0.217218634898593\\
			4.3	0.21204367925827\\
			4.4	0.207114280114763\\
			4.5	0.202413044252614\\
			4.6	0.19792420774278\\
			4.7	0.193633453136597\\
			4.8	0.189527742655351\\
			4.9	0.185595176035379\\
			5	0.181824867211697\\
			5.1	0.178206836945606\\
			5.2	0.174731917603926\\
			5.3	0.171391676789781\\
			5.4	0.168178335742094\\
			5.5	0.165084712070095\\
			5.6	0.162104161582969\\
			5.7	0.159230528239632\\
			5.8	0.15645809964345\\
			5.9	0.153781567370452\\
			6	0.151195991523126\\
		};
		
	\end{axis}
\end{tikzpicture}%
	\end{center}
	\caption{Overlap $m^\star$ versus $R$ for the linear field, i.e., $\lambda=1$, and $\sigma^2 = 0.1$. The overlap in this case never touches the line $m^\star = 0$.}
	\label{fig:2}
\end{figure}
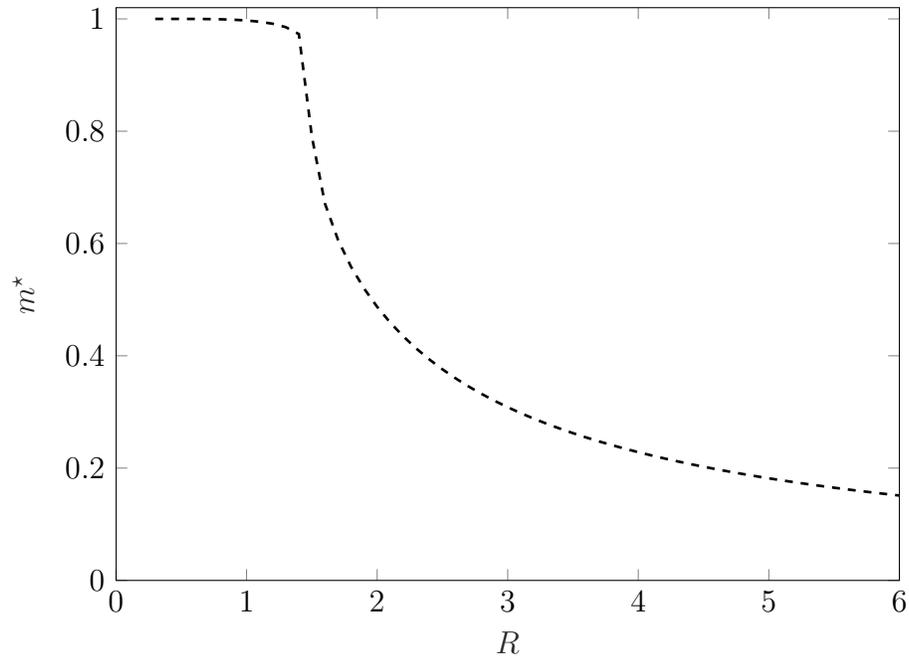

We now consider the information rate $\mai\brc{\sigma}$ and investigate its behavior. To this end, we plot the information rate against the load $R$ in Fig.~\ref{fig:3} for the same setting, i.e., $\sigma^2 = 0.1$. For sake of comparison, we further plot the upper-bound
\begin{align}
	\mai\brc{\sigma} = \lim_{N\uparrow \infty}  \frac{1}{N}I\brc{\bs;\byy\vert \mav } \leq   \frac{1}{N}H\brc{\bs\vert \mav } = R \log 2.
\end{align}

As the figure shows, the curves are closely consistent up to some critical load at which the information rates starts to saturate. The initial observation further depicts that the critical load corresponds to the point on Fig~\ref{fig:2} where $m^\star$ starts to deviate from the extreme case of $m^\star=1$. To further examine the generality of this finding, we use Corollary~\ref{cor:2}: Replacing $\rho_m\brc{\sigma}$ with \eqref{eq:rho_T_1}, and then substituting into Corollary~\ref{cor:2}, we have
\begin{align}
	\mai\brc{\sigma}  = \frac{1}{2} \dbc{\log \brc{1+\frac{1-{m^\star}}{\sigma^2}} + \frac{C_{m^\star}+{m^\star}}{ \sigma^2+1-{m^\star}}} - R \log \cosh \brc{  \frac{1}{  R \brc{\sigma^2 +1-{m^\star}} } }. \label{eq:Rate_1_analysis}
\end{align}
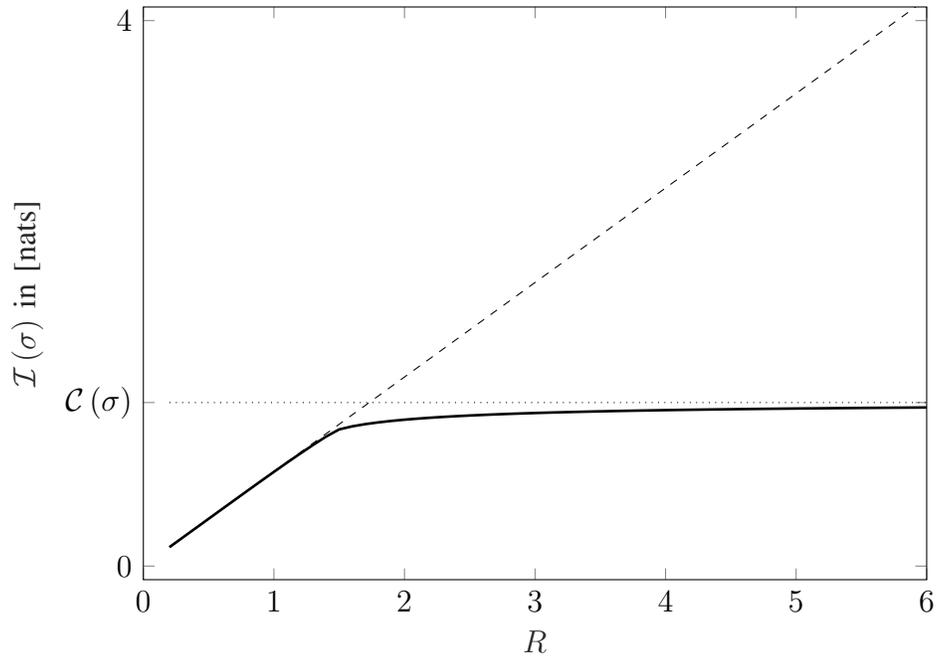
\begin{figure}
	\begin{center}
%
%
\definecolor{mycolor1}{rgb}{0.00000,0.44700,0.74100}%
\definecolor{mycolor2}{rgb}{0.85000,0.32500,0.09800}%
\definecolor{mycolor3}{rgb}{0.92900,0.69400,0.12500}%
\begin{tikzpicture}

	\begin{axis}[%
	width=4.1in,
	height=3in,
	at={(1.262in,0.703in)},
	scale only axis,
	xmin=0,
	xmax=6,
	xlabel style={font=\color{white!15!black}},
	xlabel={$R $},
	xtick={0,1,2,3,4,5,6},
	xticklabels={{$0$},{$1$},{$2$},{$3$},{$4$},{$5$},{$6$}},
	ymin=-0.1,
	ymax=4.1,
	ytick={0,1.19894763639919,4},
	yticklabels={{$0$},{$\mac\brc{\sigma}$},{$4$}},
	ylabel style={at={(-0.03,0.5)},font=\color{white!15!black}},
	ylabel={$\mai\brc{\sigma}$ in [nats]},
	axis background/.style={fill=white},
	legend style={legend cell align=left, align=left, draw=white!15!black}
	]
	
\addplot [color=black, dashed, forget plot]
  table[row sep=crcr]{%
0.2	0.138629436111989\\
0.3	0.207944154167984\\
0.4	0.277258872223978\\
0.5	0.346573590279973\\
0.6	0.415888308335967\\
0.7	0.485203026391962\\
0.8	0.554517744447956\\
0.9	0.623832462503951\\
1	0.693147180559945\\
1.1	0.76246189861594\\
1.2	0.831776616671934\\
1.3	0.901091334727929\\
1.4	0.970406052783923\\
1.5	1.03972077083992\\
1.6	1.10903548889591\\
1.7	1.17835020695191\\
1.8	1.2476649250079\\
1.9	1.3169796430639\\
2	1.38629436111989\\
2.1	1.45560907917589\\
2.2	1.52492379723188\\
2.3	1.59423851528787\\
2.4	1.66355323334387\\
2.5	1.73286795139986\\
2.6	1.80218266945586\\
2.7	1.87149738751185\\
2.8	1.94081210556785\\
2.9	2.01012682362384\\
3	2.07944154167984\\
3.1	2.14875625973583\\
3.2	2.21807097779182\\
3.3	2.28738569584782\\
3.4	2.35670041390381\\
3.5	2.42601513195981\\
3.6	2.4953298500158\\
3.7	2.5646445680718\\
3.8	2.63395928612779\\
3.9	2.70327400418379\\
4	2.77258872223978\\
4.1	2.84190344029578\\
4.2	2.91121815835177\\
4.3	2.98053287640776\\
4.4	3.04984759446376\\
4.5	3.11916231251975\\
4.6	3.18847703057575\\
4.7	3.25779174863174\\
4.8	3.32710646668774\\
4.9	3.39642118474373\\
5	3.46573590279973\\
5.1	3.53505062085572\\
5.2	3.60436533891172\\
5.3	3.67368005696771\\
5.4	3.7429947750237\\
5.5	3.8123094930797\\
5.6	3.88162421113569\\
5.7	3.95093892919169\\
5.8	4.02025364724768\\
5.9	4.08956836530368\\
6	4.15888308335967\\
};

\addplot [color=black, line width = 1.0pt, forget plot]
  table[row sep=crcr]{%
0.2	0.138629436111525\\
0.3	0.207944150627807\\
0.4	0.277258526954961\\
0.5	0.346567808970496\\
0.6	0.415848632511528\\
0.7	0.485040148269841\\
0.8	0.554033191815186\\
0.9	0.622668289330328\\
1	0.690733893987858\\
1.1	0.757954881494403\\
1.2	0.823958271235938\\
1.3	0.888182213208283\\
1.4	0.949579567363882\\
1.5	1.00197393364232\\
1.6	1.02455936361397\\
1.7	1.04076609414429\\
1.8	1.05362585501892\\
1.9	1.06425765235261\\
2	1.07327015988218\\
2.1	1.08104555028116\\
2.2	1.08784383349438\\
2.3	1.09385142925492\\
2.4	1.09920709716665\\
2.5	1.10401709613885\\
2.6	1.10836464372275\\
2.7	1.11231611643251\\
2.8	1.11592527391862\\
2.9	1.11923622806457\\
3	1.12228558449393\\
3.1	1.12510402115474\\
3.2	1.12771747384938\\
3.3	1.13014804110027\\
3.4	1.13241468467488\\
3.5	1.1345337787847\\
3.6	1.13651954551827\\
3.7	1.13838440358917\\
3.8	1.14013925023134\\
3.9	1.14179369097034\\
4	1.14335622834892\\
4.1	1.14483441803538\\
4.2	1.14623499879549\\
4.3	1.147564001359\\
4.4	1.14882684012167\\
4.5	1.15002839079484\\
4.6	1.15117305647928\\
4.7	1.15226482414824\\
4.8	1.15330731314128\\
4.9	1.15430381696885\\
5	1.15525733948963\\
5.1	1.15617062633212\\
5.2	1.1570461922802\\
5.3	1.15788634521937\\
5.4	1.15869320714102\\
5.5	1.15946873262079\\
5.6	1.16021472512102\\
5.7	1.16093285141219\\
5.8	1.16162465436353\\
5.9	1.16229156431518\\
6	1.16293490921312\\
};

\addplot [color=black,dotted, forget plot]
  table[row sep=crcr]{%
0.2	1.19894763639919\\
0.3	1.19894763639919\\
0.4	1.19894763639919\\
0.5	1.19894763639919\\
0.6	1.19894763639919\\
0.7	1.19894763639919\\
0.8	1.19894763639919\\
0.9	1.19894763639919\\
1	1.19894763639919\\
1.1	1.19894763639919\\
1.2	1.19894763639919\\
1.3	1.19894763639919\\
1.4	1.19894763639919\\
1.5	1.19894763639919\\
1.6	1.19894763639919\\
1.7	1.19894763639919\\
1.8	1.19894763639919\\
1.9	1.19894763639919\\
2	1.19894763639919\\
2.1	1.19894763639919\\
2.2	1.19894763639919\\
2.3	1.19894763639919\\
2.4	1.19894763639919\\
2.5	1.19894763639919\\
2.6	1.19894763639919\\
2.7	1.19894763639919\\
2.8	1.19894763639919\\
2.9	1.19894763639919\\
3	1.19894763639919\\
3.1	1.19894763639919\\
3.2	1.19894763639919\\
3.3	1.19894763639919\\
3.4	1.19894763639919\\
3.5	1.19894763639919\\
3.6	1.19894763639919\\
3.7	1.19894763639919\\
3.8	1.19894763639919\\
3.9	1.19894763639919\\
4	1.19894763639919\\
4.1	1.19894763639919\\
4.2	1.19894763639919\\
4.3	1.19894763639919\\
4.4	1.19894763639919\\
4.5	1.19894763639919\\
4.6	1.19894763639919\\
4.7	1.19894763639919\\
4.8	1.19894763639919\\
4.9	1.19894763639919\\
5	1.19894763639919\\
5.1	1.19894763639919\\
5.2	1.19894763639919\\
5.3	1.19894763639919\\
5.4	1.19894763639919\\
5.5	1.19894763639919\\
5.6	1.19894763639919\\
5.7	1.19894763639919\\
5.8	1.19894763639919\\
5.9	1.19894763639919\\
6	1.19894763639919\\
};

\end{axis}
\end{tikzpicture}%
	\end{center}
	\caption{Information rate versus $R$ for the linear field, i.e., $\lambda=1$, and $\sigma^2 = 0.1$. The upper-bound $R\log 2$ is shown by a dashed line. Shannon's limit on the information rate, i.e., the capacity of the Gaussian channel, is further denoted by $\mac\brc{\sigma}$; see the explicit definition in \eqref{eq:C_sig}.}
	\label{fig:3}
\end{figure}

We now consider the extreme cases: We start with the perfect recovery case in which $m^\star = 1$. In this case, we can write
\begin{align}
\mai\brc{\sigma}  =\frac{C_{1}+1}{ 2\sigma^2} - R \log \cosh \brc{  \frac{1}{  R \sigma^2 } }.
\end{align}
Assuming that $\sigma^2$ is small enough, we can further use the approximation
\begin{align}
\log \cosh \brc{  x } \approx \abs{x} -  \log 2,
\end{align}
for large $x$ and conclude that
\begin{align}
	\mai\brc{\sigma} \approx \frac{C_{1}-1}{ 2\sigma^2} + R \log 2.
\end{align}

It is easy to check that $C_{1} \approx 1$, and that $C_{m^\star}$ starts to decrease as $m^\star$ decreases. We hence heuristically conclude that in the perfect recovery region\footnote{Here, $\sigma^2$ is small but fixed. However, by setting $R$ small enough, we can send $C_{m^\star}$ very close to one.}
\begin{align}
	\mai\brc{\sigma} \approx R \log 2. \label{eq:Rate_2_analysis}
\end{align}
This is intuitive, since in this case we have
\begin{align}
 \frac{1}{N}H\brc{\bs\vert \byy , \mav } = 0,
\end{align}
due to perfect recovery, i.e., $m^\star = 1$, and hence the end-to-end model reduces to a revertible noise-free transform for which the information rate is equal to the entropy rate of the labels. 

Now, let us consider another extreme case: The information rate starts to deviate from the $R \log 2$ line, as $m^\star$ starts to decrease. Assume $R\gg 1$ which means that $m^\star = \epsilon$ for some small $\epsilon > 0$; remember Fig.~\ref{fig:2}. In this case, we can write
\begin{align}
\mai\brc{\sigma} = \frac{1}{2} \dbc{\log \brc{1+\frac{1-\epsilon}{\sigma^2}} + \frac{C_{\epsilon}+\epsilon}{ \sigma^2+1-\epsilon}} - R \log \cosh \brc{  \frac{1}{  R \brc{\sigma^2 +1-\epsilon} } }.
\end{align}
Noting that $\log \cosh \brc{  x } \approx 0$  for small $x$, we conclude that
\begin{align}
	\mai\brc{\sigma} \approx \frac{1}{2} \dbc{\log \brc{1+\frac{1-\epsilon}{\sigma^2}} + \frac{C_{\epsilon}+\epsilon}{ \sigma^2+1-\epsilon}}. \label{eq:I_R_large}
\end{align}
By taking the limit $R\uparrow \infty$, we can neglect $\epsilon$ and $C_\epsilon$ and write
\begin{align}
 \lim_{R\uparrow \infty}	\mai\brc{\sigma} = \frac{1}{2} \log \brc{1+\frac{1}{\sigma^2}},
\end{align}
which is also shown in Fig.~\ref{fig:3}. This limit is in fact the information capacity of the Gaussian channel which we denote in the sequel by 
\begin{align}
\mac\brc{\sigma} = \frac{1}{2} \log \brc{1+\frac{1}{\sigma^2}}. \label{eq:C_sig}
\end{align}
This conclusion recovers similar results developed earlier in the context of spectral efficiency analysis for randomly spread code division multiple access techniques; see for example \cite{verdu1999spectral,muller1999power,muller2001multiuser}.

The above finding is further intuitive: By increasing the entropy of the labels, i.e., increasing $R$, the information rate can increase up to the capacity of the end-to-end channel, following Shannon's channel coding theorem \cite{shannon1948mathematical}. By exceeding this limit, a fraction of the labels, i.e., 
\begin{align}
	 \frac{1}{N}H\brc{\bs\vert \mav } - \mac\brc{\sigma} = R \log 2 -  \mac\brc{\sigma} ,
\end{align}
are recovered wrongly, leading to $m^\star < 1$. 

The result in \eqref{eq:I_R_large} further shows a key finding: Since for the linear field, we have\footnote{See Lemma~\ref{lem:1}.} $m^\star \neq 0$, we can conclude that we never reach the upper limit $\mac\brc{\sigma}$ by the linear field. This is in fact the other side of the inference guarantee feature given by a linear model. We summarize the findings as follows:
\begin{conclusion}[Inference via the linear field]
	Observations from the linear field always contain some information about the true labels, i.e., $m^\star \neq 0$. On the other side, perfect inference, i.e., $m^\star \neq 1$,  is achieved at a rate $R^\star$ which is smaller than Shannon's limit, i.e., $R^\star < \mac\brc{\sigma}$.
\end{conclusion} 

\subsection{Quadratic Generative Model}
For nonlinear fields, i.e., $\lambda > 1$, it is straightforward to show that $m=0$ is always a solution of the fixed-point equation \eqref{eq:fix_m}. This means that $\mal_m\brc{\sigma}$ has a local optimum at $m=0$. In general, $m=0$ can be a local minimizer or maximizer. As a result, one conjectures that under some conditions, the overlap can be exactly zero for $\lambda>1$.  This is the extreme case of \textit{independent inference} which is not achievable in a linear model. To validate this conjecture, we start with the special case of purely quadratic Gaussian field in this section, i.e., we consider $\lambda=2$. The following lemma confirms the feasibility of $m^\star = 0$ for the quadratic field:

\begin{lemma}[Independent inference via the quadratic field]
\label{lem:2}
	Consider the purely quadratic field, i.e., $\lambda=2$, and define the threshold load
	\begin{align}
		R_{\mathrm{Th}} =  \frac{2}{\sigma^2 + 1}. \label{eq:Thr_R}
	\end{align}
For this field, the extreme point $m=0$ is a local minimizer of the modified free energy function $\mal_m\brc{\sigma}$, if $R>R_{\rm Th}$ and is a local maximizer if $R< R_{\rm Th}$.
\end{lemma}
\begin{proof}
	By taking the second derivative of $\mal_m\brc{\sigma}$ with respect to $m$ and using the fact that at an extreme point $m$
	\begin{align}
		m  =    \Ex{ \tanh \brc{  \frac{1  + \rho_m\brc{\sigma}Z }{   \rho_m^2\brc{\sigma} } } }{Z}  = \Ex{ \tanh^2 \brc{  \frac{1  + \rho_m\brc{\sigma}Z }{   \rho_m^2\brc{\sigma} } } }{Z},
	\end{align}
one can show that the second derivative at $m=0$ is positive, when $R>R_{\mathrm{Th}} $ and is negative when $R<R_{\mathrm{Th}} $. This concludes the proof.
\end{proof}

Considering the purely quadratic field, Lemma~\ref{lem:2} indicates that by passing the threshold load, $m=0$ becomes a candidate for overlap. We now examine this result by a simple experiment shown in Fig.~\ref{fig:4}. In this figure, the modified free energy term, i.e., $\mal_m\brc{\sigma}$, has been plotted for the purely quadratic field against $m$ at two different loads $R_{ \mathrm{low} } = R_{\mathrm{Th}} - \epsilon$ and $R_{ \mathrm{up} } = R_{\mathrm{Th}} + \epsilon$ for $\epsilon = 0.01$ assuming $\sigma^2 = 0.1$. As the figure shows, the extreme point $m=0$ changes its nature from a maximizer to a minimizer, as we pass the threshold load immediately. It is worth noting that variations on the vertical axis of the figure are in the order of $10^{-8}$. For instance, the lower label on the vertical axis of both sub-figures corresponds to the modified free energy at $m=0$ which for $\sigma^2$ is given by $\mal_0\brc{\sqrt{0.1}} = \mac\brc{ \sqrt{0.1} } = 1.989$. The upper label is only $\delta = 9.9 \times 10^{-8}$ larger than the lower one.

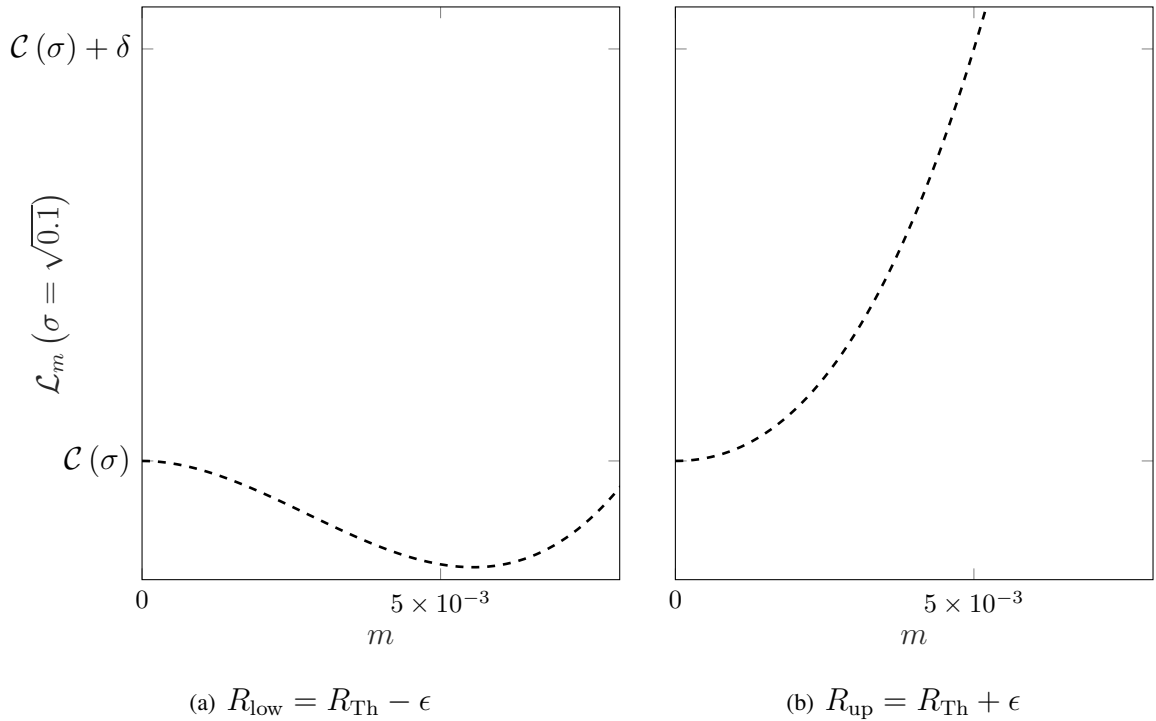
\begin{figure}
\begin{center}
	\subfigure[\normalsize$R_{ \mathrm{low} } = R_{\mathrm{Th}} - \epsilon$]{
			\begin{tikzpicture}
\begin{axis}[%
width=2.5in,
height=3in,
at={(2.6in,1.034in)},
scale only axis,
xmin=0,
xmax=.8,
xlabel style={font=\color{white!15!black}},
xlabel={$m$},
xtick={0,.5},
xticklabels={{\small$0$},{\small$5\times 10^{-3}$}},
ymin=-.3,
ymax=13.8-.3,
ytick={2.56048776209354,9.92693786323071+2.56048776209354},
yticklabels={{$\mac\brc{\sigma}$},{$\mac\brc{\sigma} + \delta$}},
ylabel style={font=\color{white!15!black}},
ylabel={$\mal_m\brc{\sigma = \sqrt{0.1} }$},
axis background/.style={fill=white},
legend style={legend cell align=left, align=left, draw=white!15!black}
]
\addplot [color=black, dashed, line width = 1.0pt, forget plot]
  table[row sep=crcr]{%
0	2.56048776209354\\
0.01	2.55800457298756\\
0.02	2.55067747831345\\
0.03	2.5386901050806\\
0.04	2.52222590148449\\
0.05	2.50146825611591\\
0.06	2.47660034894943\\
0.07	2.44780522584915\\
0.08	2.41526575386524\\
0.09	2.37916475534439\\
0.1	2.33968479931355\\
0.11	2.29700845479965\\
0.12	2.25131797790527\\
0.13	2.20279560983181\\
0.14	2.15162344276905\\
0.15	2.09798334538937\\
0.16	2.04205711185932\\
0.17	1.98402637243271\\
0.18	1.92407266795635\\
0.19	1.86237733066082\\
0.2	1.79912158846855\\
0.21	1.73448649048805\\
0.22	1.66865302622318\\
0.23	1.60180196166039\\
0.24	1.53411397337914\\
0.25	1.46576961874962\\
0.26	1.3969492316246\\
0.27	1.32783308625221\\
0.28	1.25860133767128\\
0.29	1.18943387269974\\
0.3	1.12051057815552\\
0.31	1.05201110243797\\
0.32	0.984115108847618\\
0.33	0.917001903057098\\
0.34	0.850850835442543\\
0.35	0.785841062664986\\
0.36	0.722151562571526\\
0.37	0.659961268305779\\
0.38	0.599448829889297\\
0.39	0.540792897343636\\
0.4	0.484171971678734\\
0.41	0.429764345288277\\
0.42	0.377748221158981\\
0.43	0.328301683068275\\
0.44	0.281602591276169\\
0.45	0.237828820943832\\
0.46	0.197158023715019\\
0.47	0.159767642617226\\
0.48	0.12583515048027\\
0.49	0.0955377370119095\\
0.5	0.0690525621175766\\
0.51	0.0465566068887711\\
0.52	0.028226763010025\\
0.53	0.0142396092414856\\
0.54	0.00477190315723419\\
0.55	0\\
0.56	0.000100240111351013\\
0.57	0.00524882972240448\\
0.58	0.0156217813491821\\
0.59	0.0313950479030609\\
0.6	0.0527443885803223\\
0.61	0.079845517873764\\
0.62	0.112873941659927\\
0.63	0.152005016803741\\
0.64	0.19741402566433\\
0.65	0.249276086688042\\
0.66	0.307766288518906\\
0.67	0.373059391975403\\
0.68	0.445330172777176\\
0.69	0.524753227829933\\
0.7	0.611503094434738\\
0.71	0.705754041671753\\
0.72	0.807680323719978\\
0.73	0.917456045746803\\
0.74	1.03525513410568\\
0.75	1.16125141084194\\
0.76	1.2956186234951\\
0.77	1.43853032588959\\
0.78	1.59015992283821\\
0.79	1.75068074464798\\
0.8	1.92026600241661\\
0.81	2.09908871352673\\
0.82	2.28732182085514\\
0.83	2.48513816297054\\
0.84	2.69271031022072\\
0.85	2.9102108925581\\
0.86	3.13781233131886\\
0.87	3.37568686902523\\
0.88	3.6240066587925\\
0.89	3.88294379413128\\
0.9	4.15267010033131\\
0.91	4.43335741758347\\
0.92	4.72517743706703\\
0.93	5.02830155193806\\
0.94	5.3429013043642\\
0.95	5.66914790868759\\
0.96	6.00721248984337\\
0.97	6.35726614296436\\
0.98	6.71947966516018\\
0.99	7.09402392804623\\
1	7.48106950521469\\
};

\end{axis}
\end{tikzpicture}%
		}
	\subfigure[\normalsize$R_{ \mathrm{up} } = R_{\mathrm{Th}} + \epsilon$]{
		\begin{tikzpicture}
	
	\begin{axis}[%
		width=2.5in,
		height=3in,
		at={(2.6in,1.034in)},
		scale only axis,
		xmin=0,
		xmax=.8,
		xlabel style={font=\color{white!15!black}},
		xlabel={$m$},
		xtick={0,.5},
		xticklabels={{\small$0$},{\small$5\times 10^{-3}$}},
		ymin=-2.86048776209354,
		ymax=13.8-2.86048776209354,
		ytick={0,9.92693786323071},
		yticklabels={{ },{ }},
		ylabel style={font=\color{white!15!black}},
		axis background/.style={fill=white},
		legend style={legend cell align=left, align=left, draw=white!15!black}
		]
		\addplot [color=black, dashed, line width = 1.0pt, forget plot]
		  table[row sep=crcr]{%
		0	0\\
		0.01	0.00251626968383789\\
		0.02	0.0101849883794785\\
		0.03	0.0231857150793076\\
		0.04	0.0416979491710663\\
		0.05	0.065901055932045\\
		0.06	0.0959743112325668\\
		0.07	0.132096707820892\\
		0.08	0.17444720864296\\
		0.09	0.223204627633095\\
		0.1	0.278547629714012\\
		0.11	0.340654745697975\\
		0.12	0.409704402089119\\
		0.13	0.485874831676483\\
		0.14	0.569344148039818\\
		0.15	0.660290375351906\\
		0.16	0.758891344070435\\
		0.17	0.865324735641479\\
		0.18	0.979768246412277\\
		0.19	1.10239923000336\\
		0.2	1.23339505493641\\
		0.21	1.37293283641338\\
		0.22	1.52118968963623\\
		0.23	1.678342461586\\
		0.24	1.8445680141449\\
		0.25	2.02004292607307\\
		0.26	2.20494371652603\\
		0.27	2.39944680035114\\
		0.28	2.60372841358185\\
		0.29	2.81796464323997\\
		0.3	3.04233147203922\\
		0.31	3.27700473368168\\
		0.32	3.52216023206711\\
		0.33	3.77797344326973\\
		0.34	4.04461984336376\\
		0.35	4.3222748041153\\
		0.36	4.611113473773\\
		0.37	4.91131091117859\\
		0.38	5.22304204106331\\
		0.39	5.54648162424564\\
		0.4	5.88180439174175\\
		0.41	6.22918483614922\\
		0.42	6.58879737555981\\
		0.43	6.96081624925137\\
		0.44	7.34541562199593\\
		0.45	7.74276950955391\\
		0.46	8.15305179357529\\
		0.47	8.57643620669842\\
		0.48	9.01309637725353\\
		0.49	9.46320579946041\\
		0.5	9.92693786323071\\
		0.51	10.4044657796621\\
		0.52	10.8959626555443\\
		0.53	11.401601433754\\
		0.54	11.9215550273657\\
		0.55	12.45599617064\\
		0.56	13.0050973594189\\
		0.57	13.5690311193466\\
		0.58	14.1479697972536\\
		0.59	14.7420855760574\\
		0.6	15.3515505641699\\
		0.61	15.9765366911888\\
		0.62	16.6172157973051\\
		0.63	17.2737596184015\\
		0.64	17.9463396370411\\
		0.65	18.6351273953915\\
		0.66	19.3402941823006\\
		0.67	20.0620111823082\\
		0.68	20.8004494905472\\
		0.69	21.5557799786329\\
		0.7	22.3281735926867\\
		0.71	23.1178008913994\\
		0.72	23.9248325526714\\
		0.73	24.7494389563799\\
		0.74	25.5917904376984\\
		0.75	26.4520571827888\\
		0.76	27.3304092735052\\
		0.77	28.2270166575909\\
		0.78	29.1420491337776\\
		0.79	30.0756763964891\\
		0.8	31.0280680209398\\
		0.81	31.9993934631348\\
		0.82	32.9898220747709\\
		0.83	33.9995230138302\\
		0.84	35.0286653637886\\
		0.85	36.077418088913\\
		0.86	37.1459500491619\\
		0.87	38.234429910779\\
		0.88	39.3430262804031\\
		0.89	40.4719076007605\\
		0.9	41.6212422251701\\
		0.91	42.791198387742\\
		0.92	43.9819442182779\\
		0.93	45.1936475932598\\
		0.94	46.4264764636755\\
		0.95	47.6805985271931\\
		0.96	48.9561813771725\\
		0.97	50.2533925324678\\
		0.98	51.5723993480206\\
		0.99	52.913369089365\\
		1	54.276468873024\\
	};
	\end{axis}
\end{tikzpicture}%
	}
\end{center}
	\caption{Modified free energy $\mal_m\brc{\sigma}$ versus $m$ for $\lambda = 2$ considering a small deviation $\epsilon=0.01$ below (figure (a)) and above (figure (b)) the threshold load $R_{\rm Th}$. At $m=0$, $\mal_m\brc{\sigma} = \mac\brc{\sigma = \sqrt{0.1}} = 1.989$. Variations on the vertical axis of the figure are in the order of $10^{-8}$; for instance, the difference between the upper and lower labels on the vertical axis is $\delta = 9.9 \times 10^{-8}$.}
	\label{fig:4}
\end{figure}
%
%

We now plot the modified free energy at multiple loads and take a look on the behavior of the global minimizer. Fig.~\ref{fig:5} shows the global minimizer for $0.8R_{\mathrm{Th}}$ and $1.2 R_{\mathrm{Th}}$. Interestingly, the global minimizer is jumping from $m^\star =1$ to a $m^\star = 0$ quickly. This behavior is depicted more clearly in Fig.~\ref{fig:6}, where we plot the overlap $m^\star$ against the load for the same noise variance $\sigma^2 = 0.1$. The figure shows three sets of results; namely,
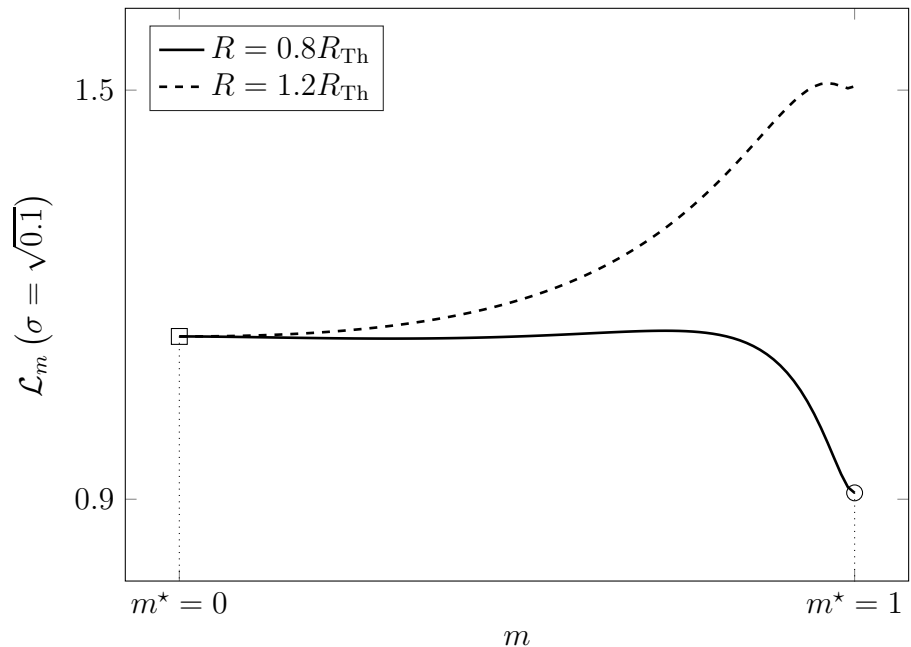
\begin{figure}
	\begin{center}
%
%
\definecolor{mycolor1}{rgb}{0.00000,0.44700,0.74100}%
\definecolor{mycolor2}{rgb}{0.85000,0.32500,0.09800}%
\begin{tikzpicture}

\begin{axis}[%
	width=4.1in,
	height=3in,
	at={(1.011in,0.642in)},
	scale only axis,
	xmin=-0.08,
	xmax=1.08,
	xlabel={$m$},
	xtick={0,1},
	xticklabels={{$m^\star=0$},{$m^\star=1$}},
	ymin=.9,
	ymax=1.6,
	ytick={1,1.5},
	yticklabels={{$0.9$},{$1.5$}},
	ylabel={$\mal_m\brc{\sigma = \sqrt{0.1}}$},
	axis background/.style={fill=white},
	legend style={at={(0.33,.97)},legend cell align=left, align=left, draw=white!15!black}
	]
\addplot [color=black, line width = 1.0pt]
  table[row sep=crcr]{%
0	1.19894763639919\\
0.01	1.19893673502867\\
0.02	1.19890579491368\\
0.03	1.19885728511339\\
0.04	1.19879346171498\\
0.05	1.19871639284863\\
0.06	1.1986279800701\\
0.07	1.19852997679829\\
0.08	1.19842400432945\\
0.09	1.19831156583185\\
0.1	1.1981940586395\\
0.11	1.1980727850998\\
0.12	1.19794896218151\\
0.13	1.19782373001211\\
0.14	1.19769815948418\\
0.15	1.19757325904696\\
0.16	1.1974499807802\\
0.17	1.19732922583226\\
0.18	1.19721184929128\\
0.19	1.19709866454797\\
0.2	1.19699044719925\\
0.21	1.19688793853442\\
0.22	1.1967918486392\\
0.23	1.19670285914658\\
0.24	1.19662162565906\\
0.25	1.19654877986167\\
0.26	1.1964849313414\\
0.27	1.19643066912457\\
0.28	1.19638656294029\\
0.29	1.1963531642144\\
0.3	1.19633100679516\\
0.31	1.1963206074085\\
0.32	1.19632246583721\\
0.33	1.19633706481504\\
0.34	1.19636486962285\\
0.35	1.19640632737043\\
0.36	1.19646186594296\\
0.37	1.196531892587\\
0.38	1.19661679210559\\
0.39	1.19671692462689\\
0.4	1.19683262290432\\
0.41	1.19696418909952\\
0.42	1.19711189099183\\
0.43	1.19727595754887\\
0.44	1.19745657378341\\
0.45	1.19765387481002\\
0.46	1.19786793900242\\
0.47	1.19809878013741\\
0.48	1.19834633839472\\
0.49	1.19861047006242\\
0.5	1.19889093577527\\
0.51	1.19918738708765\\
0.52	1.19949935115275\\
0.53	1.19982621324554\\
0.54	1.20016719682677\\
0.55	1.20052134079917\\
0.56	1.20088747355316\\
0.57	1.20126418333682\\
0.58	1.20164978441193\\
0.59	1.2020422783734\\
0.6	1.20243930990979\\
0.61	1.20283811616772\\
0.62	1.2032354687476\\
0.63	1.20362760720059\\
0.64	1.20401016271256\\
0.65	1.20437807044545\\
0.66	1.20472546875547\\
0.67	1.20504558321456\\
0.68	1.20533059302109\\
0.69	1.20557147699037\\
0.7	1.20575783585929\\
0.71	1.20587768711523\\
0.72	1.20591722796338\\
0.73	1.20586056137567\\
0.74	1.20568937942469\\
0.75	1.20538259731012\\
0.76	1.20491593066667\\
0.77	1.20426140796183\\
0.78	1.20338680915776\\
0.79	1.20225502150868\\
0.8	1.20082330370398\\
0.81	1.19904245104467\\
0.82	1.19685585780243\\
0.83	1.19419847966412\\
0.84	1.19099571141635\\
0.85	1.18716221629454\\
0.86	1.18260077950603\\
0.87	1.17720131887779\\
0.88	1.17084028623626\\
0.89	1.16338086069482\\
0.9	1.15467461508974\\
0.91	1.14456580768341\\
0.92	1.13290026465879\\
0.93	1.1195421888811\\
0.94	1.10440477934235\\
0.95	1.08750507058949\\
0.96	1.06906237148565\\
0.97	1.04967817665729\\
0.98	1.03067774371495\\
0.99	1.01480263892113\\
1	1.00776916666469\\
};
\addlegendentry{$R = 0.8R_{\mathrm{Th}}$}

\addplot [color=black, draw=none, mark=o, mark options={mark size =  3pt, solid, black},forget plot]
  table[row sep=crcr]{%
1	1.00776916666469\\
};

\addplot [color=black, dashed, line width = 1.0pt]
  table[row sep=crcr]{%
0	1.19894763639919\\
0.01	1.19895541980712\\
0.02	1.19897957976809\\
0.03	1.19902128850421\\
0.04	1.19908166923278\\
0.05	1.19916180471973\\
0.06	1.19926274497162\\
0.07	1.19938551419858\\
0.08	1.19953111715729\\
0.09	1.19970054496553\\
0.1	1.19989478046558\\
0.11	1.20011480320251\\
0.12	1.20036159407408\\
0.13	1.20063613970168\\
0.14	1.20093943656538\\
0.15	1.20127249494095\\
0.16	1.20163634267284\\
0.17	1.2020320288129\\
0.18	1.20246062715224\\
0.19	1.2029232396706\\
0.2	1.20342099992573\\
0.21	1.20395507640341\\
0.22	1.20452667584707\\
0.23	1.20513704658496\\
0.24	1.20578748187137\\
0.25	1.20647932325775\\
0.26	1.20721396400872\\
0.27	1.20799285257717\\
0.28	1.20881749615225\\
0.29	1.20968946429355\\
0.3	1.21061039266407\\
0.31	1.21158198687476\\
0.32	1.21260602645231\\
0.33	1.21368436894225\\
0.34	1.21481895415836\\
0.35	1.21601180858964\\
0.36	1.21726504997496\\
0.37	1.21858089205541\\
0.38	1.21996164951342\\
0.39	1.22140974310686\\
0.4	1.22292770500526\\
0.41	1.22451818433371\\
0.42	1.22618395292858\\
0.43	1.22792791130651\\
0.44	1.22975309484582\\
0.45	1.23166268017571\\
0.46	1.23365999176452\\
0.47	1.23574850869301\\
0.48	1.23793187159212\\
0.49	1.24021388971652\\
0.5	1.24259854811555\\
0.51	1.24509001485084\\
0.52	1.24769264819518\\
0.53	1.25041100372862\\
0.54	1.25324984122584\\
0.55	1.25621413120082\\
0.56	1.25930906094162\\
0.57	1.26254003982657\\
0.58	1.26591270366252\\
0.59	1.26943291772389\\
0.6	1.27310677809454\\
0.61	1.276940610821\\
0.62	1.28094096826976\\
0.63	1.28511462193904\\
0.64	1.28946855079967\\
0.65	1.29400992402226\\
0.66	1.29874607667913\\
0.67	1.3036844766757\\
0.68	1.30883268075178\\
0.69	1.31419827687787\\
0.7	1.31978880972895\\
0.71	1.32561168511692\\
0.72	1.33167404825959\\
0.73	1.33798262951017\\
0.74	1.34454354959806\\
0.75	1.35136207446062\\
0.76	1.35844230727255\\
0.77	1.36578680218158\\
0.78	1.37339608038372\\
0.79	1.38126802434272\\
0.8	1.38939711997621\\
0.81	1.39777350929538\\
0.82	1.40638180712955\\
0.83	1.41519962515411\\
0.84	1.42419573473444\\
0.85	1.43332778802368\\
0.86	1.4425395065731\\
0.87	1.45175724337053\\
0.88	1.46088583791303\\
0.89	1.46980373514286\\
0.9	1.47835746868259\\
0.91	1.48635589797225\\
0.92	1.49356520039816\\
0.93	1.49970688349348\\
0.94	1.50446368168243\\
0.95	1.5075035849962\\
0.96	1.50854358858206\\
0.97	1.50749933163658\\
0.98	1.50482239878857\\
0.99	1.50226143219471\\
1	1.50464424517548\\
};
\addlegendentry{$R = 1.2R_{\mathrm{Th}}$}

\addplot [color=black, draw=none, mark=square, mark options={mark size = 3pt, solid, black}]
  table[row sep=crcr]{%
0	1.19894763639919\\
};

\addplot [color=black, dotted, forget plot]
table[row sep=crcr]{%
	0	1.19894763639919\\
	0	.1\\
};

\addplot [color=black, dotted, forget plot]
table[row sep=crcr]{%
	1	1.00776916666469\\
	1	.1\\
};

\end{axis}
\end{tikzpicture}%
	\end{center}
	\caption{Modified free energy $\mal_m\brc{\sigma}$ versus $m$ for $\lambda=2$ at load $0.8R_{\mathrm{Th}}$ (solid line) and $1.2 R_{\mathrm{Th}}$ (dashed line). The overlap jumps from $m^\star = 1$ to $m^\star = 0$ as the threshold load $R_{\mathrm{Th}}$ is passed.}
	\label{fig:5}
\end{figure}

\begin{itemize}
	\item The set of solutions to the fixed-point equations which are shown by the dotted-dashed line. This set includes both minimizers and maximizers of the modified free energy. As the figure shows, below a critical load $R_{\rm Alg}$, the modified free energy has only two extreme points, a maxima at $m=0$ and a minima at $m=1$. By passing $R_{\rm Alg}$, a local minimizer and a local maximizer start to appear at $m_1$ and $m_2$ respectively, where $0<m_1,m_2<1$. As the load grows, $m_1$ moves towards zero and meets $m=0$ as the rate passes the threshold load $R_{\rm Th}$, at which the curvature of the curve at $m=0$ changes and $m=0$ becomes a local minimizer.
	\begin{figure}
		\begin{center}
			\input{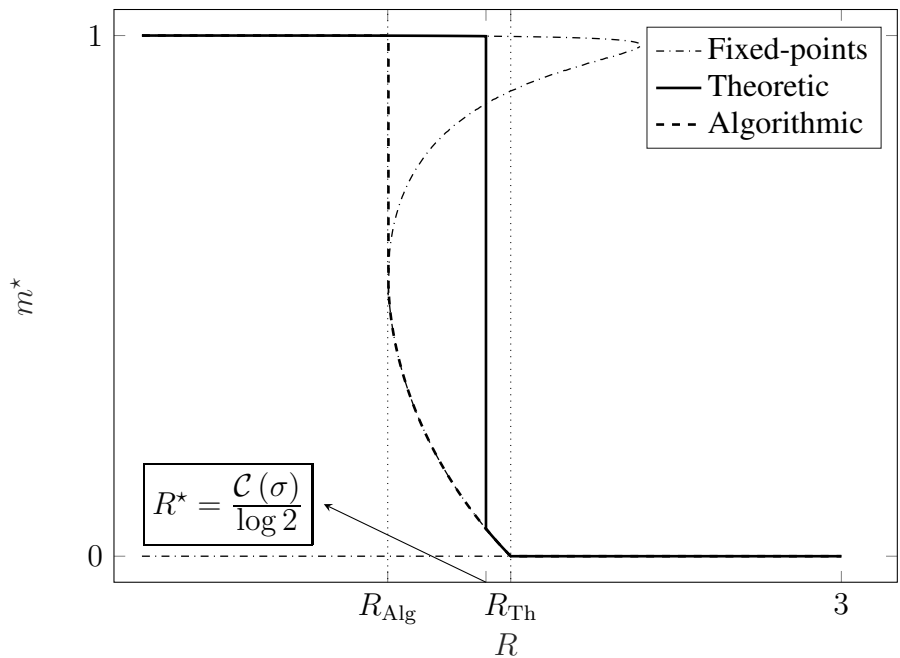}
		\end{center}
		\caption{Overlap against load for the purely quadratic field, i.e., $\lambda=2$. The dotted-dashed line shows the fixed-point solutions, the solid line represents the \ac{rs} solution given by the global minimizer of the modified free energy, and the dashed line denotes the overlap which can be achieved via the approximate message passing algorithm.}
		\label{fig:6}
	\end{figure}
	\item The second set of results show the theoretically achievable overlap as a function of the load $R$ which corresponds to the global minimizer of the modified free energy and is shown by the solid line. Here, the figure shows a \textit{first-order phase transition} at the load
	\begin{align}
		R^\star = \frac{\mac\brc{\sigma}}{\log 2} \label{eq:R_star}
	\end{align}
which is in fact the channel capacity given in bits. A \textit{second-order phase transition} is further observed at $R_{\rm Th}$. To confirm the validity of this experimental conclusion, we further plot the maximum convergence time to a local extreme point against the load in Fig.~\ref{fig:7}. In this figure, we iteratively solve the fixed-point equation \eqref{eq:fix_m}. The algorithm starts from an initial point $m_0\in\brc{0,1}$ and iterates till it converges to a solution. This solution is a local extreme point. For a given load, we sweep $m_0$ on a fine grid with 100 points on the interval $\brc{0,1}$ and find the maximum number of iterations required to converge, where we maximize over the grid. As the convergence criteria, we stop iterating at the iteration in which the value of $m$ is deviated from its value in the previous iteration less than $10^{-10}$, i.e., $\abs{m_t - m_{t-1}}<10^{-10}$ with $m_t$ denoting the value of $m$ in iteration $t$. 
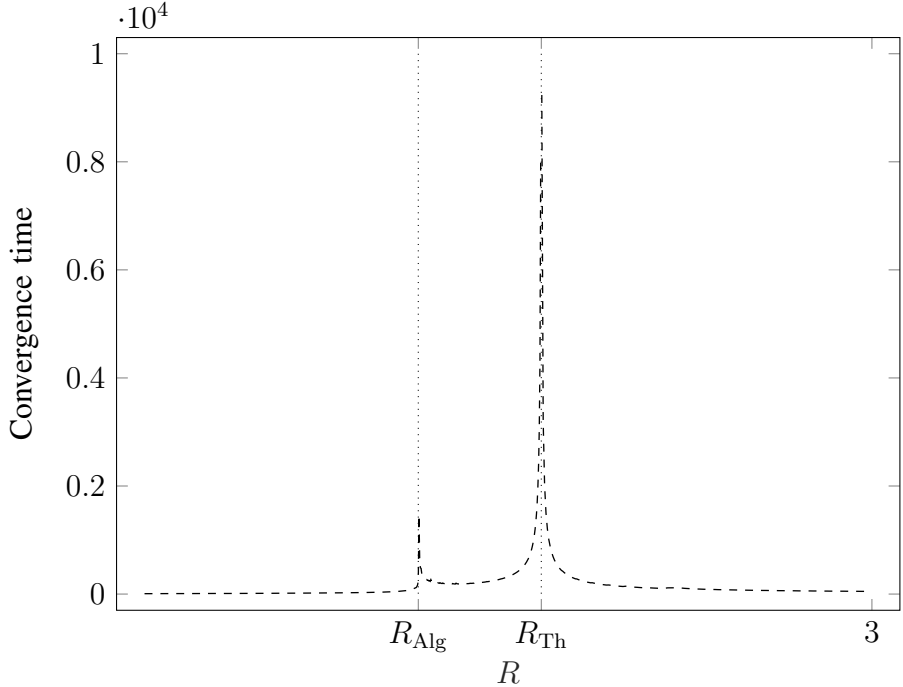
\begin{figure}
	\begin{center}
%
%
\definecolor{mycolor1}{rgb}{0.00000,0.44700,0.74100}%
\begin{tikzpicture}

\begin{axis}[%
	width=4.1in,
height=3in,
at={(1.262in,0.703in)},
scale only axis,
xmin=0.3,
xmax=3.1,
xlabel style={font=\color{white!15!black}},
xlabel={$R$},
xtick={0,1.3785,1.81818181818182,3},
xticklabels={{$0$},{$R_{\rm Alg}$},{$R_{\mathrm{Th}}$},{$3$}},
ymin=-300,
ymax=10300,
ylabel={Convergence time},
ytick={0,2000,4000,6000,8000,10000},
yticklabels={{$0$},{$0.2$},{$0.4$},{$0.6$},{$0.8$},{$1$}},
axis background/.style={fill=white},
legend style={legend cell align=left, align=left, draw=white!15!black}
]
\addplot [color=black, dashed, line width = .5pt, forget plot]
  table[row sep=crcr]{%
1.82	9236\\
1.825	3185\\
1.83	2016\\
1.835	1499\\
1.84	1203\\
1.845	1010\\
1.85	880\\
1.855	777\\
1.86	696\\
1.865	632\\
1.87	580\\
1.875	536\\
1.88	499\\
1.885	467\\
1.89	439\\
1.895	422\\
1.9	398\\
1.905	378\\
1.91	360\\
1.915	344\\
1.92	329\\
1.925	316\\
1.93	304\\
1.935	293\\
1.94	282\\
1.945	280\\
1.95	269\\
1.955	260\\
1.96	252\\
1.965	244\\
1.97	237\\
1.975	230\\
1.98	224\\
1.985	218\\
1.99	213\\
1.995	208\\
2	210\\
2	210\\
2.01	198\\
2.02	189\\
2.03	181\\
2.04	174\\
2.05	167\\
2.06	168\\
2.07	160\\
2.08	154\\
2.09	149\\
2.1	144\\
2.11	140\\
2.12	145\\
2.13	137\\
2.14	132\\
2.15	128\\
2.16	125\\
2.17	121\\
2.18	118\\
2.19	123\\
2.2	118\\
2.21	114\\
2.22	111\\
2.23	109\\
2.24	106\\
2.25	114\\
2.26	108\\
2.27	115\\
2.31	117\\
2.32	111\\
2.33	106\\
2.34	103\\
2.35	100\\
2.36	97\\
2.37	95\\
2.38	93\\
2.39	91\\
2.4	89\\
2.41	88\\
2.42	86\\
2.43	85\\
2.44	84\\
2.45	82\\
2.46	81\\
2.47	80\\
2.48	79\\
2.49	78\\
2.5	77\\
2.51	76\\
2.52	75\\
2.53	74\\
2.54	73\\
2.55	72\\
2.56	71\\
2.57	70\\
2.58	70\\
2.59	69\\
2.6	68\\
2.61	67\\
2.62	67\\
2.63	66\\
2.64	65\\
2.65	65\\
2.66	64\\
2.67	63\\
2.68	63\\
2.69	62\\
2.7	62\\
2.71	61\\
2.72	61\\
2.73	60\\
2.74	59\\
2.75	59\\
2.76	58\\
2.77	58\\
2.78	58\\
2.79	57\\
2.8	57\\
2.81	56\\
2.82	56\\
2.83	55\\
2.84	55\\
2.85	54\\
2.86	54\\
2.87	54\\
2.88	53\\
2.89	53\\
2.9	53\\
2.91	52\\
2.92	52\\
2.93	51\\
2.94	51\\
2.95	51\\
2.96	50\\
2.97	50\\
2.98	50\\
2.99	49\\
3	49\\
};

\addplot [color=black, dashed, line width = 0.5pt, forget plot]
  table[row sep=crcr]{%
0.4	7\\
0.45	8\\
0.5	9\\
0.55	9\\
0.6	10\\
0.65	11\\
0.7	11\\
0.75	12\\
0.8	13\\
0.85	14\\
0.9	16\\
0.95	17\\
1	19\\
1.05	21\\
1.1	23\\
1.15	26\\
1.2	30\\
1.25	36\\
1.3	46\\
1.35	72\\
1.371	121\\
1.3735	139\\
1.376	169\\
1.3785	242\\
1.381	1464\\
1.3835	647\\
1.386	487\\
1.3885	475\\
1.391	376\\
1.3935	349\\
1.396	332\\
1.3985	326\\
1.401	294\\
1.4035	290\\
1.406	271\\
1.4085	270\\
1.411	255\\
1.4135	257\\
1.416	244\\
1.4185	249\\
1.421	236\\
1.4235	266\\
1.426	231\\
1.4285	223\\
1.431	234\\
1.4335	220\\
1.436	214\\
1.4385	226\\
1.441	213\\
1.4435	208\\
1.446	204\\
1.4485	209\\
1.451	204\\
1.4535	200\\
1.456	210\\
1.4585	202\\
1.461	198\\
1.4635	195\\
1.466	204\\
1.4685	197\\
1.471	194\\
1.4735	191\\
1.476	205\\
1.4785	195\\
1.481	192\\
1.4835	190\\
1.486	188\\
1.4885	197\\
1.491	192\\
1.4935	189\\
1.496	187\\
1.4985	186\\
1.501	195\\
1.5035	191\\
1.506	188\\
1.5085	187\\
1.511	186\\
1.5135	201\\
1.516	192\\
1.5185	190\\
1.521	188\\
1.5235	187\\
1.526	186\\
1.5285	199\\
1.531	193\\
1.5335	191\\
1.536	190\\
1.5385	189\\
1.541	189\\
1.5435	190\\
1.546	196\\
1.5485	194\\
1.551	193\\
1.5535	193\\
1.556	193\\
1.5585	194\\
1.561	195\\
1.5635	200\\
1.566	198\\
1.5685	198\\
1.571	199\\
1.5735	200\\
1.576	201\\
1.5785	202\\
1.581	209\\
1.5835	205\\
1.586	205\\
1.5885	206\\
1.591	207\\
1.5935	209\\
1.596	210\\
1.5985	211\\
1.601	219\\
1.6035	215\\
1.606	216\\
1.6085	217\\
1.611	219\\
1.6135	220\\
1.616	222\\
1.6185	224\\
1.621	226\\
1.6235	231\\
1.626	229\\
1.6285	231\\
1.631	233\\
1.6335	235\\
1.636	237\\
1.6385	240\\
1.641	242\\
1.6435	244\\
1.646	247\\
1.6485	249\\
1.651	251\\
1.6535	254\\
1.656	257\\
1.6585	260\\
1.661	262\\
1.6635	265\\
1.666	268\\
1.6685	272\\
1.671	275\\
1.6735	280\\
1.676	282\\
1.6785	285\\
1.681	289\\
1.6835	293\\
1.686	297\\
1.6885	301\\
1.691	305\\
1.6935	310\\
1.696	315\\
1.6985	319\\
1.701	324\\
1.7035	330\\
1.706	335\\
1.7085	341\\
1.711	347\\
1.7135	353\\
1.716	359\\
1.7185	366\\
1.721	373\\
1.7235	381\\
1.726	388\\
1.7285	397\\
1.731	405\\
1.7335	414\\
1.736	424\\
1.7385	434\\
1.741	445\\
1.7435	456\\
1.746	468\\
1.7485	480\\
1.751	494\\
1.7535	509\\
1.756	524\\
1.7585	541\\
1.761	558\\
1.7635	578\\
1.766	598\\
1.7685	623\\
1.771	646\\
1.7735	674\\
1.776	706\\
1.7785	742\\
1.781	782\\
1.7835	828\\
1.786	880\\
1.7885	941\\
1.791	1012\\
1.7935	1097\\
1.796	1199\\
1.7985	1324\\
1.801	1483\\
1.8035	1690\\
1.806	1984\\
1.8085	2396\\
1.811	3064\\
1.8135	4347\\
1.816	8013\\
};

\addplot [color=black, dotted, forget plot]
table[row sep=crcr]{%
	1.81818181818182	10000\\
	1.81818181818182	0\\
};

\addplot [color=black, dotted, forget plot]
table[row sep=crcr]{%
	1.3785	10000\\
	1.3785	0\\
};

\end{axis}
\end{tikzpicture}%
	\end{center}
	\caption{Convergence time against load $R$ for the purely quadratic field, i.e., $\lambda=2$. The result shows a critical slow-down at $R_{\rm Th}$. This indicates a second-order phase transition at $R_{\rm Th}$ and a first-order phase transition at $R^\star$.}
	\label{fig:7}
\end{figure}
The result shows a \textit{critical slowing down} at load $R_{\rm Th}$. As the figure shows, the runtime shows a small spike at $R_{\rm Alg}$ where the local extreme points in the middle of the interval $\dbc{0,1}$ start to appear. The other spike occurs at $R_{\mathrm{Th}}$, where the curvature at $m=0$ changes. In the considered fine grid, the algorithm does not converge at $R_{\mathrm{Th}}$. This observation reports the well-known phenomenon of \textit{critical slowing down} which is a sign for a first-order phase transition at $R^\star$ and a second-order phase transition at $R_{\rm Th}$. In fact, the latter is due to the zero curvature of the curve which leads to an infinite number of iterations for convergence and describes a second-order phase transition. However, the former phase transition corresponds to this phenomenon that an already-existing local minimum passes the local minimum at $m=1$ and becomes the global minimum. More details on the critical slowing down phenomenon and its connection to the phase transition can be found in \cite{scheffer2009early,lade2012early}.
\item The last set of the results are represented by the dashed line in the figure and predict the overlap which is achieved by the an \ac{amp} algorithm \cite{bandeira2018notes,zdeborova2016statistical,aubin2020exact}. \ac{amp} algorithms are a family of first-order algorithms designed to efficiently approximate the posterior marginals, and have been shown to be optimal among all first order methods for certain classes of problems \cite{Celentano2020}. Due to this property, \ac{amp} algorithms provide a useful benchmark to probe the computational hardness of a problem. The asymptotic performance od \ac{amp} algorithms can be exactly traced by a set of scalar \emph{state evolution equations} \cite{javanmard2013state}, which are intimately connected with the expected cross entropy: the state evolution equations coincide with the fixed-point equations associated with the decoupled problem. Therefore, the overlap achieved by \ac{amp} from a uninformed initialization is given by the minima of the cross entropy which is closest to $m=0$. As the results show, in our problem, there is a load interval in which the closest minima to $m=0$ is a local minimum, showing that there exists a gap between the theoretically achievable overlap\footnote{That is given by the global minimum of the modified free energy.}, and the one achieved by \ac{amp} from a random initialization. This gap is often referred to as the \textit{algorithmic gap} and indicates the set of problem settings in which perfect inference is information-theoretically feasible however computationally hard to achieve.
\end{itemize}

Unlike the linear model, the \ac{rs} solution in the purely quadratic model shows both first-order and second-order \textit{phase transitions}. For this particular setting, a first-order phase transition occurs at Shannon's limit and a second-order phase transition happens at higher loads. These phase transitions divide the behavior of system into three states:
\begin{enumerate}
	\item The state which occurs at $R<R^\star$. In this state \textit{all} labels are recovered (perfect recovery).
	\item The second state is the state of \textit{some} labels being inferred which occurs at $R^\star \leq R \leq R_{\rm Th}$ (weak or partial recovery).
	\item Finally, the state which happens as the load passes the threshold $R_{\mathrm{Th}}$.  We call this new state, the state of \textit{nothing} being recovered in which the Bayesian algorithm results in \textit{zero-information learning}. This new state does not exist in linear models, as our conclusions illustrate that $m^\star\neq 0$ for the linear field. 
\end{enumerate}

Although the first-order phase transition is generic, the second-order might not be present. This is simply observed by comparing the two critical loads, namely $R^\star$ and $R_{\rm Th}$. In fact, it is easily shown that there exists a critical noise variance $\sigma_{\rm C}^2$ below which $R_{\rm Th} < R^\star$ and hence we observe only a first-order phase transition at which the overlap jumps from \textit{perfect inference}, i.e., $m^\star = 1$, to \textit{zero-information inference}, i.e., $m^\star = 0$. This critical noise variance for the considered setting in Fig.~\ref{fig:6} is $\sigma_{\rm C}^2 = 0.084$. For settings below this critical variance, the inference problem shows the so-called \textit{all-or-nothing} phenomenon\footnote{We discuss this phenomenon in details on the next subsection.} \cite{reeves2019all}. As we show later on, with higher order fields, the inference problem always show the \textit{all-or-nothing} behavior, regardless of the noise variance.

The transition from \textit{all-to-something} behavior, i.e., $\sigma^2 > \sigma_{\rm C}^2$, to \textit{all-or-nothing} case is demonstrated in Figs.~\ref{fig:6_algorithmic} and \ref{fig:6_algorithmic_2}, where we plot the theoretic and algorithmic results for $\sigma^2 = 0.2 > \sigma_{\rm C}^2$ and $\sigma^2 = 0.05 < \sigma_{\rm C}^2$, respectively. Similar to Fig.~\ref{fig:6}, the fixed-point solutions, the \ac{rs} solution and the algorithmically achievable overlap are shown by dotted-dashed, solid and dashed lines respectively. The algorithmic gap is also denoted by $\Delta_{\rm Alg}$ in the figures.  By comparing Figs.~\ref{fig:6_algorithmic} and \ref{fig:6_algorithmic_2}, one can further observe that the \textit{all-or-nothing} behavior\footnote{As we show in Section~\ref{sec:SecLearn}, this behavior is favorable for secure learning.} in this case is achieved at the expense of higher algorithmic gaps. Further investigations conjecture that this behavior is generic.

\begin{figure}
	\begin{center}
		\input{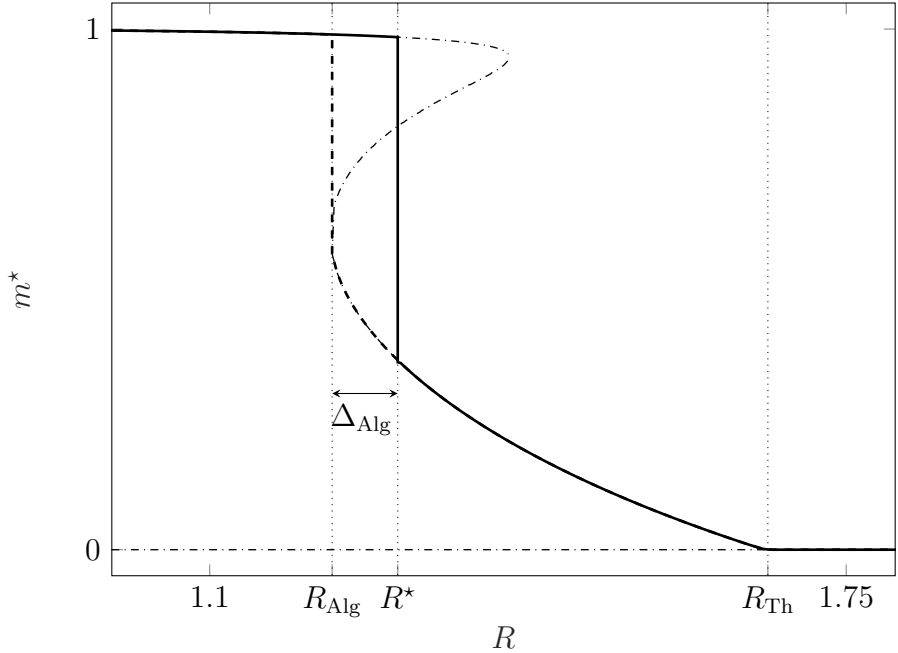}
	\end{center}
	\caption{Overlap $m^\star$ against load $R$ for the purely quadratic field, i.e., $\lambda=2$ considering $\sigma = 0.2 > \sigma_{\rm C}^2$. The dotted-dashed, solid and dashed lines represent the fixed-points, the \ac{rs} solution and the algorithmically feasible overlap respectively. The \ac{rs} solution in this case shows a first-order and a second-order phase transmission.}
	\label{fig:6_algorithmic}
\end{figure}

\begin{figure}
	\begin{center}
		\input{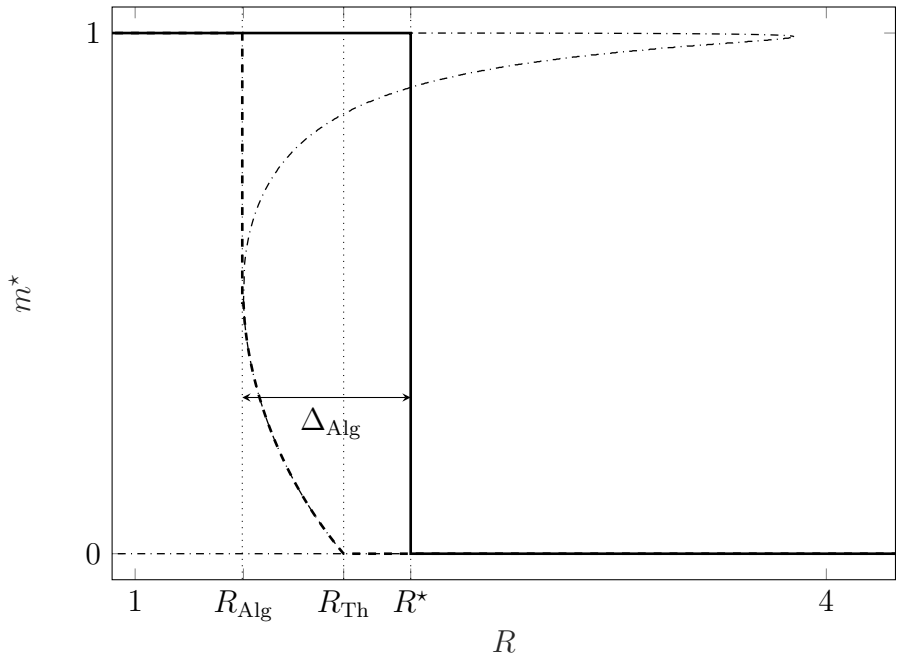}
	\end{center}
	\caption{Overlap $m^\star$ against load $R$ for the purely quadratic field, i.e., $\lambda=2$ considering $\sigma = 0.05 < \sigma_{\rm C}^2$. The dotted-dashed, solid and dashed lines represent the fixed-points, the \ac{rs} solution and the algorithmically feasible overlap respectively. The \ac{rs} solution in this case shows only a first-order phase transmission.}
	\label{fig:6_algorithmic_2}
\end{figure}

Although $R_{\mathrm{Th}}$ and $R^\star$ are the ones which are derived analytically, we are also interested on the critical load $R_{\rm Alg}$  from algorithmic viewpoint. In fact, $R_{\rm Alg}$ represents the load at which we transit from the initial state of \textit{perfect learning} of labels to the state of \textit{erroneous learning}, when we employ the approximate message passing algorithm to implement the Bayesian estimator. More precisely, it represents the maximum load at which perfect learning is feasible computationally. It is worth noting that, unlike the theoretical results, the algorithmic results always show the three states of \textit{all}, \textit{some} and \textit{nothing} being recovered. This is in fact the cost we pay to achieve the computationally tractability.

A summary of the findings for $\lambda=2$ is given below:
\begin{conclusion}[Inference on a quadratic generative model]
	Observations from a purely quadratic generative model show a first-order and a second-order phase transition with respect to the load: Below Shannon's limit, perfect inference is achieved. As we pass this limit, a first-order phase transition occurs, where the overlap jumps from $m^\star = 1$ to some $m^\star < 1$. At the threshold load $R_{\rm Th}$, the overlap shows a second-order phase transition and becomes exactly zero. This means that by surpassing the threshold load, the learned labels contain no information about the true labels anymore. The two critical loads, i.e., Shannon's limit and the threshold, become closer as the noise variance reduces. As the noise variance drops below a critical level, the overlap shows all-or-nothing behavior at Shannon's limit. The algorithmic performance however demonstrates always the \textit{all-to-something-to-nothing} behavior.
\end{conclusion} 

\subsection{Higher-order Nonlinear Generative Models}
For higher-order fields, i.e., $\lambda \geq 3$, the detailed analytic investigations of the first-order phase shift is not trivial. Nevertheless, the numerical investigations and intuition behind the problem suggest the following conjecture: For larger choices of $\lambda$, the problem shows the same behavior with three major changes
\begin{enumerate}
	\item The threshold load, at which the curvature of $\mal_m\brc{\sigma}$ at $m=0$ changes, reduces as $\lambda$ increases. The analytical derivation of the threshold load for higher-order fields is not trivial; however, numerical investigation suggests that either the threshold load drops significantly for $\lambda\geq 3$ or it is exactly zero, i.e., the curvature change happens at $R_{\mathrm{Th}} = \epsilon$, where either $\epsilon\approx0$ or $\epsilon = 0$ for higher-order fields. Our finite resolution searches failed to find this load for $\lambda\geq 3$. This means that for higher fields $m=0$ is (almost) always a local minimum.
	\item Since $R_{\rm Th} \rightarrow 0$, there exists no setting at which the Bayesian algorithms gets into the state with \textit{some} labels being correctly learned. This means that at higher-orders, the overlap jumps at the critical load $R^\star$ directly from \textit{all recovered} state, i.e., $m^\star = 1$, to the state of \textit{nothing} being recovered, i.e., $m^\star = 0$. In other words, the Bayesian inference always shows the \textit{all-or-nothing} phenomenon for $\lambda \geq 3$; the behavior which we see for $\sigma^2 \leq \sigma_{\rm C}^2$ with the purely quadratic field. The all-or-nothing phenomenon is a well-known asymptotic behavior being reported in various other inference problems; see for example \cite{gamarnik2017high,reeves2019all,zadik2019computational,niles2021all,barbier2020all,maillard2020phase}.
	\item The overlap achieved by AMP from uninformed initialization tends to be zero for almost all loads. This indicates that a computationally feasible implementation of the Bayesian algorithm via the approximate message passing is unable to perform a perfect learning, even though it is theoretically achievable. A known solution to this issue is to spatially couple the random field elements \cite{krzakala2012statistical,donoho2013information}. Investigations in this respect are, however, out of the scope of this study and are left as a potential direction for future work.
\end{enumerate}

The above discussions are summarized with some heuristics as follows: For $\lambda\geq 3$, the modified free energy $\mal_m\brc{\sigma}$ contains three local extreme points: $m=0$, $m=m_0$ for $0< m_0<1$ and $m=1$. The boundary points $m=0$ and $m=1$ are local minima, and $m=m_0$ is a local maximum. Starting with small loads, the local minimum at $m=1$ is smaller than the one at $m=0$ and hence the overlap reads $m^\star = 1$. By increasing the rate, $\mal_1\brc{\sigma}$ starts to increase while $\mal_0\brc{\sigma} = \mac\brc{\sigma}$ remains unchanged. As we surpass a critical load $\mal_1\brc{\sigma} >\mal_0\brc{\sigma}$ and hence $m=0$ becomes the global minimizer which concludes $m^\star = 0$. From the channel coding theorem of Shannon and our observations with the purely quadratic field, we expect that the critical load be $R^\star$ defined in \eqref{eq:R_star}. Numerical investigations confirm this conjecture; see Figs.\ref{fig:7b} and \ref{fig:8}. We give a heuristic proof to this conjecture later on. 

The all-or-nothing phenomenon, along with a numerical validation of our conjecture, is illustrated in Fig.~\ref{fig:7b} where we plot the modified free energy $\mal_m\brc{\sigma}$ against $m$ for small deviations below and above the conjectured critical load $R^\star=\mac\brc{\sigma}/\log 2$, i.e., we set $R = R^\star - \epsilon$ (solid line) and $R = R^\star + \epsilon$ (dashed line) for $\epsilon = 0.01$. As the figure shows by surpassing the critical load, $\mal_1\brc{\sigma} $ passes $\mac\brc{\sigma} = \mal_0\brc{\sigma}$, and hence the global minimum jumps from $m^\star = 1$ to $m^\star = 0$. This jump is shown from another perspective in Fig.~\ref{fig:8}, in which we plot the overlap against the load for $\lambda=3$ and $5$. As the figure shows, the phase transition happens at $\mac\brc{\sigma}/\log 2$ for both of them, and the curves are not numerically distinguishable. This result further explains the findings by Sourlas in \cite{sourlas1989spin}.

\begin{figure}
	\begin{center}
		\input{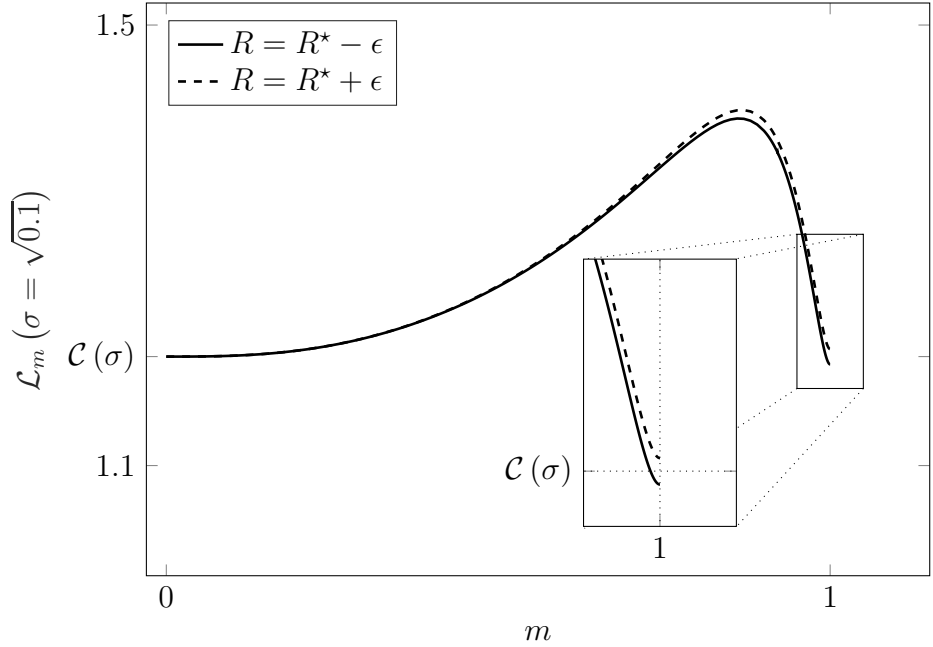}
	\end{center}
	\caption{Modified free energy $\mal_m\brc{\sigma}$ against $m$ for $\lambda=3$ considering $\sigma^2={0.1}$. The figure is plotted for $R=R^\star - \epsilon$ (solid line) $R=R^\star + \epsilon$ (dashed line), where we set $R^\star = \mac\brc{\sigma}/\log 2$ and $\epsilon= 0.01$. At the critical load, the values at $m=0$ and $m=1$ coincide up to numerical errors.}
	\label{fig:7b}
\end{figure}

\begin{remark}
	\label{remark:2}
	The Gaussian random field can be observed as the encoder of the random coding scheme introduced by Sourlas in \cite{sourlas1989spin}. A related side-note in this respect is the dimensional scaling of Sourlas' coding: Similar to Shannon's random coding approach, Sourlas' coding directly maps the message into a real-valued sequence\footnote{This is in contrast to known coding schemes whose operations are performed in a finite field.}. On the other hand, Fig.~\ref{fig:8} indicates that Sourlas' coding achieves the channel capacity with a pure cubic Gaussian random field. This observation demonstrates a significant drop in the number of random components in the coding scheme: Unlike Shannon's random coding, in which a codebook with an exponentially large number of random entries is required, this code requires a generator function whose number of random components scales only cubically with the message length. This is an interesting finding and is worth further investigations. 
\end{remark}

\begin{figure}
	\begin{center}
		\input{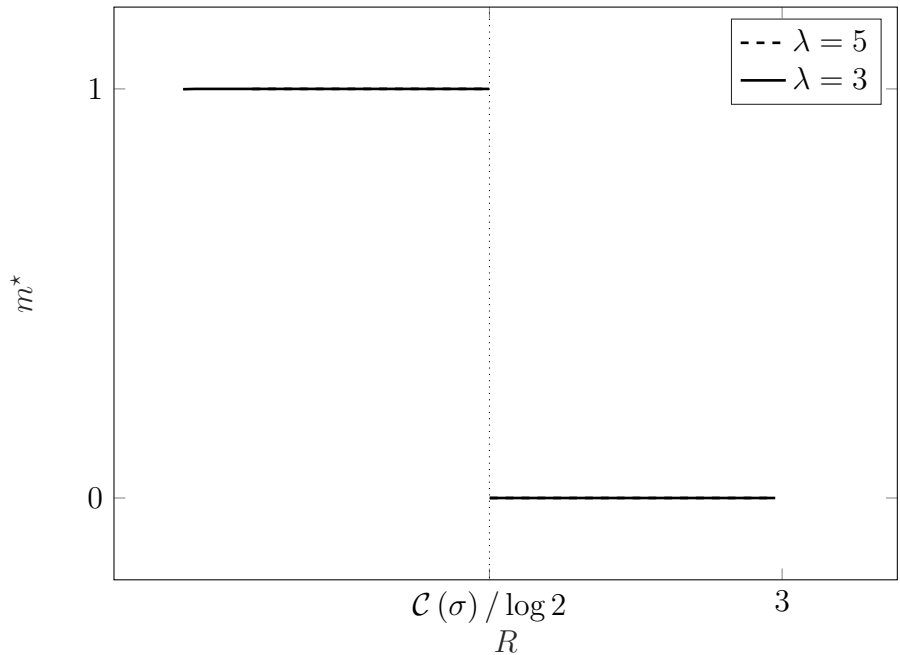}
	\end{center}
	\caption{Overlap $m^\star$ against load $R$ for $\lambda=3$ and $\lambda=5$. The results show a first-order phase transmission at the channel capacity from $m^\star = 1$ to $m^\star = 0$. This demonstrates the all-or-nothing phenomenon in higher order models. As one can observe, the curves for $\lambda=3$ and $\lambda=5$ are not numerically distinguishable.}
	\label{fig:8}
\end{figure}

We now focus on the information rate achieved by higher-order fields and try to derive heuristically the critical load: With higher order-fields, we roughly have two possible candidates for the global minimum; namely, $m=1$ and $m = 0$. For the former case, we can follow the same lines of justifications as in \eqref{eq:Rate_1_analysis}-\eqref{eq:Rate_2_analysis} to show that
\begin{align}
	\mal_1\brc{\sigma} \approx R \log 2. 
\end{align}
For the latter one, it is straightforwardly shown by substituting $m=0$ into \eqref{eq:L_m} that
\begin{align}
	\mal_0\brc{\sigma} = \mac\brc{\sigma}.
\end{align}
The global optimum is hence found by comparing these two values: We have $m^\star = 1$, if 
\begin{align}
	\mal_1\brc{\sigma} \leq  \mal_0\brc{\sigma}.
\end{align}
Using the approximated value for $\mal_1\brc{\sigma}$, we can conclude that $m^\star = 1$, if 
\begin{align}
R \leq \frac{\mac\brc{\sigma}}{\log 2}. 
\end{align}
Consequently, $m^\star = 0$, if $R$ exceeds this critical load.

Based on this heuristic justification, as well as the numerical validations, we have the following conjecture for higher-order Gaussian fields:

\begin{conjecture}
There exists an order $\lambda_0\geq 3$, such that for $\lambda \geq \lambda_0$, the overlap reads
\begin{align}
	m^\star = \begin{cases}
		1 &R\leq  \dfrac{\displaystyle \mac\brc{\sigma}}{\log 2} \vspace{5mm}\\
		0 &R >  \dfrac{\displaystyle \mac\brc{\sigma}}{\log 2}\\
	\end{cases} ,
	\label{eq:Conj1}
\end{align}
and the information rate is given by
\begin{align}
	\mai\brc{\sigma} = \begin{cases}
		R \log 2 &R\leq  \dfrac{\displaystyle \mac\brc{\sigma}}{\log 2}\vspace{5mm}\\
		 \mac\brc{\sigma} &R >  \dfrac{\displaystyle \mac\brc{\sigma}}{\log 2}\\
	\end{cases} .
\label{eq:Conj}
\end{align}
\end{conjecture}
The evidence for \eqref{eq:Conj1} is simply what we observe in Fig~\ref{fig:8}. To examine \eqref{eq:Conj}, we further sketch $\mai\brc{\sigma}$ against the load for $\lambda=3$ in Fig.~\ref{fig:9}. The figure suggests that $\lambda_0 =3$ in the conjecture, since the curve for $\lambda = 3$ meets the conjectured curve almost perfectly.

\begin{figure}
	\begin{center}
		\input{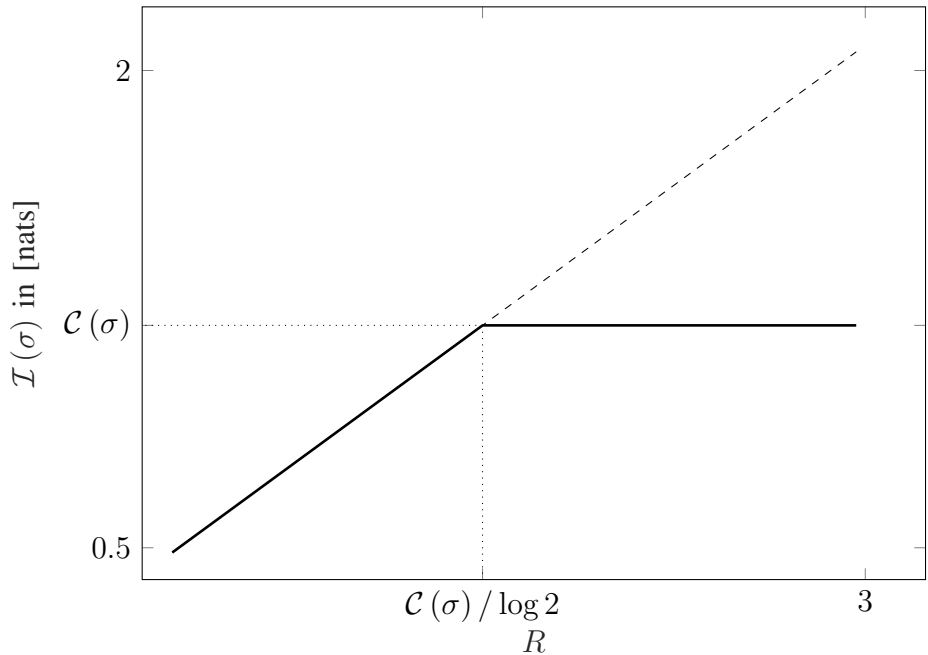}
	\end{center}
	\caption{Information rate against load for $\lambda=3$. Although $\lambda=3$, the figure tightly matches the conjecture.}
	\label{fig:9}
\end{figure}

We summarize the findings for $\lambda\geq 3$ as follows:
\begin{conclusion}[Inference on higher-order generative models]
	Bayesian estimation for nonlinear generative models of higher-order, i.e., $\lambda\geq 3$ experiences the all-or-nothing phenomenon: Below a critical load, perfect inference, i.e., $m^\star = 1$,  is achieved. By surpassing the critical load, the learned labels contain no information about the true labels anymore, i.e., the overlap jumps from $m^\star = 1$ to $m^\star = 0$. The critical load meets Shannon's limit. Nevertheless, for these fields, perfect inference cannot be achieved via the feasible realization of the Bayesian learning algorithm with approximate message passing.
\end{conclusion} 

\section{Comments on Secure Learning}
\label{sec:SecLearn}
Our findings lead us to an interesting conclusion: In settings in which we need a phase transition in the outcome of the learning algorithm, we can use nonlinear models. Instances of such settings are those with secrecy constraints. In these settings, the samples in $\byy$ are observed by two terminals: a legitimate terminal and a malicious one. The desired behavior in this setting is to enable the legitimate terminal to learn the labels perfectly while the malicious one receives zero information from its eavesdropping. 

The secure transmission setting in its generic form is studied by Wyner in \cite{wyner1975wire}. Wyner finds a maximum load for secure transmission, similar to the one given by Shannon for reliable transmission. The limit indicates that in a \textit{wiretap} setting, the maximum load for which a \textit{perfectly} secure transmission is achievable is given by the difference between Shannon's limits for the legitimate and the eavesdropping channels. A short introduction to Wyner's results and its extension to Gaussian channels is given in Appendix~\ref{app:wyner}.

Considering the conclusions in the previous section, one can conjecture that the same behavior as the one explained by Wyner is observed in our setting when the generative model is \textit{strictly nonlinear}: In this case, it is enough to degrade the channel to the malicious terminal, such that its corresponding threshold load falls below the system load. This guarantees that the overlap of this terminal is $m^\star = 0$ which corresponds to independent inference. 

In this section, we try to show that this conjecture is, in fact, correct. This leads to an interesting conclusion: For secure inference, it is enough to use higher-order Gaussian random fields. This not only re-discovers Wyner's result, but also specifies explicitly the model by which secure inference is enabled.

We start our investigations by considering a simple wiretap setting for our learning problem: A sequence of labels $\bs\in\set{\pm 1}^K$ is observed through a nonlinear generative Gaussian model $\mav\brc{\cdot}$ of order $\lambda$. The transformed sequence $\bxx = \mav\brc{\bs} \in \setR^N$ is then observed by two terminals which know the model, i.e., $\mav \brc{\cdot}$: A legitimate terminal  indexed by $i=1$ and a malicious terminal indexed by $i=2$. Terminal $i \in \set{1,2}$ observes
\begin{align}
	\byy_i = \bxx + \bww_i, \label{eq:y_i}
\end{align}
where $\bww_i\sim \mac\man\brc{\boldsymbol{0},\sigma^2_i \mI_N}$. To guarantee the feasibility of secure inference\footnote{Wyner's result indicates that with $\sigma_1^2<\sigma_2^2$, secure transmission is information-theoretically impossible.}, we set $\sigma_2^2 > \sigma_1^2$. Terminal $i$ employs the optimal Bayesian algorithm to recover the labels. For this setting, we adopt the notion of \textit{perfect} secrecy:
\begin{definition}[Perfectly secure and reliable inference]
The inference is said to be asymptotically reliable and secure if for $N\uparrow\infty$, we have
\begin{itemize}
	\item The overlap of the legitimate terminal is $m^\star_1 = 1$.
	\item The information rate to the malicious terminal is zero, i.e.,
	\begin{align}
		\lim_{N\uparrow \infty} \frac{1}{N} I\brc{\bs ; \byy_2 \vert \mav} = 0.
	\end{align}
\end{itemize}
\end{definition}
In the above definition, the first constraint restricts the legitimate terminal to recover the labels asymptotically perfectly. The second constraint moreover enforces the information leakage to the malicious terminal to be zero.

\subsection{Secure Inference via Nonlinear Generative Models}
\label{sec:Serecy1}
We now consider $\mav\brc{\cdot}$ to be a Gaussian random field with 
\begin{align}
	\Phi \brc{x} = x^\lambda.
\end{align}
Hence, the transformed sequence $\bxx$ contains unit-variance Gaussian symbols. We further assume that $\lambda$ is chosen such that \eqref{eq:Conj1} and \eqref{eq:Conj} are exact. To secure our learning model, we follow the strategy below:
\begin{itemize}
	\item Generate $K_2 = \lceil N \mac\brc{\sigma_2}/\log 2 \rceil$ \ac{iid} uniform labels 
	\begin{align}
		\br = \dbc{r_1, \ldots,r_{K_2}}^\trp.
	\end{align}
\item Add $\br$ as a prefix to $\bs$ and permute the extended vector via a randomly-generated permutation matrix $\mPi$, i.e., construct the new sequence of labels $\tilde{\bs}\in\set{\pm 1}^{K+K_2}$ as
	\begin{align}
	\tilde{\bs} = \mPi\dbc{\br^\trp, \bs^\trp}^\trp,
\end{align}
for a permutation matrix $\mPi\in \set{0,1}^{\brc{{K+K_2}} \times \brc{K+K_2} }$ which is known to all parties in the setting, and whose permuted vector $\tilde{\bs}$ has the following property: For $\ell \in [K_2]$, define the $\ell$-th \textit{bin}\footnote{We use this appellation due to its connection to the \textit{random linear binning} approach; see \cite{wyner1974recent}.} of $\tilde{\bs}$ as
\begin{align}
	\dbc{\tilde{\bs}}_\ell =  \set{ \tilde{s}_k: k\in \dbc{\brc{\ell-1} B +1 : \ell B } },
\end{align}
for the bin size 
\begin{align}
	B = \left\lceil1+\frac{K}{K_2} \right\rceil.
\end{align}
Permutation $\mPi$ guarantees that for every choice of $\ell \in [K_2]$, there exists only one symbol of $\br$ in $\dbc{\tilde{\bs}}_\ell$. 
\end{itemize} 

The randomized permutation applied by $\mPi$ can be observed as the random binning approach; see \cite{wyner1974recent,cover1975proof,yassaee2014achievability,muramatsu2012construction} for more details on the random binning technique and its applications in information theory. Interestingly, the described property of the random permutation implies that there existence \textit{exponentially} many permutations.

To secure the model, we pass $\tilde{\bs}$ through the model $\mav\brc{\cdot}: \set{\pm 1}^{K+K_2} \mapsto \setR^N$. The effective load of this model is defined by excluding the randomly generated labels, i.e., $\br$. We distinguish between the effective load and the system load by denoting the former with $R = K/N$ and the latter by $\tilde{R} = \brc{K+K_2}/N$. Note that $R$ and $\tilde{R}$ are related as
\begin{align}
	\tilde{R} = R + \frac{\mac\brc{\sigma_2}}{\log 2} + \epsilon_N \label{eq:R_tild}
\end{align} 
for some $\epsilon_N = \mao\brc{1/N}$ which tends to zero in the asymptotic regime.

We claim that using this simple modification, the inference in this model is perfectly secure and reliable in the large-system limit if the effective load satisfies
\begin{align}
R < \frac{\mac\brc{\sigma_1} - \mac\brc{\sigma_2}}{\log 2}. \label{eq:Wyner_bound}
\end{align} 
This upper-bound is in fact Wyner's asymptotic limit for the Gaussian wiretap channel \cite{leung1978gaussian}. 

In the sequel, we give a heuristic proof to this claim. We start the analysis by investigating the security constraint: By the chain-rule, we have
\begin{subequations}
	\begin{align}
I\brc{\tilde{\bs} ; \byy_2 \vert \mav} &= I\brc{\br,\bs ; \byy_2 \vert \mav},\\
&= I\brc{\bs ; \byy_2 \vert \mav} + I\brc{\br ; \byy_2 \vert \bs ,  \mav}.
\end{align}
\end{subequations}
Hence, the information leakage to the eavesdropper is given by
\begin{align}
	I\brc{\bs ; \byy_2 \vert \mav}  &= I\brc{\tilde{\bs} ; \byy_2 \vert \mav}-  I\brc{\br ; \byy_2 \vert \bs ,  \mav}. \label{eq:chain}
\end{align}
Noting that the system load $\tilde{R} \geq \mac\brc{\sigma_2}/\log 2$, we use \eqref{eq:Conj} and write 
\begin{align}
	\lim_{N\uparrow \infty} \frac{1}{N} I\brc{\tilde{\bs} ; \byy_2 \vert \mav} = \mac\brc{\sigma_2}.
\end{align}
This limit calculates the asymptotic information rate corresponding to the first term on the right hand side of \eqref{eq:chain}.

To calculate the information rate corresponding to the second term in \eqref{eq:chain}, we note that $I\brc{\br ; \byy_2 \vert \bs ,  \mav}$ calculates the mutual information in a modified model, where the recovery algorithm tries to recover $\br$ while the true labels in $\bs$ are revealed to it: Let $\bs = \bs_0$ be known at both sides. Since $\bs_0$ is known at the recovery algorithm, each bin of $\tilde{\bs}$, which is constructed from $\bs_0$ and $\br$ via the random permutation $\mPi$, can be seen as a uniform Bernoulli symbol\footnote{Note that the bins are discrete random variables, and hence the information rate only depends on their distribution.}. Since we have $K_2$ bins in this setting, the load in this new setting reads $K_2/N = \mac\brc{\sigma_2} / \log 2$, and thus we have
\begin{align}
	\lim_{N\uparrow \infty} \frac{1}{N} I\brc{\tilde{\bs} ; \byy_2 \vert \bs = \bs_0 ,  \mav} = \mac\brc{\sigma_2}.
\end{align}
As the right hand side does not depend on $\bs_0$, we can conclude that
\begin{align}
	\lim_{N\uparrow \infty} \frac{1}{N} I\brc{\tilde{\bs} ; \byy_2 \vert \bs  ,  \mav} = \mac\brc{\sigma_2}.
\end{align} 
Finally, using the fact that $I\brc{\br ; \byy_2 \vert \bs ,  \mav} = I\brc{\tilde{\bs} ; \byy_2 \vert \bs ,  \mav}$, we have
\begin{align}
	\lim_{N\uparrow \infty} \frac{1}{N} I\brc{\br ; \byy_2 \vert \bs ,  \mav} = \lim_{N\uparrow \infty} \frac{1}{N} I\brc{\tilde{\bs} ; \byy_2 \vert \bs ,  \mav} =  \mac\brc{\sigma_2}.
\end{align}
This concludes that 
\begin{align}
	\lim_{N\uparrow \infty} \frac{1}{N} I\brc{\bs ; \byy_2 \vert \mav} =0,
\end{align}
which means that the information rate to the malicious terminal is zero for any choice of $R$.

For reliable inference, the legitimate terminal needs to infer perfectly all labels of $\tilde{\bs}$ and remove the first $K_2$ labels. From the conjecture in \eqref{eq:Conj1}, we know that this is achievable, if
\begin{align}
	\tilde{R} \leq \frac{\mac\brc{\sigma_1}}{\log 2}. \label{eq:Upp2}
\end{align}
By plugging in \eqref{eq:R_tild} into \eqref{eq:Upp2} and sending $\epsilon_N$ to zero for asymptotically large $N$, we conclude that perfect inference is achieved if \eqref{eq:Wyner_bound} is satisfied. Hence, \eqref{eq:Wyner_bound} guarantees secure and reliable inference.

\begin{remark}
	Following Remark~\ref{remark:2}, one can have the same conclusion here: Our proposed secure coding scheme achieves Wyner's limit with a random code whose generator function is given by a Gaussian field. This is an alternative random coding approach whose number of randomly-generated components grows at most cubically with the message length. 
\end{remark}

\subsection{Connecting Fyodorov's Results to Wyner's}
\label{sec:Secrecy2}
Although the model considered in \cite{fyodorov2019spin} is different from what we study, related behaviors are reported in both studies. Fyodorov reports a second-order phase transition with respect to $\sigma^2$ for fields which do not contain linear terms. This phase transition breaks the replica symmetry of the problem and hence is specified via the full \ac{rsb} solution. This symmetry breaking comes from the fact that the performance of the \ac{map} algorithm, used for learning by Fyodorov, is described by the free energy at zero temperature at which the symmetry of the corresponding spin glass can be broken\footnote{For the linear field, the corresponding spin glass shows \ac{rs}, even at zero temperature.}.

Given the fact that our considered learning algorithm is Bayesian optimal\footnote{And hence, it is information-theoretically optimal too.}, such an asymmetry does not occur.  We however observe phase transition with respect to the system setting, i.e., load. From the information-theoretic viewpoint, this is the same behavior: The input dimension, i.e., dimension of $\bs$, in our setting is controlled by the load\footnote{Since we have a finite-alphabet input.}, while it is controlled by the \ac{snr} in Fyodorov's work\footnote{Since the input in \cite{fyodorov2019spin} is continuous.}. Our final conclusion is hence similar to what Fyodorov concludes: In strictly nonlinear Gaussian models, we can achieve perfect secrecy.

Our approach sketches a direct connection between Wyner's result and Gaussian random fields, and recovers the same limit for the Gaussian model. We now try to connect Fyodorov's result to Wyner's limit by investigating the same wiretap model. This investigation not only clarifies the connection between the two results, but also introduces an alternative secure coding approach. An advantage of this alternative approach is that it is more intuitive from practical viewpoints; however, as we will see, a downside is that it does not achieve Wyner's limit. We finally discuss the key sub-optimal point of this approach and explain the solution to this point which can be considered in future work.

Let us now start again with the wiretap setting described at the beginning of this section and specified by \eqref{eq:y_i}. We intend to transfer securely $K = NR$  labels to terminal $1$ while keeping terminal $2$ blind. Let $\bs\in\set{\pm 1 }^K$ denote the vector of labels. We follow these steps:
\begin{itemize}
	\item We first apply \textit{spherical source coding}: In this step, we choose $2^{K}$ distinct points on a scaled $\brc{D-1}$-sphere, i.e., we find $\setC = \set{\tilde{\bs}_1, \ldots,\tilde{\bs}_{2^K}}$ on the $\brc{D-1}$-dimensional hypersphere with radius $\sqrt{D P_S}$. We denote the minimum normalized distance between distinct points by $2d$; this means
	\begin{align}
		2d = \min_{ i\neq j \in \dbc{2^K} }  \frac{\norm{\tilde{\bs}_i - \tilde{\bs}_j}^2}{D P_S}.
	\end{align}
The value of $d$ is related to the \textit{maximum normalized inner product} $p \in \dbc{0,1}$ as $d = 1-p$, where $p$ is given by \cite{sloane1982recent}
	\begin{align}
		p = \max_{ i\neq j \in \dbc{2^K} } \frac{\inner{\tilde{\bs}_i ; \tilde{\bs}_j } }{ P_S}.
	\end{align}
To each outcome of $\bs$, we allocate one distinct point $\tilde{\bs}_k$ and refer to it as the \textit{codeword}.
	\item In the second step, we select the codeword which corresponds to the given realization of $\bs$ and pass it through an order $\lambda$ Gaussian field $\mav\brc\cdot : \setS^{D-1} \mapsto\setR^N$. We then add artificial noise to it: Let the selected codeword be $\tilde{\bs}^\star \in \setC$, the secure transformed sequence is
		\begin{align}
			\bxx = \mav\brc{\tilde{\bs}^\star } + \bww_0,
		\end{align}
	where $\bww_0 \sim\man \brc{\boldsymbol{0}, \theta^2 \mI_N}$ for some noise variance $\theta^2$.
\end{itemize}

Terminal $i$ observes $\byy_i$ described in \eqref{eq:y_i}. To learn the labels, it uses the \ac{map} algorithm\footnote{In this particular case, the \ac{map} algorithm reduces to the least-squares method}: Terminal $i$ finds point $\hat{\bs}_i$ on the hypersphere whose transform by $\mav\brc{\cdot}$ is closest to its observations, i.e., it finds 
	\begin{align}
	\hat{\bs}_i &= \argmin_{ \frac{\bu}{\sqrt{{DP_S}}} \; \in \; \setS^{D-1}} \;  \norm{\byy_i - \mav\brc{\buu} }^2.  
	\end{align}
It then decodes the labels by finding the codeword in $\setC$ which is closest to $\hat{\bs}_i$.

To achieve perfect secrecy and reliability, we follow these arguments:
\begin{itemize}
	\item If the overlap between the recovery $\hat{\bs}_i$ and the codeword, i.e., 
	\begin{align}
			m_i = \frac{\inner{\hat{\bs}_i ; \tilde{\bs}^\star } }{ P_S},
	\end{align}
	is zero, the recovery and the codeword are completely uncorrelated. This leads to an error probability of $1$.
	\item If the overlap is larger than $1-d/2$, the closest codeword to the recovery is the codeword. This leads to error probability of $0$.
\end{itemize}
Hence, to guarantee secure and reliable transform, we need to design the parameters such that the asymptotic overlap between $\hat{\bs}_1$ and the codeword, denoted by $m_1$, is larger than $1-d/2$, and the asymptotic overlap between $\hat{\bs}_2$ and the codeword, denoted by $m_2$, is zero. It is worth mentioning that these conclusions are given by a slight intuition whose proof of rigor is skipped for the sake of simplicity and can be considered in more detailed studies.

From Fyodorov's result, we can calculate the overlap explicitly. As a result, for reliability, we limit the load $R = K/N$ and $P_S$, such that the minimum distance of codewords is bounded by $2\brc{1-m_1}$ from above, and  for security we set $\theta^2$ and $P_S$, such that $m_2$ passes the phase transition and satisfies $m_2=0$. To this end, we consider the following results:
\begin{itemize}
	\item Finding the maximum load $R$ for a specific minimum distance reduces to the well-known problem of \textit{spherical covering} which is still an open problem. We hence approximate it with an available tight bound in the literature: We assume that Kabatiansky-Levenshtein's bound is tight in the large-system limit, and approximate maximum load as \cite{sloane1982recent}
\begin{align}
	R_{\max} \approx \mu \brc{ \frac{1+ q}{2 q} \log_2 \frac{1+ q}{2 q} - \frac{1- q}{2 q} \log_2 \frac{1- q}{2 q} } \label{eq:R_max}
\end{align}
where\footnote{Note that $p$ depends on $P_S$.} $q = \sqrt{1-p^2}$ setting $p=2m_1-1$, and $\mu = D/N$.
\item For the purely quadratic field with covariance function $\Phi\brc{x} = x^2/2$, the full \ac{rsb} solution for the asymptotic overlap between the recovery and codeword is given by Result~\ref{result:1} given in Appendix~\ref{app:Fyodorov}. We hence assume $\mav\brc{\cdot}$ is the purely quadratic field with the same covariance.
\item Considering the codeword $\tilde{\bs}^\star$, the power of the transformed symbol is given by
\begin{align}
	P_{\rm T} = \Ex{\xx_n^2}{ } =  \Phi\brc{\inner{\tilde{\bs}^\star;\tilde{\bs}^\star}}+ \theta^2 = \frac{P_S^2}{2} + \theta^2.
\end{align}
As a result, Wyner's limit on maximum load for secure and reliable inference is
\begin{align}
	C_{\rm W} = \frac{1}{2} \dbc{\log_2 \brc{1+ \frac{P_S^2}{2\sigma_1^2} + \frac{\theta^2}{\sigma_1^2}}
- \log_2 \brc{1+ \frac{P_S^2}{2\sigma_2^2} + \frac{\theta^2}{\sigma_2^2}}	
}.
\end{align}
Note that this bound is achievable when we use the strategy explained in Section~\ref{sec:Serecy1}.
\end{itemize}

Fig.~\ref{fig:10} shows the value of $R_{\max}$ against $\log\brc{ 1/ \sigma_1^2 }$ when we set $P_S =\sigma_2^2 = 1$. In this figure, for each choice of $\sigma_1^2$, we find $\theta^2$ such that $m_2 = 0$, i.e., we set
\begin{align}
	\theta^2 = \frac{P_S}{ \Gamma_{\rm Th} } - \sigma_2^2,
\end{align}
following Result~\ref{result:1}. We then use Result~\ref{result:1} once again to determine $m_1$ by setting
\begin{align}
	\Gamma = \frac{P_S}{\sigma_1^2+\theta^2}.
\end{align}
To fulfill the reliability constraint, we  set $p=2m_1-1$ and determine $R_{\max}$ from \eqref{eq:R_max}. Finally, $R_{\max}$ is maximized numerically over $\mu \leq 1$. For sake of comparison, Wyner's limit $C_{\rm W}$ is further plotted in the figure by a dashed line.
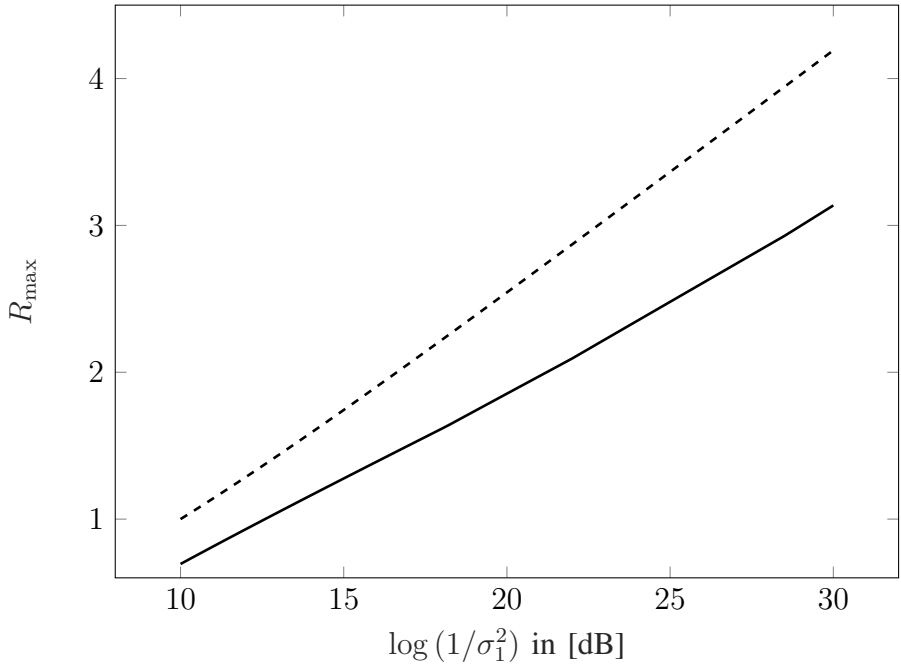
\begin{figure}
	\begin{center}
%
%
\definecolor{mycolor1}{rgb}{0.00000,0.44700,0.74100}%
\begin{tikzpicture}

	\begin{axis}[%
	width=4.1in,
	height=3in,
	at={(1.262in,0.703in)},
	scale only axis,
	xmin=8,
	xmax=32,
	xlabel style={font=\color{white!15!black}},
	xlabel={$\log \brc{1/\sigma_1^2}$ in [dB]},
	xtick={10,15,20,25,30},
	xticklabels={{$10$},{$15$},{$20$},{$25$},{$30$}},
	ymin=.6,
	ymax=4.5,
	ylabel style={font=\color{white!15!black}},
	ylabel={$R_{\max}$},
	ytick={1,2,3,4},
	yticklabels={{$1$},{$2$},{$3$},{$4$}},
	axis background/.style={fill=white},
	legend style={legend cell align=left, align=left, draw=white!15!black}
	]
	
\addplot [color=black, line width = 1.0pt]
  table[row sep=crcr]{%
10	0.694342685118314\\
10.25	0.724404294728116\\
10.5	0.754327703905179\\
10.75	0.784120864679408\\
11	0.813791337106472\\
11.25	0.843346287027928\\
11.5	0.872792491210356\\
11.75	0.90213634744298\\
12	0.931383887864004\\
12.25	0.960540794279953\\
12.5	0.989612414597402\\
12.75	1.01860377974286\\
13	1.04751962063222\\
13.25	1.07636438488602\\
13.5	1.10514225308448\\
13.75	1.13385715442753\\
14	1.16251278171634\\
14.25	1.1911126056101\\
14.5	1.21965988813836\\
14.75	1.24815769546804\\
15	1.27660890993742\\
15.25	1.30501624137829\\
15.5	1.3333822377535\\
15.75	1.36170929514052\\
16	1.38999966709385\\
16.25	1.41825547341979\\
16.5	1.4464787083972\\
16.75	1.47467124847705\\
17	1.50283485949285\\
17.25	1.53097120341252\\
17.5	1.5590818446609\\
17.75	1.58716825604073\\
18	1.61523182427809\\
18.25	1.64405038436795\\
18.5	1.67416563630832\\
18.75	1.70425511619527\\
19	1.73432029880161\\
19.25	1.76436257290383\\
19.5	1.79438324645476\\
19.75	1.82438355143093\\
20	1.85436464837604\\
20.25	1.88432763066062\\
20.5	1.91427352847652\\
20.75	1.94420331258372\\
21	1.97411789782557\\
21.25	2.00401814642768\\
21.5	2.03390487109444\\
21.75	2.06377883791637\\
22	2.09383637890675\\
22.25	2.12605204714285\\
22.5	2.15825227812939\\
22.75	2.19043794695972\\
23	2.22260987853508\\
23.25	2.25476885050444\\
23.5	2.28691559602607\\
23.75	2.31905080636208\\
24	2.35117513331671\\
24.25	2.38328919152811\\
24.5	2.41539356062304\\
24.75	2.44748878724286\\
25	2.47957538694916\\
25.25	2.51165384601617\\
25.5	2.54372462311728\\
25.75	2.57578815091195\\
26	2.60784483753936\\
26.25	2.63989506802425\\
26.5	2.67193920560062\\
26.75	2.70397759295788\\
27	2.73601055341451\\
27.25	2.76803839202344\\
27.5	2.80006139661313\\
27.75	2.83207983876846\\
28	2.86409397475479\\
28.25	2.89610404638873\\
28.5	2.92811028185871\\
28.75	2.96249677982437\\
29	2.99719384054545\\
29.25	3.03188565117636\\
29.5	3.06657250668082\\
29.75	3.10125468538397\\
30	3.13593244991815\\
};

\addplot [color=black, dashed, line width = 1.0pt]
  table[row sep=crcr]{%
10	1\\
10.25	1.03476729650486\\
10.5	1.06985395500065\\
10.75	1.10524754284901\\
11	1.14093586044323\\
11.25	1.1769069618527\\
11.5	1.21314917264112\\
11.75	1.24965110498802\\
12	1.28640167025596\\
12.25	1.32339008915584\\
12.5	1.36060589966883\\
12.75	1.39803896288742\\
13	1.4356794669392\\
13.25	1.47351792915611\\
13.5	1.51154519664926\\
13.75	1.5497524454449\\
14	1.58813117833191\\
14.25	1.62667322156462\\
14.5	1.6653707205573\\
14.75	1.70421613469935\\
15	1.7432022314117\\
15.25	1.78232207955676\\
15.5	1.82156904230612\\
15.75	1.86093676956153\\
16	1.90041919001697\\
16.25	1.94001050294136\\
16.5	1.97970516975432\\
16.75	2.01949790545973\\
17	2.05938366999526\\
17.25	2.09935765954961\\
17.5	2.13941529789289\\
17.75	2.17955222776041\\
18	2.21976430232457\\
18.25	2.26004757678495\\
18.5	2.30039830010226\\
18.75	2.34081290689796\\
19	2.38128800953755\\
19.25	2.42182039041226\\
19.5	2.46240699443126\\
19.75	2.50304492173336\\
20	2.54373142062517\\
20.25	2.58446388075066\\
20.5	2.62523982649486\\
20.75	2.66605691062317\\
21	2.7069129081563\\
21.25	2.74780571047941\\
21.5	2.78873331968323\\
21.75	2.82969384313395\\
22	2.87068548826785\\
22.25	2.91170655760619\\
22.5	2.95275544398509\\
22.75	2.99383062599504\\
23	3.034930663624\\
23.25	3.07605419409816\\
23.5	3.11719992791392\\
23.75	3.1583666450547\\
24	3.19955319138611\\
24.25	3.24075847522303\\
24.5	3.28198146406188\\
24.75	3.32322118147188\\
25	3.36447670413873\\
25.25	3.40574715905453\\
25.5	3.44703172084767\\
25.75	3.48832960924687\\
26	3.52964008667318\\
26.25	3.57096245595456\\
26.5	3.61229605815721\\
26.75	3.6536402705285\\
27	3.69499450454609\\
27.25	3.73635820406847\\
27.5	3.77773084358181\\
27.75	3.81911192653878\\
28	3.86050098378454\\
28.25	3.90189757206593\\
28.5	3.94330127261953\\
28.75	3.98471168983467\\
29	4.02612844998783\\
29.25	4.06755120004448\\
29.5	4.10897960652522\\
29.75	4.15041335443284\\
30	4.19185214623703\\
};

\end{axis}
\end{tikzpicture}%
	\end{center}
	\caption{Maximum load for secure and reliable inference achieved by Fyodorov's result (solid line) compared with Wyner's limit $C_{\rm W}$ (dashed line). We set $\sigma_2^2 = 1$ and $P_S=1$.}
	\label{fig:10}
\end{figure}

As the figure shows, the achieved maximum load does not meet Wyner's limit. This follows the fact that the considered strategy is not information-theoretically optimal in various respects. We finalize our discussions in this section by pointing out few final remarks:
\begin{itemize}
	\item The parameters in Fig.~\ref{fig:10} are chosen such that the approximation in \eqref{eq:R_max} is tight. It is however worth mentioning that the right hand side of \eqref{eq:R_max} is in general an upper bound on $R_{\max}$. This means that the given approximation is slightly optimistic.
	\item There is a key source of sub-optimality in the strategy considered in this section: The spherical coding technique which maps the labels on the hypersphere uses spherical covering. This approach is known to be sub-optimal \cite{mackay2003information}. From the literature, we believe that the gap in Fig.~\ref{fig:10} is mainly due to this source of sub-optimality. In fact, it is widely known in the coding theory that minimum distance is not the key player in codebook design \cite{mackay2003information,richardson2008modern}. We therefore believe that the gap between the two curves shrinks by replacing the spherical covering codebook with a denser codebook\footnote{This could be simply done by random coding on hypersphere; see for example the recent work in \cite{muller2021approximation}.}. We leave this direction of study as a potential future work.
\end{itemize}

\section{Conclusions and Final Remarks}
Strictly nonlinear generative models show a first-order phase transition with respect to the load: At a threshold load, the overlap jumps from one to zero and the information rate changes from a linear function of load to a constant. This phase transition leads to the \textit{all-or-nothing} phenomenon for \textit{strictly} nonlinear fields. Numerical investigations show almost perfect match to the conjectured asymptotic behavior for orders more than two. This is a close match to the asymptotic behavior of random coding schemes described by the channel coding theorem of Shannon which cannot be seen by the linear field. This observation can be illustrated through the connection of our model to Sourlas' coding scheme. A key finding in this respect is that using a Gaussian random field, the channel capacity can be achieved by a random code whose generator function is described by a Gaussian field. The number of random components in this code grows in terms of the message length with a polynomial order. 

In strictly nonlinear models, the overlap can be set to zero. This is another behavior which is not observed in linear models. This enables us to secure the model from passive attacks in which malicious terminals try to learn the model parameters from overheard data. The heuristic derivation of maximum load, by which perfect secrecy is guaranteed, shows that achieved maximum load for large-order Gaussian fields matches the information-theoretic limit derived by Wyner. Hence, it proves Wyner's secure coding theorem via an alternative coding scheme whose encoder is a Gaussian random field. This finding depicts the efficiency of non-linear fields for secure coding and suggests further investigations in this respect through finite dimensional analyses and simulations. To this end, one can develop \ac{amp} algorithms for Bayesian inference on nonlinear generative Gaussian models. This is an interesting perspective  for future work.

The findings of this study are shown to be connected to what Fyodorov reports in \cite{fyodorov2019spin}. This is not surprising as the models, despite their differences, are related. Using spherical covering, we can invoke this connection and design a secure coding scheme based on Fyodorov's encryption. Nevertheless, as shown in the last section, it seems to be more challenging to achieve Wyner's limit by this scheme. The key challenge in this respect is the codebook design on the hypersphere. Investigation of the same strategy with denser randomly generated codebooks on the hypersphere is another interesting line for future work.

All above conclusions and directions for future work require the extension of the current study to a broader set of models and target problems. The work in this direction is currently ongoing. 

\section*{Acknowledgments}
The first author had the honor to discuss the results of this work and receive helpful feedback from Yan Fyodorov throughout this study. He is hence grateful to Yan Fyodorov for the helpful discussions and useful comments on the work. The first author would like to further thank Lenka Zdeborov\'a, Nicolas Macris, Saba Asaad and Rafael Schaefer for their comments on the work. The financial support by the ETI at FAU Erlangen-N\"urnberg entitled "Secure Coding via Gaussian Random Fields" is also acknowledged with thanks.

\newpage
\appendices

\section{Fyodorov's Nonlinear Encryption}
\label{app:Fyodorov}
In Section~\ref{sec:SecLearn}, we use a particular result of \cite{fyodorov2019spin} to illustrate our contribution to secure learning. We hence state this particular result in this appendix, and represent the core finding of Fyodorov in the shadow of this result: Let the \ac{snr} $\Gamma$ be\footnote{\footnotesize{Fyodorov defines the noise-to-signal ratio in his work and denote it with $\gamma$. In this respect, $\Gamma = \gamma^{-1}$.}}
\begin{align}
	\Gamma = \frac{P_S}{\sigma^2}.
\end{align}
Fyodorov determines the \textit{overlap} $m^\star$ which is defined in \eqref{eq:m_Fyo}. The main result of \cite{fyodorov2019spin} represents the following findings:
\begin{enumerate}
	\item For linear fields, the solution is \ac{rs}. In this case, $m^\star$ increases analytically with the \ac{snr}. Hence, $m^\star\downarrow0$ when $\Gamma \downarrow 0$, and $m^\star=0$ does not occur.
	\item For higher-order fields, the solution is \ac{rsb}. In such cases, $m^\star$ still increases analytically with the \ac{snr}, as long as the field contains a linear term.
	\item For \textit{purely quadratic fields}, $m^\star$ shows a second-order \textit{phase transition}. This means that there exists a \textit{threshold \ac{snr}} $\Gamma_{\rm Th}$ at which $m^\star$ becomes absolutely zero.
\end{enumerate}

For the particular case of purely quadratic fields, Fyodorov's result is given below:
\begin{result}[Full-RSB solution for quadratic fields]
	\label{result:1} 
	Consider the quadratic Gaussian field $\mav\brc{\cdot}$ whose covariance function is given by
	\begin{align}
		\Phi\brc{ x } = \frac{x^2}{2}. \label{eq:quad_phi}
	\end{align}
	Define $\Gamma^\star$ and $\Gamma_{\rm Th}$ as
	\begin{subequations}
		\begin{align}
			\Gamma^\star &= \frac{\mu}{\brc{\mu-1}^2},\\
			\Gamma_{\rm Th} &= \frac{2\mu}{2 - \mu},
		\end{align}
	\end{subequations}
	and let $\xx$ be the solution of 
	\begin{align}
		\mu \; \xx^3+ 3\brc{\frac{1}{2} - \mu } \xx^2 + 3\brc{\mu-1} \xx + \frac{1}{\Gamma} - \frac{1}{\Gamma^\star} = 0,
	\end{align}
	which lies in the interval $\dbc{0,1}$. Then, as $N,K\uparrow \infty$ with a fixed $\mu \leq 1$, we have
	\begin{align}
		m^\star = \begin{cases}
			\sqrt{1-\dfrac{\mu }{ \brc{1 -\mu} \Gamma} } & \Gamma^\star < \Gamma \\
			\sqrt{\mu {\brc{1-\xx}^3 } } & \Gamma_{\rm Th} <\Gamma \leq  \Gamma^\star \\
			0 & \hphantom{\Gamma_{\rm Th} < } \ \Gamma \leq \Gamma_{\rm Th}\\
		\end{cases}.
	\end{align}
\end{result}

\begin{figure}
	\begin{center}
%
%
\definecolor{mycolor1}{rgb}{0.00000,0.44700,0.74100}%
\begin{tikzpicture}

\begin{axis}[%
	width=4.1in,
	height=3in,
	at={(1.011in,0.642in)},
	scale only axis,
	xmin=-11,
	xmax=21,
	xlabel={$\log \Gamma$ in [dB]},
	xtick={-10,0,10,20},
	xticklabels={{$-10$},{$\log \Gamma_{\rm Th}$},{$10$},{$20$}},
	ymin=-.1,
	ymax=1.1,
	ytick={0,1},
	yticklabels={{$0$},{$1$}},
	ylabel={Overlap $m^\star$},
	axis background/.style={fill=white},
	legend style={at={(0.23,.97)},legend cell align=left, align=left, draw=white!15!black}
	]
\addplot [color=black, line width=1.0pt,forget plot]
  table[row sep=crcr]{%
-10	0\\
-9.9	0\\
-9.8	0\\
-9.7	0\\
-9.6	0\\
-9.5	0\\
-9.4	0\\
-9.3	0\\
-9.2	0\\
-9.1	0\\
-9	0\\
-8.9	0\\
-8.8	0\\
-8.7	0\\
-8.6	0\\
-8.5	0\\
-8.4	0\\
-8.3	0\\
-8.2	0\\
-8.1	0\\
-8	0\\
-7.9	0\\
-7.8	0\\
-7.7	0\\
-7.6	0\\
-7.5	0\\
-7.4	0\\
-7.3	0\\
-7.2	0\\
-7.1	0\\
-7	0\\
-6.9	0\\
-6.8	0\\
-6.7	0\\
-6.6	0\\
-6.5	0\\
-6.4	0\\
-6.3	0\\
-6.2	0\\
-6.1	0\\
-6	0\\
-5.9	0\\
-5.8	0\\
-5.7	0\\
-5.6	0\\
-5.5	0\\
-5.4	0\\
-5.3	0\\
-5.2	0\\
-5.1	0\\
-5	0\\
-4.9	0\\
-4.8	0\\
-4.7	0\\
-4.6	0\\
-4.5	0\\
-4.4	0\\
-4.3	0\\
-4.2	0\\
-4.1	0\\
-4	0\\
-3.9	0\\
-3.8	0\\
-3.7	0\\
-3.6	0\\
-3.5	0\\
-3.4	0\\
-3.3	0\\
-3.2	0\\
-3.1	0\\
-3	0\\
-2.9	0\\
-2.8	0\\
-2.7	0\\
-2.6	0\\
-2.5	0\\
-2.4	0\\
-2.3	0\\
-2.2	0\\
-2.1	0\\
-2	0\\
-1.9	0\\
-1.8	0\\
-1.7	0\\
-1.6	0\\
-1.5	0\\
-1.4	0\\
-1.3	0\\
-1.2	0\\
-1.1	0\\
-1	0\\
-0.9	0\\
-0.799999999999999	0\\
-0.699999999999999	0\\
-0.6	0\\
-0.5	0\\
-0.399999999999999	0\\
-0.299999999999999	0\\
-0.199999999999999	0\\
-0.0999999999999996	0\\
0	7.64116718921488e-24\\
0.100000000000001	0.0368723323078728\\
0.200000000000001	0.0626766521139333\\
0.300000000000001	0.0855207494734927\\
0.4	0.106604013214469\\
0.5	0.126438433063333\\
0.600000000000001	0.145304689616108\\
0.700000000000001	0.163378494700828\\
0.800000000000001	0.180779499462677\\
0.9	0.197594097836368\\
1	0.213887436594333\\
1.1	0.229710315709995\\
1.2	0.245103421730988\\
1.3	0.260100061904965\\
1.4	0.27472800470505\\
1.5	0.289010761855218\\
1.6	0.302968507296747\\
1.7	0.316618752192529\\
1.8	0.32997685128706\\
1.9	0.343056389794106\\
2	0.355869483813745\\
2.1	0.368427016970534\\
2.2	0.380738829213149\\
2.3	0.39281386918862\\
2.4	0.404660318503713\\
2.5	0.416285694021807\\
2.6	0.427696932806825\\
2.7	0.438900463217426\\
2.8	0.449902264843853\\
2.9	0.460707919378885\\
3	0.471322654063662\\
3.1	0.481751379007313\\
3.2	0.491998719417434\\
3.3	0.502069043575929\\
3.4	0.511966487236641\\
3.5	0.521694974996851\\
3.6	0.531258239096141\\
3.7	0.540659836017341\\
3.8	0.549903161200997\\
3.9	0.558991462133546\\
4	0.567927850027677\\
4.1	0.57671531027921\\
4.2	0.585356711856724\\
4.3	0.593854815756873\\
4.4	0.602212282639048\\
4.5	0.610431679736832\\
4.6	0.618515487130199\\
4.7	0.626466103450967\\
4.8	0.634285851084364\\
4.9	0.641976980921396\\
5	0.649541676709693\\
5.1	0.656982059044565\\
5.2	0.664300189036871\\
5.3	0.671498071689908\\
5.4	0.678577659013695\\
5.5	0.685540852901795\\
5.6	0.692389507792881\\
5.7	0.699125433136848\\
5.8	0.705750395683037\\
5.9	0.712266121606277\\
6	0.718674298484777\\
6.1	0.724976577142435\\
6.2	0.731174573366833\\
6.3	0.737269869513085\\
6.4	0.743264016002647\\
6.5	0.749158532725365\\
6.6	0.75495491035219\\
6.7	0.760654611565337\\
6.8	0.76625907221199\\
6.9	0.771769702387149\\
7	0.777187887450657\\
7.1	0.782514988983045\\
7.2	0.787752345684417\\
7.3	0.792901274220222\\
7.4	0.797963070017453\\
7.5	0.802939008014528\\
7.6	0.807830343367824\\
7.7	0.812638312117612\\
7.8	0.8173632108876\\
7.9	0.821972007143926\\
8	0.826451064194231\\
8.1	0.830804836419182\\
8.2	0.835037574192179\\
8.3	0.839153336623503\\
8.4	0.843156003270341\\
8.5	0.847049284915258\\
8.6	0.850836733503804\\
8.7	0.854521751321562\\
8.8	0.858107599481975\\
8.9	0.861597405788442\\
9	0.864994172027284\\
9.1	0.868300780742206\\
9.2	0.871520001535546\\
9.3	0.874654496936987\\
9.4	0.877706827876269\\
9.5	0.880679458792816\\
9.6	0.883574762411966\\
9.7	0.886395024214644\\
9.8	0.889142446624735\\
9.9	0.891819152936204\\
10	0.894427190999916\\
10.1	0.896968536688349\\
10.2	0.899445097154747\\
10.3	0.901858713901796\\
10.4	0.904211165673604\\
10.5	0.906504171183592\\
10.6	0.90873939168982\\
10.7	0.910918433428331\\
10.8	0.913042849914212\\
10.9	0.91511414411929\\
11	0.917133770534672\\
11.1	0.919103137125677\\
11.2	0.921023607186137\\
11.3	0.92289650109848\\
11.4	0.924723098005561\\
11.5	0.926504637399712\\
11.6	0.928242320634118\\
11.7	0.929937312361217\\
11.8	0.931590741902516\\
11.9	0.933203704553871\\
12	0.934777262830008\\
12.1	0.936312447651799\\
12.2	0.937810259479565\\
12.3	0.939271669395432\\
12.4	0.940697620137613\\
12.5	0.942089027089229\\
12.6	0.943446779224178\\
12.7	0.944771740012343\\
12.8	0.946064748286313\\
12.9	0.947326619071652\\
13	0.948558144382592\\
13.1	0.949760093984955\\
13.2	0.950933216127964\\
13.3	0.952078238246501\\
13.4	0.953195867635307\\
13.5	0.954286792096489\\
13.6	0.955351680561649\\
13.7	0.956391183689855\\
13.8	0.957405934442613\\
13.9	0.958396548636929\\
14	0.959363625477483\\
14.1	0.960307748068891\\
14.2	0.961229483908961\\
14.3	0.962129385363822\\
14.4	0.96300799012572\\
14.5	0.963865821654282\\
14.6	0.964703389601951\\
14.7	0.965521190224305\\
14.8	0.966319706775911\\
14.9	0.96709940989232\\
15	0.967860757958826\\
15.1	0.968604197466503\\
15.2	0.969330163356098\\
15.3	0.970039079350246\\
15.4	0.970731358274506\\
15.5	0.971407402367673\\
15.6	0.972067603581786\\
15.7	0.972712343872258\\
15.8	0.973341995478512\\
15.9	0.973956921195503\\
16	0.974557474636467\\
16.1	0.975144000487261\\
16.2	0.975716834752589\\
16.3	0.976276304994442\\
16.4	0.976822730563046\\
16.5	0.977356422820576\\
16.6	0.977877685357943\\
16.7	0.97838681420487\\
16.8	0.978884098033531\\
16.9	0.979369818355972\\
17	0.979844249715547\\
17.1	0.980307659872572\\
17.2	0.98076030998442\\
17.3	0.981202454780227\\
17.4	0.981634342730428\\
17.5	0.982056216211283\\
17.6	0.982468311664561\\
17.7	0.982870859752575\\
17.8	0.983264085508694\\
17.9	0.983648208483511\\
18	0.984023442886793\\
18.1	0.984389997725368\\
18.2	0.984748076937079\\
18.3	0.985097879520932\\
18.4	0.985439599663562\\
18.5	0.985773426862149\\
18.6	0.986099546043878\\
18.7	0.986418137682072\\
18.8	0.986729377909096\\
18.9	0.987033438626138\\
19	0.987330487609958\\
19.1	0.987620688616714\\
19.2	0.987904201482943\\
19.3	0.988181182223791\\
19.4	0.988451783128576\\
19.5	0.988716152853771\\
19.6	0.98897443651347\\
19.7	0.989226775767441\\
19.8	0.989473308906805\\
19.9	0.98971417093744\\
20	0.989949493661167\\
};

\addplot [color=black, dotted, forget plot]
table[row sep=crcr]{%
	0	0\\
	0	-1\\
};

\addplot [color=black, dotted, forget plot]
table[row sep=crcr]{%
	-11	1\\
	21	1\\
};

\end{axis}
\end{tikzpicture}%
	\end{center}
	\caption{Overlap against the \ac{snr} for a pure quadratic Gaussian random field with covariance function $\Phi\brc{x} = x^2/2$ and $\mu = 2/3$. The curve is sketched using the result derived by Fyodorov in \cite{fyodorov2019spin}.}
	\label{fig:0}
\end{figure}
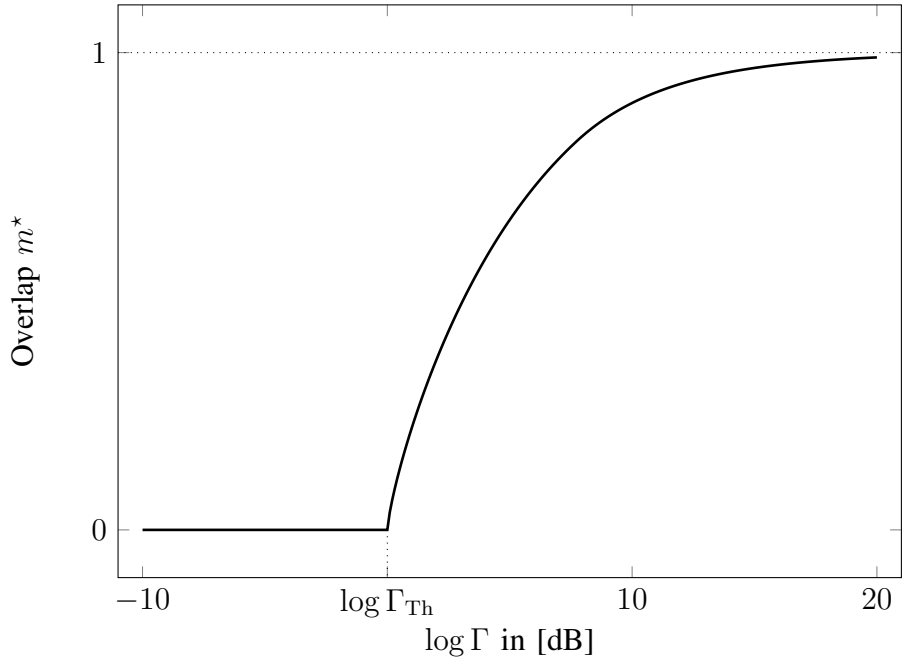

Using Result~\ref{result:1}, we plot $m^\star$ against the \ac{snr} in Fig.~\ref{fig:0}. It is observed that as the \ac{snr} goes below $\Gamma_{\rm Th}$, the overlap $m^\star$ becomes zero. This key observation makes the connection between the work by Fyodorov and Wyner's results on the wiretap setting; nevertheless, this connection was left unaddressed in \cite{fyodorov2019spin}.

\section{Derivations via the Replica Method}
\label{app:2}
We start the derivations by considering the scaled version of free energy, i.e., 
\begin{align}
	-\beta \maf \brc{\beta} = \frac{1}{N} \Ex{ \log \maz_{\beta} \brc{ \byy , \mav} }{}.
\end{align}
Using Riesz's identity \cite{riesz1930valeurs}, we have 
\begin{align}
	\frac{1}{N} \Ex{ \log \maz_{\beta} \brc{ \byy , \mav} }{} = \frac{1}{N}  \lim_{t\downarrow 0} \frac{\partial}{\partial t} \log \Ex{ \dbc{\maz_{\beta} \brc{ \byy , \mav} }^t }{}.
\end{align}

Riesz's identity reduces the logarithmic integration into the derivation of $t$-th moment, i.e., 
\begin{align}
	\maz_{\beta} \brc{t} &\coloneqq \Ex{ \dbc{\maz_{\beta} \brc{ \byy , \mav} }^t }{}.
\end{align}
Nevertheless, it is still a challenging task to determine a real-valued moment. We hence calculate the right hand side for integer $t$ as an analytic term in $t$ and continue it to the real axis by assuming the so-called \textit{replica continuity}. This is a classical practice in the replica analyses; see \cite{krzakala2021statistical,bereyhi2020thesis,merhav2010statistical,mezard1987spin,mezard2009information} for detailed discussions on the replica continuity assumption. The replica continuity assumption leads to
\begin{subequations}
	\begin{align}
	\maz_{\beta} \brc{t} 	&= \Ex{ \dbc{\sum_{\bu\in\set{\pm 1}^K} \exp\set{- \beta \mae\brc{\bu\vert \byy , \mav} } }^t }{},\\
	&= \Ex{ \sum_{\bu_1\in\set{\pm 1}^K} \ldots \sum_{\bu_t\in\set{\pm 1}^K} \prod_{a=1}^t \exp\set{- \beta \mae\brc{\bu_a\vert \byy , \mav} }  }{},\\
	&= \Ex{ \sum_{\set{\bu_a}} \exp\set{- \beta \sum_{a=1}^t  \mae\brc{\bu_a\vert \byy , \mav} }  }{},\\
	&= \Ex{ \sum_{\set{\bu_a}} \exp\set{-  \frac{\beta}{2}  \sum_{a=1}^t \norm{\byy - \mav\brc{\bu_a}}^2 }  }{\byy, \mav},\\
	&= \Ex{ \sum_{\set{\bu_a}} \exp\set{-  \frac{\beta}{2}  \sum_{a=1}^t \norm{\mav\brc{\bs} + \bww - \mav\brc{\bu_a} }^2 }  }{\bs, \bww, \mav},
\end{align}
\end{subequations}
where we define the notation 
	\begin{align}
\set{\bu_a} = \brc{\bu_1, \ldots,\bu_t} \in \prod_{a=1}^t \set{\pm 1}^K.
\end{align}

The sequence $\bs$ is a random sequence uniformly distributed on $\dbc{2^{NR}}$. Using the independency of $\bww$, $\mav$ and $\bs$, the expectation is written as
\begin{subequations}
	\begin{align}
	\maz_{\beta} \brc{t} &= \Ex{ \Ex{ \sum_{\set{\bu_a}} \exp\set{-  \frac{\beta}{2}  \sum_{a=1}^t \norm{\mav\brc{\bs} + \bww - \mav\brc{\bu_a} }^2 }  }{\bs}}{\bww, \mav}\\
	&= \Ex{  \sum_{\set{\bu_a}}  \sum_{ \bu_{t+1}\in\set{\pm 1}^K}  2^{-NR} \exp\set{-  \frac{\beta}{2}  \sum_{a=1}^t \norm{\mav\brc{\bu_{t+1}} + \bww - \mav\brc{\bu_a} }^2 }  }{\bww, \mav}\\
	&= 2^{-NR}  \sum_{\set{\bu_{t+1},\bu_a}} \Ex{ \exp\set{-  \frac{\beta}{2}  \sum_{a=1}^t \sum_{n=1}^{N} \brc{\mav_n\brc{\bu_{t+1}} + \rmw_n - \mav_n\brc{\bu_a} }^2 }  }{\bww, \mav}
\end{align}
\end{subequations}
We now use the fact that outputs of the random field and Gaussian noise entries are \ac{iid}. We hence can factorize the expectation as
\begin{align}
	\maz_{\beta} \brc{t} &= 2^{-NR}  \sum_{\set{\bu_{t+1},\bu_a}} \prod_{n=1}^{N} \Ex{ \exp\set{- \frac{\beta}{2}  \sum_{a=1}^t  \brc{\mav_n\brc{\bu_{t+1}} + \rmw_n - \mav_n\brc{\bu_a} }^2 }  }{\rmw_n, \mav_n}.
\end{align}
We note that the noise process and the field are both Gaussian. Using the standard Gaussian integration, it is shown that \cite{fyodorov2019spin}
\begin{align}
	\Ex{ \exp\set{- \frac{\beta}{2}  \sum_{a=1}^t  \brc{\mav_n\brc{\bu_{t+1}} + \rmw_n - \mav_n\brc{\bu_a} }^2 }  }{\rmw_n, \mav_n} = \frac{1}{\sqrt{\det\set{\mG\brc{\inner{\mU; \mU} }} }}
\end{align}
where $\mG\brc{\cdot} $ is a $t \times t$ matrix whose element on row $a$ and column $b$ is
\begin{align}
	\dbc{\mG\brc{ \inner{\mU; \mU} }}_{a,b} = \delta_{a,b} + \beta\dbc{ \sigma^2 
		+  \Phi \brc{1} \hspace*{-1mm} + \hspace*{-1mm} \Phi \brc{{\inner{\bu_a;\bu_b} }} \hspace*{-1mm} - \hspace*{-1mm} \Phi \brc{{\inner{\bu_a;\bu_{t+1}} }} \hspace*{-1mm} - \hspace*{-1mm} \Phi \brc{{\inner{\bu_b;\bu_{t+1}} }} }.\label{eq:G_U}
\end{align}
Here, we use the notation $\mU \in \set{\pm 1}^{K\times\brc{t+1}}$ to shorten the sequence of \textit{replicas}, i.e., $\bu_1, \ldots,\bu_{t+1}$, as follows
\begin{align}
	\mU = \dbc{\bu_1,\ldots , \bu_{t+1}}.
\end{align}

By substituting, the moment function is finally written as
\begin{align}
	\maz_{\beta} \brc{t} &= 2^{-NR}  \sum_{\mU} \brc{\det\set{\mG\brc{\inner{\mU;\mU} }}}^{-N/2}. \label{eq:Z_11}
\end{align}
We now rewrite the summation over matrices of replicas to a summation over \textit{replica correlation matrices} by dividing the space of all replicas into a set of \textit{sub-shells}: We define the sub-shell $\mas\brc{\mQ}$ for a \textit{correlation matrix} $\mQ \in \dbc{-1,1}^{\brc{t+1} \times \brc{t+1}}$ as
\begin{align}
	\mas\brc{\mQ} = \set{ \mU \in \set{\pm 1}^{K\times\brc{t+1}} : \inner{\mU;\mU} = \mQ}. \label{eq:Subshell}
\end{align}
Note that $\mQ$ is a symmetric matrix with unit diagonal, i.e.,
\begin{subequations}
	\begin{align}
	\dbc{\mQ}_{a,a} &= 1, \\
	\dbc{\mQ}_{a,b} &= \dbc{\mQ}_{b,a}.
\end{align}
\end{subequations}
We now can write
\begin{align}
	\maz_{\beta} \brc{t} &= 2^{-NR}  \int \dif \mQ \exp\set{-N\mai\brc{\mQ} } \brc{\det\set{\mG\brc{\mQ}}}^{-N/2}. \label{eq:Z_1}
\end{align}
where $\exp\set{-N\mai\brc{\mQ} }$ is the sub-shell density, defined as
\begin{align}
	\exp\set{-N\mai\brc{\mQ} } = \sum_{\mU} \prod_{a=1}^{t+1} \prod_{b=a+1}^{t+1} \delta \brc{\inner{\bu_a;\bu_b} - Q_{a,b}}
\end{align}
and
\begin{align}
	\dif \mQ = \prod_{a<b} \dif Q_{a,b}.
\end{align}
The moment is hence derived by finding the sub-shell density and substituting it in \eqref{eq:Z_1}.

\subsection{Deriving the Sub-Shell Density}
To calculate the sub-shell density, we utilize the inverse Laplace transform of the Dirac impulse function which indicates
\begin{align}
	\delta\brc{x} = \int_{\setJ} \frac{\dif L}{2\pi \rmj} \exp\set{Lx}
\end{align}
with $\setJ$ being the $\Re\set{L}=C$ line in the complex plane for some constant $C$. Substituting it into the definition of the sub-shell density, we have
\begin{subequations}
	\begin{align}
	\exp\set{-N\mai\brc{\mQ} } &= \sum_{\mU} \prod_{a<b}  \delta \brc{\inner{\bu_a;\bu_b} -  Q_{a,b}}\\
	&= \sum_{\mU} \prod_{a<b}  \int_{\setJ} \frac{\dif L_{a,b}}{2\pi \rmj} \exp\set{L_{a,b}\brc{\inner{\bu_a;\bu_b} -  Q_{a,b}}}\\
	&= \int \dif \mL \sum_{\mU}  \exp\set{\sum_{a<b} L_{a,b}\brc{\inner{\bu_a;\bu_b} -  Q_{a,b}}} \label{eq:subshell}
\end{align}
\end{subequations}
Here, we define the matrix $\mL\in \setR^{\brc{t+1} \times \brc{t+1}}$ to be a symmetric matrix with unit diagonal elements and off-diagonal elements $L_{a,b} = L_{b,a}$ and
\begin{align}
	\dif \mL = \prod_{a<b} \frac{\dif L_{a,b}}{2\pi \rmj}.
\end{align}
We refer to $\mL$ as the \textit{frequency-domain correlation matrix}.

Following the structure of matrices $\mL$ and $\mQ$, we have
\begin{align}
	\tr{\mL \mQ }= t+1 + 2 \sum_{a<b} L_{a,b} Q_{a,b}
\end{align}
By noting that $\mU^\trp \mU$ is a symmetric matrix with diagonals being all $K$, we can further write
\begin{align}
	\tr{\mL \mU^\trp \mU}  = K\dbc{\brc{t+1} + 2 \sum_{a<b} L_{a,b} \inner{\bu_a;\bu_b}}.
\end{align}
This concludes that
\begin{subequations}
	\begin{align}
	\sum_{a<b} L_{a,b}\brc{\inner{\bu_a;\bu_b} - Q_{a,b}} &= \frac{1}{2} \tr{\mL \mU^\trp \mU} - \frac{K}{2} \tr{\mL \mQ },\\
	&= \frac{1}{2} \tr{ \mU \mL \mU^\trp} - \frac{K}{2} \tr{\mL \mQ }.
\end{align}
\end{subequations}

By substituting into the final expression for the sub-shell density, i.e., \eqref{eq:subshell}, we have
\begin{subequations}
	\begin{align}
	\exp\set{-N\mai\brc{\mQ} } &= \int \dif \mL \sum_{\mU}  \exp\set{\frac{1}{2} \tr{ \mU \mL \mU^\trp} - \frac{K}{2} \tr{\mL \mQ } },\\
	&= \int \dif \mL \exp\set{ -\frac{K}{2} \tr{\mL \mQ } }   \sum_{\mU}  \exp\set{\frac{1}{2} \tr{ \mU \mL \mU^\trp}  }.
\end{align}
\end{subequations}

We now perform a simple variable exchange: Let $\mV = \mU^\trp$ and define $\bvv_k \in \set{\pm 1}^{t+1}$ to be the $k$-th column in $\mV$. We then have
\begin{subequations}
	\begin{align}
	\sum_{\mU}  \exp\set{\frac{1}{2} \tr{ \mU \mL \mU^\trp}  } &= \sum_{\mV}  \exp\set{\frac{1}{2} \tr{ \mV^\trp \mL \mV}  }, \\
	&= \sum_{\mV}  \exp\set{\frac{1}{2} \sum_{ k=1}^K  \bvv_k^\trp \mL \bvv_k}, \\
	&= \sum_{\mV} \prod_{k=1}^K  \exp\set{\frac{1}{2}  \bvv_k^\trp \mL \bvv_k}, \\
	&= \brc{\sum_{\bvv\in \set{\pm 1}^{t+1}}  \exp\set{\frac{1}{2}  \bvv^\trp \mL \bvv} }^K.
\end{align}
\end{subequations}
Defining $\mam \brc{\cdot}$ in terms of the frequency-domain correlation matrix as
\begin{align}
	\mam \brc{\mL} = \log \sum_{\bvv}  \exp\set{\frac{1}{2}  \bvv^\trp \mL \bvv},
\end{align}
we can finally write the summation over $\mU$ as
\begin{align}
	\sum_{\mU}  \exp\set{\frac{1}{2} \tr{ \mU \mL \mU^\trp}  } = \exp\set{ K  \mam \brc{\mL} }.
\end{align}
This concludes that
\begin{align}
	\exp\set{-N\mai\brc{\mQ} } = \int \dif \mL \exp\set{ -K \brc{\frac{1}{2} \tr{\mL \mQ }  -  \mam \brc{\mL} } }.
\end{align}

\subsection{Moment function in the Thermodynamic Limit}
Substituting the sub-shell density into the expression for the moment function $\maz_{\beta} \brc{t}$, we have
\begin{subequations}
	\begin{align}
	\maz_{\beta} \brc{t} &= 2^{-NR}  \int \dif \mQ \dif \mL \exp\set{ -K \brc{\frac{1}{2} \tr{\mL \mQ }  -  \mam \brc{\mL} } } \brc{\det\set{\mG\brc{\mQ}}}^{-N/2},\\
	&= 2^{-NR}  \int \dif \mQ \dif \mL \exp\set{ -K \brc{\frac{1}{2} \tr{\mL \mQ }  -  \mam \brc{\mL} } - \frac{N}{2} \log{ \det\set{\mG\brc{\mQ}} } }, \\
	&= 2^{-NR}  \int \dif \mQ \dif \mL \exp\set{ -N \brc{\frac{1}{2} \log{ \det\set{\mG\brc{\mQ}} } + \frac{R}{2} \tr{\mL \mQ }  - R \mam \brc{\mL} } }. 
\end{align}
\end{subequations}

We now send the spin glass to the thermodynamic limit, i.e., $N\uparrow \infty$. This follows from the slightly non-rigorous assumption of \textit{limit exchange}: We assume that the limits with respect to $N$ and $t$ exchange. Though limit exchange is not necessarily a valid assumption, it is a commonly accepted trick in the replica analysis to assume this particular exchange is valid \cite{mezard2009information}. After exchanging the limit, we deal with the following limit
\begin{align}
\lim_{N\uparrow \infty}	\int \dif \mQ \dif \mL \exp\set{ -N F\brc{\mQ,\mL} },
\end{align}
where the function $F\brc{\mQ,\mL}$ is 
\begin{align}
F\brc{\mQ,\mL} \coloneqq \frac{1}{2} \log{ \det\set{\mG\brc{\mQ}} } + \frac{R}{2} \tr{\mL \mQ }  - R \mam \brc{\mL}. \label{eq:F_QL}
\end{align}
This is a classical limit rise in large-deviations theory \cite{stroock2012introduction,varadhan1984large,guionnet2002large}. Using the saddle-point method and assuming that the exponent function has a unique minimum\footnote{It does not lose the generality of the derivations.}, we have
\begin{align}
\int \dif \mQ \dif \mL \exp\set{ -N F\brc{\mQ,\mL} }
	\doteq  \exp\set{ -N F\brc{\mQ^\star,\mL^\star} },
\end{align}
where $\doteq$ denotes the equivalence in logarithmic scale, and $\brc{\mQ^\star,\mL^\star}$ is the \textit{saddle-point} of the exponent function.

We now find the saddle-point of $F\brc{\mQ,\mL}$: To this end, we first let
\begin{align}
	\frac{\partial }{ \partial \mQ } \set{ {\frac{1}{2} \log{ \det\set{\mG\brc{\mQ}} } + \frac{R}{2} \tr{\mL \mQ }  - R \mam \brc{\mL} } } = \boldsymbol{0}_{t+1},
\end{align}
which concludes that
\begin{align}
	2 R \mL = -	\frac{\partial }{ \partial \mQ } \log{ \det\set{\mG\brc{\mQ}} }.
\end{align}
Using the fact that $\mG\brc{\mQ}$ is symmetric for a symmetric $\mQ$, we can write that
\begin{align}
	\frac{\partial }{ \partial \mQ } \log{ \det\set{\mG\brc{\mQ}} } = 2 \beta \Phi'\brc{\mQ} \odot \tilde{\mG}^{-1}\brc\mQ
\end{align}
where $\odot$ is the Hadamard product and $ \tilde{\mG}^{-1}\brc\mQ$ is a $\brc{t+1} \times \brc{t+1}$ matrix defined in terms of the inverse of $\mG \brc\mQ$ as
\begin{align}
	\tilde{\mG}^{-1}\brc\mQ = \begin{bmatrix}
		\mG^{-1} \brc\mQ + c \mI_t &-\mG^{-1} \brc\mQ \mone_{t\times 1} \\
		-\mone_{1\times t} \mG^{-1} \brc\mQ &\hat{c}
	\end{bmatrix},
\end{align}
for some $c$ and $\hat{c}$. This concludes that
\begin{align}
	\mL^\star = -\frac{\beta}{R} \Phi'\brc{\mQ^\star} \odot \tilde{\mG}^{-1}\brc{\mQ^\star},
\end{align}
where we choose $c$ and $\hat{c}$, such that $L_{a,a}=1$ for $a=1,\ldots,t+1$. We further let
\begin{align}
	\frac{\partial }{ \partial \mL } \set{ {\frac{1}{2} \log{ \det\set{\mG\brc{\mQ}} } + \frac{R}{2} \tr{\mL \mQ }  - R \mam \brc{\mL} } } = \boldsymbol{0}_{t+1}
\end{align}
which concludes that
\begin{align}
	\mQ^\star 
	= \sum_{\bvv} \frac{ \displaystyle \bvv \bvv^\trp \exp\set{\frac{1}{2}  \bvv^\trp \mL^\star \bvv} }{ \displaystyle \sum_{\bvv}  \exp\set{\frac{1}{2}  \bvv^\trp \mL^\star \bvv} }.
	\label{eq:Q_fix}
\end{align}

The fixed-point equation \eqref{eq:Q_fix} indicates that $\mQ^\star$ is in fact the average correlation among the decoupled replicas\footnote{This is why $\mQ^\star$  is often called the replica correlation matrix.}. To see this point clearly, let us define the \textit{decoupled distribution of replicas} as 
\begin{align}
	q_V\brc{\bvv \mL^\star} \coloneqq \frac{ \displaystyle \exp\set{\frac{1}{2}  \bvv^\trp \mL^\star \bvv} }{ \displaystyle \sum_{\bvv}  \exp\set{\frac{1}{2}  \bvv^\trp \mL^\star \bvv} }. \label{eq:q_v}
\end{align}
This is a Boltzmann distribution which is defined for a thermodynamic system with reduced dimension $t+1$. Using this notation, we can write
\begin{align}
	\mQ^\star 
	= \Ex{ \bvv \bvv^\trp }{q_V}
\end{align}
which is the correlation matrix of the vector of replicas $\bvv$.

\subsection{Determining the Free Energy}
As we exchange the limit, we can write
\begin{align}
	- \beta \maf \brc{\beta} =  \lim_{t\downarrow 0} \frac{\partial}{\partial t}  \lim_{N\uparrow \infty} \frac{1}{N} \log \maz_{\beta} \brc{t}.
\end{align}
We now replace the final expression for the moment function, i.e., 
\begin{align}
	\maz_{\beta} \brc{t} &\doteq 2^{-NR}  \exp\set{ -N \min_{\mL, \mQ^\star} \brc{\frac{1}{2} \log{ \det\set{\mG\brc{\mQ^\star}} } + \frac{R}{2} \tr{\mL^\star \mQ^\star }  - R \mam \brc{\mL^\star} } },
\end{align}
and use the asymptotic equivalence in the logarithmic scale to conclude 
\begin{align}
	\beta \maf \brc{\beta} = \lim_{t\downarrow 0} \frac{\partial}{\partial t}  \brc{ R \log 2 +  \frac{1}{2} \log{ \det\set{\mG\brc{\mQ^\star}} } + \frac{R}{2} \tr{\mL \mQ^\star }  - R \mam \brc{\mL^\star} }
\end{align}
for $\mL^\star$ and $\mQ^\star$ which are given by the saddle-point equations.

\subsection{Replica Symmetric Solution: Fixed-Point Equations}
We now assume replica symmetry. This means that the saddle-point matrix $\mQ^\star$ is a symmetric matrix of the following form
\begin{subequations}
	\begin{alignat}{2}
	&\dbc{\mQ^\star}_{a,a} =1  && 1\leq a \leq t+1 ,\\
	&\dbc{\mQ^\star}_{a,b}  = Q \qquad&&  1\leq a < b \leq t,\\
	&\dbc{\mQ^\star}_{a,t+1} =m &&  1\leq a \leq t,
\end{alignat}
\end{subequations}
where $m,Q\in \dbc{0,1}$.  In this case, $\mG\brc{\mQ^\star}$ is of the following form
\begin{align}
	\mG\brc{\mQ^\star} = g\brc{Q} \mI_t + f\brc{Q,m} \mone_t,
\end{align}
where we define
\begin{subequations}
	\begin{align}
	g\brc{Q} &= 1+\beta \dbc{ \Phi \brc{1} - \Phi \brc{Q}},\\
	f\brc{Q,m} &= \beta \dbc{ \sigma^2 + \Phi \brc{1} + \Phi \brc{Q} - 2 \Phi\brc{m} }.
\end{align}
\end{subequations}
As a result, the inverse matrix is given by
\begin{align}
	\mG^{-1}\brc{\mQ} = \frac{1}{g\brc{Q}} \brc{ \mI_t - \frac{f\brc{Q,m}}{ g\brc{Q} + t f\brc{Q,m}  } \mone_t}.
\end{align}

From the saddle-point equations, we conclude that the frequency-domain matrix at the saddle-point exhibits the same symmetry. Namely, its entries are given by 
\begin{subequations}
	\begin{alignat}{2}
		&\dbc{\mL^\star}_{a,a}=1  && 1\leq a \leq t+1,\\
	&\dbc{\mL^\star}_{a,b} = \frac{\beta}{R} \Phi'\brc{Q} \frac{f\brc{Q,m}}{g\brc{Q} \brc{ g\brc{Q} + t f\brc{Q,m} } }  \qquad &&1\leq a < b \leq t,\\
	&\dbc{\mL^\star}_{a,t+1} =\frac{\beta}{R} \Phi'\brc{m} \frac{1}{g\brc{Q} + t f\brc{Q,m}}    &&1\leq a \leq t.
\end{alignat}
\end{subequations}
By sending $t \downarrow 0$, the frequency-domain correlation matrix reduces to an \ac{rs} matrix with unit diagonal and non-diagonal entries 
\begin{align}
 L = \frac{\beta}{R} \Phi'\brc{Q} \frac{f\brc{Q,m}}{g^2\brc{Q}}
\end{align}
for first $t$ columns and
\begin{align}
 E =  \frac{\beta}{R} \Phi'\brc{m} \frac{1}{g\brc{Q}} 
\end{align}
for the last column.

To specify $Q$ and $m$, we derive the second saddle-point equation in terms of $Q$ and $m$. We start by considering the fixed-point equation, i.e., \eqref{eq:Q_fix}. 
From the structure of $\mQ^\star$, we conclude that
\begin{subequations}
	\begin{align}
	t\brc{t-1}Q + t &= \sum_{\bvv} \frac{ \displaystyle \brc{\sum_{a=1}^t \vv_a }^2 \exp\set{\frac{1}{2}  \bvv^\trp \mL \bvv} }{ \displaystyle \sum_{\bvv}  \exp\set{\frac{1}{2}  \bvv^\trp \mL \bvv} }\\
	m t &= \sum_{\bvv} \frac{\displaystyle \brc{\sum_{a=1}^t \vv_a }\vv_{t+1} \exp\set{\frac{1}{2}  \bvv^\trp \mL \bvv} }{\displaystyle \sum_{\bvv}  \exp\set{\frac{1}{2}  \bvv^\trp \mL \bvv} }
\end{align}
\end{subequations}
where the above equations are derived by summing up the $t\times t$ inner block of $\mQ$ and the last column. From the structure on $\mL$, we further have
\begin{align}
	\bvv^\trp \mL \bvv = t \brc{1-L} + 1 + L\brc{\sum_{a=1}^t \vv_a}^2 + 2 E  \brc{\sum_{a=1}^t \vv_a} \vv_{t+1}
\end{align}
This is then straightforward to see that 
\begin{subequations}
	\begin{align}
	\sum_{\bvv} \frac{ \displaystyle \brc{\sum_{a=1}^t \vv_a }^2 \exp\set{\frac{1}{2}  \bvv^\trp \mL \bvv} }{ \displaystyle \sum_{\bvv}  \exp\set{\frac{1}{2}  \bvv^\trp \mL \bvv} } &=  t + 2 \frac{\partial}{\partial L}  \log \sum_{\bvv}  \exp\set{\frac{1}{2}  \bvv^\trp \mL \bvv}, \\
	\sum_{\bvv} \frac{ \displaystyle \brc{\sum_{a=1}^t \vv_a }\vv_{t+1} \exp\set{\frac{1}{2}  \bvv^\trp \mL \bvv} }{ \displaystyle \sum_{\bvv}  \exp\set{\frac{1}{2}  \bvv^\trp \mL \bvv} } &=  \frac{\partial}{\partial E}  \log \sum_{\bvv}  \exp\set{\frac{1}{2}  \bvv^\trp \mL \bvv}.
\end{align}
\end{subequations}
By substituting, we conclude the following two fixed-point equations
\begin{subequations}
	\begin{align}
	t\brc{t-1}Q &=2  \frac{\partial}{\partial L}  \log \sum_{\bvv}  \exp\set{\frac{1}{2}  \bvv^\trp \mL \bvv} \\
	m t &= \frac{\partial}{\partial E}  \log \sum_{\bvv}  \exp\set{\frac{1}{2}  \bvv^\trp \mL \bvv}
\end{align}
\label{eq:Fix_Final_1}
\end{subequations}

We proceed the derivations by calculating the right hand side function. Namely, we write
\begin{subequations}
	\begin{align}
	\mam_{\rm e}\brc{L,E} &\coloneqq   \sum_{\bvv}  \exp\set{\frac{1}{2}  \bvv^\trp \mL \bvv}\\
	&=   \sum_{\bvv}  \exp\set{\frac{t \brc{1-L} + 1}{2} + \frac{L}{2}\brc{\sum_{a=1}^t \vv_a}^2 +  E \vv_{t+1} \sum_{a=1}^t \vv_a}
\end{align}
\end{subequations}
By using Gaussian integration technique, we can further write
\begin{align}
	\exp\set{ \frac{L}{2}\brc{\sum_{a=1}^t \vv_a}^2 } = \int \md z \exp\set{ \sqrt{L} \brc{\sum_{a=1}^t \vv_a } z },
\end{align}
where we define the notation
\begin{align}
	\int \md z  = \int_{-\infty}^{+\infty} \frac{\dif z}{\sqrt{2\pi} } \exp\set{ -\frac{z^2}{2} }.
\end{align}
This concludes that
\begin{subequations}
	\begin{align}
	\mam_{\rm e}\brc{L,E} &=   \sum_{\bvv} \int \md z \exp\set{\frac{t \brc{1-L} + 1}{2} + \sqrt{L} \brc{\sum_{a=1}^t \vv_a } z +  E \vv_{t+1} \sum_{a=1}^t \vv_a}\\
	&=   {
		\exp\set{\frac{t \brc{1-L} + 1}{2}} \int \md z  \sum_{\bvv} \exp\set{ {\sum_{a=1}^t \vv_a } \brc{ \sqrt{L} z +  E \vv_{t+1} } }
	}\\
	&=  {
		\exp\set{\frac{t \brc{1-L} + 1}{2}} \int \md z \sum_{ \vv_{t+1}} \sum_{\bvv \backslash \vv_{t+1}}  \exp\set{ {\sum_{a=1}^t \vv_a } \brc{ \sqrt{L} z +  E \vv_{t+1} } }
	}\\
	&=  {
		\exp\set{\frac{t \brc{1-L} + 1}{2}} \int \md z \sum_{ \vv_{t+1}} \sum_{\bvv \backslash \vv_{t+1}}  \prod_{a=1}^t \exp\set{ \vv_a  \brc{ \sqrt{L} z +  E \vv_{t+1} } }
	}\\
	&=   {
		\exp\set{\frac{t \brc{1-L} + 1}{2}} \int \md z \sum_{ \vv_{t+1}} \brc{\sum_{\vv} \exp\set{ \vv  \brc{ \sqrt{L} z +  E \vv_{t+1} } }}^t
	}\\
	&=  {
		\exp\set{\frac{t \brc{1-L} + 1}{2}} \int \md z \sum_{ \vv_{t+1}} \brc{2 \cosh \brc{  \sqrt{L} z +  E \vv_{t+1} } }^t }
\end{align}
\label{eq:M_deriv}
\end{subequations}

We now define the random variable $S$ to be uniform on $\set{\pm 1}$, and the random variable $Z$ to be zero-mean unit-variance Gaussian. We then have
\begin{subequations}
	\begin{align}
	\log \mam_{\rm e}\brc{L,E} &=  \log \set{ 2 \exp\set{\frac{t \brc{1-L} + 1}{2}} \Ex{ \brc{2 \cosh  \brc{ \sqrt{L} Z +  E S } }^t}{Z,S}}\\
	&=  \log 2 + \frac{t \brc{1-L} + 1}{2}  + \log\Ex{ \brc{2 \cosh \brc{ \sqrt{L} Z +  E S } }^t}{Z,S}\\
	&=  \brc{t+1}  \log 2 + \frac{t \brc{1-L} + 1}{2}  + \log\Ex{  \cosh^t \brc{ \sqrt{L} Z +  E S } }{Z,S}
\end{align}
\end{subequations}
Using the symmetry of $\cosh{\cdot}$, we can write
\begin{align}
	\Ex{  \cosh^t \brc{ \sqrt{L} Z +  E S } }{Z,S} = \Ex{  \cosh^t \brc{ \sqrt{L} Z +  E } }{Z}.
\end{align}
Defining function $W_t\brc{\cdot}$ as
\begin{align}
	W_t\brc{L,E} \coloneqq \log \Ex{  \cosh^t \brc{ \sqrt{L} Z +  E } }{Z},
\end{align}
we finally conclude that
\begin{align}
	\log \mam_{\rm e}\brc{L,E} 
	&=  \brc{t+1}  \log 2 + \frac{t \brc{1-L} + 1}{2}  + W_t\brc{L,E}.
\end{align}

We now calculate the derivatives of $\log \mam_{\rm e}\brc{L,E}$. To this end, we note that
\begin{align}
	\frac{\partial}{\partial E } W_t\brc{L,E} = t \frac{ \Ex{  \cosh^{t-1} \brc{ \sqrt{L} Z +  E } \sinh \brc{ \sqrt{L} Z +  E } }{Z}}{ \Ex{  \cosh^t \brc{ \sqrt{L} Z +  E } }{Z}}.
\end{align}
Moreover, for the derivative with respect to $L$, we have
\begin{align}
	\frac{\partial}{\partial L } W_t\brc{L,E} = \frac{t}{2\sqrt{L}} \frac{ \Ex{  \cosh^{t-1} \brc{ \sqrt{L} Z +  E } \sinh \brc{ \sqrt{L} Z +  E } Z }{Z}}{ \Ex{  \cosh^t \brc{ \sqrt{L} Z +  E } }{Z}}
\end{align}
Using integration by part, one can further show that for $X\sim \man\brc{\eta , \rho^2}$, we have
\begin{align}
\Ex{  \cosh^{t-1} \hspace*{-.5mm}\brc{ X} \sinh \brc{ X } \frac{X-\eta}{\rho^2} }{ } \hspace*{-1mm}= \hspace*{-1mm} \Ex{  \cosh^{t} \brc{ X}  }{ }  \hspace*{-1mm} + \hspace*{-1mm} \brc{t\hspace*{-.5mm}-\hspace*{-.5mm}1} \Ex{  \cosh^{t-2} \brc{X } \sinh^2 \brc{ X }  }{ }.
\end{align}
Hence, we conclude that
\begin{subequations}
	\begin{align}
	\frac{\partial}{\partial L}\log \mam_{\rm e}\brc{L,E} &= \frac{ t \brc{t-1} }{2}  \frac{\Ex{  \cosh^{t-2} \brc{ \sqrt{L} Z +  E } \sinh^2 \brc{ \sqrt{L} Z +  E }  }{Z}}{\Ex{  \cosh^{t} \brc{ \sqrt{L} Z +  E }  }{Z} } \\
	\frac{\partial}{\partial E} \log \mam_{\rm e}\brc{L,E} &=   t \frac{ \Ex{  \cosh^{t-1} \brc{ \sqrt{L} Z +  E } \sinh \brc{ \sqrt{L} Z +  E } }{Z}}{ \Ex{  \cosh^t \brc{ \sqrt{L} Z +  E } }{Z}}
\end{align}
\end{subequations}

By substituting into \eqref{eq:Fix_Final_1}, we finally have
\begin{subequations}
	\begin{align}
	Q&=  \frac{\Ex{  \cosh^{t-2} \brc{ \sqrt{L} Z +  E } \sinh^2 \brc{ \sqrt{L} Z +  E }  }{Z}}{\Ex{  \cosh^{t} \brc{ \sqrt{L} Z +  E }  }{Z} } ,\\
	m &=  \frac{ \Ex{  \cosh^{t-1} \brc{ \sqrt{L} Z +  E } \sinh \brc{ \sqrt{L} Z +  E } }{Z}}{ \Ex{  \cosh^t \brc{ \sqrt{L} Z +  E } }{Z}}.
\end{align}
\end{subequations}
whose limits as $t\downarrow 0$  read
\begin{subequations}
	\begin{align}
	Q&=   \Ex{ \tanh^2 \brc{ \sqrt{L} Z +  E } }{Z}, \\
	m &=  \Ex{ \tanh \brc{ \sqrt{L} Z +  E }  }{Z}.
\end{align}
\end{subequations}
This concludes the fixed-point equations.

\subsection{Replica Symmetric Free Energy}
We now calculate the free energy for the \ac{rs} structure. We start with noting that with \ac{rs} structured $\mQ^\star$, matrix $\mG\brc{\mQ^\star} $ has two eigenvalues; namely, \begin{itemize}
	\item $g\brc{Q}$ with multiplicity $t-1$, and
	\item $g\brc{Q} + t f\brc{Q,m}$ with multiplicity $1$.
\end{itemize}
We therefore have
\begin{align}
	\log \det \mG\brc{\mQ} = \brc{t-1} \log g\brc{Q}  + \log\brc{g\brc{Q} + t f\brc{Q,m}}.
\end{align}
By replacing the \ac{rs} structures, we further have
\begin{align}
	\tr{ \mL \mQ} = \brc{t+1} + t\brc{t-1} Q L  + 2 t E m
\end{align}
and finally using the Riesz identity, we have for $t$ in the vicinity of $t=0$, 
\begin{align}
	\mam\brc{\mL} &= \brc{t+1} \log 2 + \frac{t \brc{1-L} + 1}{2}  + t\Ex{ \log  \cosh \brc{ \sqrt{L} Z +  E }  }{Z}.
\end{align}

We now apply the replica continuity assumption and write
\begin{align}
	\beta \maf \brc{\beta} = \lim_{t\downarrow 0} \frac{\partial}{\partial t}  \brc{ R \log 2 +  \frac{1}{2} \log{ \det\set{\mG\brc{\mQ}} } + \frac{R}{2} \tr{\mL \mQ }  - R \mam \brc{\mL} }
\end{align}
which 
concludes the \ac{rs} solution given in Proposition~\ref{proposition:RS}.

\section{Decoupling Principle}
\label{app:dec}

This appendix aims to give a proof for Proposition~\ref{Prop:2}. To this end, we use the multivariate form of Carleman's theorem which indicates the determinacy of the moment problem with uniformly bounded moments \cite{aheizer1965classical}.

Let us start with the decoupled setting described in Proposition~\ref{Prop:2}. For this setting, the $\brc{\kappa,\tau}$ joint moment of the label $S$ and recovery $\hat{S}$ is given by
\begin{align}
	M_{\kappa,\tau} &= \Ex{S^\kappa \hat{S}^\tau }{} \\
	&=\Ex{S^\kappa \tanh^\tau\brc{\frac{S+\rho Z}{\hat{\rho}^2} } }{}
\end{align}
Using the odd symmetry of the $\tanh\brc{\cdot}$, we can conclude that
\begin{align}
	M_{\kappa,\tau} = \brc{1+\brc{-1}^{\kappa+\tau}} \Ex{ \tanh^\tau\brc{\frac{1+\rho Z}{\hat{\rho}^2} } }{}.
\end{align}

Since $M_{\kappa,\tau} $ is bounded for all non-negative integers $\kappa$ and $\tau$, Carleman's theorem indicate that the mapping from corresponding distribution, i.e., the joint distribution of $S$ and $\hat{S}$, to the joint moments is unique. This means that any other random pairs with same joint moments have the same distributions. Therefore, it is sufficient to show that for any $k\in\dbc{K}$, the joint $\brc{\kappa,\tau}$ moment asymptotically converges to $M_{\kappa,\tau} $.

Consider index $k\in\dbc{K}$. The joint $\brc{\kappa,\tau}$ moment of the label $s_k$ and its corresponding recovery is given by 
\begin{subequations}
	\begin{align}
	\Ex{ s_k^\kappa \hat{s}_k^\tau}{} &= \Ex{ s_k^\kappa \brc{\sum_{\bu}    p_S\brc{\bu\vert \byy,\mav} u_k }^\tau  }{}\\
	&= \Ex{s_k^\kappa \sum_{\set{\bu_a}} \prod_{a=1}^\tau  u_{a,k}    p_S\brc{\bu_a\vert \byy,\mav}   }{}
\end{align}
\end{subequations}
where $p_S\brc{\bu\vert \byy,\mav }$ is the posterior distribution of the labels given by
\begin{align}
	p_S\brc{\bu\vert \byy,\mav }= \frac{\displaystyle  q_Y\brc{\byy\vert \bu,\mav} }{\displaystyle \sum_{\bu} q_Y\brc{\byy\vert \bu,\mav} },
\end{align}
with $q_Y\brc{\byy\vert \bu,\mav} $ denoting the likelihood of $\bu$ determined for the postulated noise variance $\hat{\sigma}^2$.

To determine the joint moment, we define the \textit{replicated} posterior distribution 
\begin{align}
	  p_S^{\brc{\tau}}\brc{\set{\bu_a}\vert \byy,\mav}  \coloneqq  \prod_{a=1}^t  p_S\brc{\bu_a\vert \byy,\mav} \label{eq:BoltzT},
\end{align}
which is a Boltzmann distribution with Hamiltonian
\begin{align}
\mae^{\brc{\tau}}\brc{\set{\bu_a}\vert \byy , \mav} \coloneqq  \sum_{a=1}^\tau  \mae\brc{\bu_a\vert \byy , \mav}
\end{align}
at inverse temperature $\beta = 1/\hat{\sigma}^2$. Using this definition, the joint moment is represented as
\begin{align}
	\Ex{ s_k^\kappa \hat{s}_k^\tau}{} = \Ex{s_k^\kappa \prod_{a=1}^\tau  u_{a,k}   }{}
\end{align}
where $\brc{\bs,\set{\bu_a}}$ are jointly distributed with $2^{-K} p_S^{\brc{\tau}}\brc{\set{\bu_a}\vert \byy,\mav} $ conditioned by $\by$ and $\mav$. 

Noting the fact that the joint distribution is a Boltzmann distribution, we follow the classical averaging trick in statistical mechanics. Namely, we define the function
\begin{align}
f\brc{\bs,\set{\bu_a}} \coloneqq \frac{1}{\abs{\setI_k }}\sum_{ j\in \setI_k} s_j^\kappa  \prod_{a=1}^\tau u_{a,j} 
\end{align}
for some index set $\setI_k \subset \dbc{K}$ which contains $k$ and is of width $\abs{ \setI_k } = 1+ \alpha K$ for some $\alpha \leq 1$. We then construct the modified partition function for the given to Boltzmann distribution as
\begin{align}
	\maz_{\beta,h} \brc{\tau} 	&= \Ex{ \sum_{\set{\bu_a}} \exp\set{- \beta \mae^{\brc{\tau}}\brc{\bu_a\vert \byy , \mav} + h f\brc{\bs,\set{\bu_a}}  }  }{}.
\end{align}
It is then easy to show that at inverse temperature $\beta$, we have
\begin{align}
 \Ex{f\brc{\bs,\set{\bu_a}}}{} =\lim_{h\downarrow 0}  \frac{\partial }{\partial h} \log \maz_{\beta,h} \brc{\tau} .
\end{align}
We now note that the asymptotic joint moment is given at the limit $K\uparrow\infty$ and $\alpha\downarrow 0$. We hence conclude that
\begin{align}
	\Ex{s_k^\kappa \hat{s}_k^\tau }{} &= \lim_{\alpha\downarrow 0}  \lim_{K\uparrow \infty} \lim_{h\downarrow 0}  \frac{\partial }{\partial h} \log \maz_{\beta,h} \brc{\tau} 	
\end{align}
when we set $\beta = 1/\hat{\sigma}^2$.

The modified partition function is closely similar to the replicated partition function $\maz_{\beta}\brc{t}$ determined in Appendix~\ref{app:2}. We hence follow the same lines of derivations to conclude that
\begin{align}
	\maz_{\beta,h} \brc{\tau} 
	&= 2^{-K}  \sum_{\mU\in\set{\pm 1}^{K\times\brc{\tau+1}} }  \exp\set{ h f\brc{\mU } } {\det\set{\mG\brc{\inner{\mU; \mU} }} }^{-N/2}
\end{align}
for $\mG\brc\cdot$ defined in \eqref{eq:G_U}, matrix $\mU = \dbc{\bu_1,\ldots,\bu_\tau, \bu_{\tau+1}}$ and 
\begin{align}
f\brc{\mU } \coloneqq f\brc{\bu_{\tau+1},\set{\bu_a}}.
\end{align}

The current form is a simple modification of the derivation in \eqref{eq:Z_11}. We hence follow the same technique and exchange the summation variable with respect to the sub-shells defined in \eqref{eq:Subshell}. Using the Laplace representation of the Dirac impulse, we finally conclude that
\begin{align}
	\maz_{\beta, h} \brc{\tau} 	&= 2^{-K}  \int \dif \mQ \dif \mL \exp\set{ -N \brc{\frac{1}{2} \log{ \det\set{\mG\brc{\mQ}} } + \frac{R}{2} \tr{\mL \mQ }  - R \mam_h \brc{\mL} } }. 
\end{align}
for $\bv\in\set{\pm 1}^{\tau+1}$, some integral measures $\dif\mQ$ and $\dif \mL$ and
\begin{align}
	\mam_h \brc{\mL} = \Ex{\log \sum_{\bvv}  \exp\set{h I f_0\brc{\bvv} + \frac{1}{2}  \bvv^\trp \mL \bvv} }{I} ,
\end{align}
where function $f_0\brc{\bvv}$ is defined as
\begin{align}
	f_0\brc{\bvv} = \vv_{\tau+1}^\kappa \prod_{a=1}^\tau \vv_a,
\end{align}
and $I \in \set{0,1}$ is a binary random variable with 
\begin{align}
	\Pr\set{I=1} = 1-\Pr\set{I=0} = \frac{\abs{\setI_k}}{K}.
\end{align}

By standard derivation, we have
\begin{align}
 \lim_{h\downarrow 0}  \frac{\partial }{\partial h} \log \maz_{\beta,h} \brc{\tau} 	= \int \dif \mQ \dif \mL \frac{\exp\set{ -N F\brc{\mQ,\mL} } }{\int \dif \mQ \dif \mL \exp\set{ -N F\brc{\mQ,\mL} }}  \sum_{\bvv} f_0\brc{\bvv}  q_V^{\brc{\tau}}\brc{ \bvv} 
\end{align}
where $F\brc{\mQ,\mL}$ is defined as in \eqref{eq:F_QL}, and
\begin{align}
	q^{\brc{\tau}}_V \brc{ \bvv\vert\mL}  = \frac{\displaystyle \exp\set{ \frac{1}{2}  \bvv^\trp \mL \bvv} }{\displaystyle \sum_{\bvv}  \exp\set{ \frac{1}{2}  \bvv^\trp \mL \bvv} }. \label{eq:q_tau}
\end{align}
By standard deviations arguments, we can further determine the asymptotic limit as
\begin{align}
\lim_{K\uparrow\infty}	\lim_{h\downarrow 0}  \frac{\partial }{\partial h} \log \maz_{\beta,h} \brc{\tau} 	=  \sum_{\bvv} f_0\brc{\bvv}  q_V^{\brc{\tau}} \brc{ \bvv\vert \mL^\star} 
\end{align}
where $\brc{\mQ^\star , \mL^\star}$ is the saddle point\footnote{Note that this is the same saddle point as the one derived for the replica solution in Appendix~\ref{app:2}.} of $F\brc{\mQ,\mL}$. As the right-hand side does not depend on $\alpha$, we finally have
\begin{align}
	\Ex{s_k^\kappa \hat{s}_k^\tau }{} &= \sum_{\bvv} f_0\brc{\bvv}  q_V^{\brc{\tau}}\brc{ \bvv\vert \mL^\star} 
\end{align}
for $\beta = 1/\hat{\sigma}^2$.

Comparing \eqref{eq:q_tau} and \eqref{eq:q_v}, it is observed that the joint moment is given by the same replicated distribution. This is a proof for the generic form of the decoupling principle; see \cite{bereyhi2020thesis,bereyhi1}. We however need to show that this decoupled form reduces under \ac{rs} to the same form given by Proposition~\ref{Prop:2}. We hence write the derived joint distribution in terms of the \textit{decoupled distribution of replicas}, i.e., the one given in \eqref{eq:q_v}, and use the replica continuity assumption to conclude that
\begin{align}
	\Ex{s_k^\kappa \hat{s}_k^\tau }{} &= \lim_{t\downarrow 0} \sum_{\bvv} f_0\brc{\bvv}  q_V \brc{ \bv , \tilde{\bv} \vert \mL^\star} 
\end{align}
where $\tilde{\bv} = \dbc{\vv_{\tau+2},\ldots,\vv_{t+1}}$. Here, $t$ denotes the number of replicas which is continued to the real axis via the replica continuity assumption\footnote{Remember that unlike $t$, $\tau$ is a fixed integer.}. We now set the replica correlation matrix to be \ac{rs} and follow the same lines of derivations as those taken in \eqref{eq:M_deriv} to show that
	\begin{align}
	 \Ex{s_k^\kappa \hat{s}_k^\tau }{} &= \lim_{t\downarrow 0}  \Ex{S^\kappa \frac{ \displaystyle  \int \md z  \cosh^{t-\tau} \brc{  \frac{\displaystyle \rho z +  S}{\displaystyle \hat{\rho}^2} } \sinh^{\tau} \brc{  \frac{\displaystyle \rho z +  S}{\displaystyle \hat{\rho}^2} } }{ \displaystyle \Ex{\int \md z  \cosh^{t} \brc{  \frac{\displaystyle \rho z +  S}{ \displaystyle \hat{\rho}^2} }}{S}  }}{S},
\end{align}
for some uniform random variable $S$ on $\set{\pm 1}$. Using symmetry of the hyperbolic functions, we can further simplify the expression as
\begin{subequations}
	\begin{align}
	\Ex{s_k^\kappa \hat{s}_k^\tau }{} &=  \brc{-1}^{\kappa+\tau}  \lim_{t\downarrow 0} {\frac{ \displaystyle\int \md z  \cosh^{t-\tau} \brc{  \frac{\displaystyle \rho z +  1}{ \displaystyle \hat{\rho}^2} } \sinh^{\tau} \brc{  \frac{\displaystyle \rho z +  1}{ \displaystyle \hat{\rho}^2} } }{ \displaystyle {\int \md z  \cosh^{t} \brc{ \frac{\displaystyle \rho z +  1}{ \displaystyle \hat{\rho}^2}  }}  }}\\
	&= \brc{-1}^{\kappa+\tau}  \int \md z  \tanh^{\tau} \brc{  \frac{ \displaystyle \rho z +  1}{ \displaystyle \hat{\rho}^2}  }\\
	&=M_{\kappa,\tau}.
\end{align}
\end{subequations}
The proof is concluded by applying Carleman's theorem.

\section{The Wiretap Channel}
\label{app:wyner}
Consider a Gaussian wiretap channel with two receivers: a legitimate receiver and an eavesdropper. In this channel, the transmitter uses encoder
\begin{align}
	f_N \brc\cdot : \dbc{2^{NR}}\mapsto \setR^N
\end{align}
to encode its message $M\in\dbc{2^{NR}}$ into the sequence $\bxx = \dbc{\xx_1, \ldots , \xx_N}^\trp$, such that
\begin{align}
	\frac{1}{N} \sum_{n=1}^N \abs{\xx_n}^2 \leq P
\end{align}
for some \textit{average transmit power} $P$. It then transmits this sequence over an \ac{awgn} channel using $N$ subsequent transmission time intervals. The legitimate receiver therefore receives sequence $\byy = \dbc{\yy_1, \ldots,\yy_N}^\trp$, where $\yy_n$ reads
\begin{align}
	\yy_n = \xx_n + \xi_n,
\end{align}
with $n\in \dbc{N}$ and $\xi_n \sim \man \brc{0, \sigma^2}$. Using the decoder 
\begin{align}
	g_N \brc\cdot : \setR^N \mapsto \dbc{2^{NR}}
\end{align}
the legitimate receiver decodes the transmitted message as $\hat{M} = g_N\brc{\byy}$.

The eavesdropper overhears the transmitted sequence $\bxx$ over an independent \ac{awgn} channel and receives sequence $\bzz = \dbc{\zz_1, \ldots,\zz_N}^\trp$ with $\zz_n$ being
\begin{align}
	\zz_n = \xx_n + \zeta_n
\end{align}
for $n\in \dbc{N}$ and noise term $\zeta_n \sim \man \brc{0, \rho^2}$.

Noting that each realization of $M$ is encoded via $NR$ bits and transmitted within $N$ time slots, it is concluded that the \textit{transmission rate} in this setting is $R$ bits per channel use.

\subsection{Reliable and Secure Transmission}
From information-theoretic viewpoints, we ensure reliability and security of transmission by imposing the following constraints \cite{liang2009information}:
\begin{itemize}
	\item To have a \textit{reliable} communication link between the transmitter and the legitimate receiver, we require to have
	\begin{align}
		\Pr\set{ M \neq \hat{M} } \leq \epsilon_N,
	\end{align}
	for some $\epsilon_N$ tending to zero as $N$ grows large.
	\item The \textit{secrecy} constraint is imposed under the assumption of \textit{infinite} computational capacity of the eavesdropper. In this regards, we desire that information leakage $L_N$, defined as
	\begin{align}
		L_N = \frac{1}{N} I \brc{ M ; \bzz } ,
	\end{align}
	be bounded by a vanishing upper bound, i.e., 
	\begin{align}
		L_N \leq \delta_N \label{eq:Sec_Const}
	\end{align}
	for some sequence $\delta_N$ tending to zero as $N$ grows large. This constraint guarantees that even with unlimited computational power, the eavesdropper is not capable of recovering the transmitted message from its observation.
\end{itemize}

Based on these two constraints, we define an \textit{achievable secrecy rate} as
\begin{definition}[Achievable secrecy rate]
	A given value for $R$ is called an \textit{achievable secrecy rate}, if there exists a sequence of encoders and decoders indexed by $N$, such that both the constraints are satisfied asymptotically.
\end{definition}
The \textit{secrecy capacity} is then defined as follows:
\begin{definition}[Secrecy capacity]
	The \textit{maximum} achievable secrecy rate, often denoted by $C$, is further called \textit{secrecy capacity} of this channel.
\end{definition}

Wyner in \cite{wyner1975wire} determined the secrecy capacity for a generic discrete-alphabet channel. The result was later extended to the Gaussian wiretap channel in \cite{leung1978gaussian} which indicates
\begin{align}
	C = \frac{1}{2}\dbc{ \log \brc{1+\frac{P}{\sigma^2}} - \log \brc{1+\frac{P}{\rho^2}}}^{+}
\end{align}
where $\dbc{x}^+ \coloneqq \max\set{0,x}$. The result is intuitive: The transmitter can transmit data confidentially to the legitimate receiver, if and only if there is a channel with better conditions between him and the receiver compared to the overhearing link. The secrecy capacity is moreover the difference between the channel capacity terms for the legitimate and overhearing channels.

\bibliography{ref}
\bibliographystyle{IEEEtran}
\end{document}